\documentclass[11pt]{article}
\def\Anonymity{1}%

\usepackage{geometry}
\usepackage[T1]{fontenc}
\usepackage{graphicx}
\usepackage[linesnumbered,boxed,ruled,vlined,algo2e,resetcount,algosection]{algorithm2e}
\usepackage{xspace}
\usepackage[utf8]{inputenc}
\usepackage[noadjust]{cite}
\usepackage{amsmath, amssymb, amsfonts, amsthm,mathtools}
\usepackage{enumerate}
\usepackage[dvipsnames]{xcolor}
\definecolor{newcolor}{hsb}{0.6,1,0.75}
\usepackage[pagebackref=true,colorlinks,linkcolor=magenta,citecolor=blue,bookmarks,bookmarksopen,bookmarksnumbered]{hyperref}
\usepackage{bm}
\usepackage{ebproof}
\usepackage{paralist}
\usepackage{thmtools}
\usepackage{rotating}
\usepackage{mdframed}
\usepackage{multirow}
\usepackage{fancyvrb}
\usepackage{mathrsfs}
\usepackage{indentfirst}
\usepackage{mleftright}
\usepackage[nottoc]{tocbibind}%
\usepackage{complexity}
\usepackage{tablefootnote}
\usepackage{relsize}
\usepackage{exscale}
\usepackage{tikz}
\usetikzlibrary{positioning,arrows.meta}
\usepackage{afterpage}
\usepackage{subcaption}

\usepackage[shortlabels]{enumitem}

\usepackage[capitalize]{cleveref}%

\newtheorem{theorem}{Theorem}[section]
\newtheorem{lemma}[theorem]{Lemma}
\newtheorem{corollary}[theorem]{Corollary}
\newtheorem{claim}[theorem]{Claim}

\newtheorem{observation}[theorem]{Observation}
\newtheorem{fact}[theorem]{Fact}

\newtheorem{conjecture}[theorem]{Conjecture}

\declaretheoremstyle[
spaceabove=\topsep, spacebelow=\topsep,
headfont=\normalfont\bfseries,
notefont=\bfseries, notebraces={}{},
bodyfont=\normalfont\itshape,
postheadspace=0.5em,
name={\ignorespaces},
numbered=no,
headpunct=.]
{mystyle}
\declaretheorem[style=mystyle]{namedtheorem}

\theoremstyle{definition}
\newtheorem{definition}[theorem]{Definition}
\newtheorem{problem}[theorem]{Problem}
\newtheorem{example}[theorem]{Example}

\newtheorem{assumption}[theorem]{Assumption}

\theoremstyle{remark}
\newtheorem{remark}{Remark}

\newenvironment{claimproof}[1][\proofname]
{\proof[#1]}
{\endproof}

\makeatletter
\def\moverlay{\mathpalette\mov@rlay}
\def\mov@rlay#1#2{\leavevmode\vtop{%
		\baselineskip\z@skip \lineskiplimit-\maxdimen
		\ialign{\hfil$\m@th#1##$\hfil\cr#2\crcr}}}
\newcommand{\charfusion}[3][\mathord]{
	#1{\ifx#1\mathop\vphantom{#2}\fi
		\mathpalette\mov@rlay{#2\cr#3}
	}
	\ifx#1\mathop\expandafter\displaylimits\fi}
\makeatother

\renewcommand{\poly}{\mathrm{poly}}

\renewcommand{\polylog}{\mathrm{polylog}}

\newcommand{\eps}{\varepsilon}

\newlang{\MCSP}{MCSP}
\newlang{\MFSP}{MFSP}
\newlang{\MKtP}{MKtP}
\newlang{\MKTP}{MKTP}
\newlang{\itrMCSP}{itrMCSP}
\newlang{\itrMKTP}{itrMKTP}
\newlang{\itrMINKT}{itrMINKT}
\newlang{\MINKT}{MINKT}
\newlang{\MINK}{MINK}
\newlang{\MINcKT}{MINcKT}
\newlang{\CMD}{CMD}
\newlang{\DCMD}{DCMD}
\newlang{\CGL}{CGL}
\newlang{\PARITY}{PARITY}
\renewlang{\Gap}{Gap}
\newlang{\Avoid}{\textsc{Avoid}}
\newlang{\LossyCode}{\textsc{Lossy-Code}}
\newlang{\MissingString}{\textsc{Missing-String}}
\newlang{\SinkOfDAG}{\textsc{Sink-Of-DAG}}
\newlang{\Iter}{\textnormal{\textsc{Iter}}}
\newlang{\Palindromes}{\textsc{Palindromes}}
\newlang{\Sparsification}{\textsc{Sparsification}}
\newlang{\HamEst}{\mathsf{HammingEst}}
\newlang{\HamHit}{\mathsf{HammingHit}}
\newlang{\CktEval}{\textsc{Circuit-Eval}}
\newlang{\Hard}{\textsc{Hard}}
\newlang{\cHard}{\textsc{cHard}}
\newlang{\CAPP}{CAPP}
\newlang{\GapUNSAT}{GapUNSAT}
\newlang{\OV}{OV}
\newlang{\PRIMES}{PRIMES}
\renewlang{\PCP}{PCP}
\newlang{\PCPP}{PCPP}
\newclass{\FMA}{FMA}
\newclass{\Avg}{Avg}
\newclass{\ZPEXP}{ZPEXP}
\newclass{\DLOGTIME}{DLOGTIME}
\newclass{\ALOGTIME}{ALOGTIME}
\newclass{\ATIME}{ATIME}%
\newclass{\SZKA}{SZKA}
\newclass{\Laconic}{Laconic\text{-}}
\newclass{\APEPP}{APEPP}
\newclass{\SAPEPP}{SAPEPP}
\newclass{\TFSigma}{TF\Sigma}
\newclass{\NTIMEGUESS}{NTIMEGUESS}
\newclass{\FZPP}{FZPP}
\newclass{\UEoPL}{UEoPL}
\newclass{\EoPL}{EoPL}
\newclass{\SoPL}{SoPL}
\newclass{\CLS}{CLS}
\newclass{\PWPP}{PWPP}

\newlang{\Formula}{Formula}
\newlang{\THR}{THR}
\newlang{\MAJ}{MAJ}
\newlang{\DOR}{DOR}
\newlang{\ETHR}{ETHR}
\newlang{\Midbit}{Midbit}
\newlang{\LCS}{LCS}
\newlang{\TAUT}{TAUT}

\newcommand{\SearchCNF}{\mathrm{Search}}

\newcommand{\False}{\mathsf{false}}

\newcommand{\PHP}{\mathrm{PHP}}
\newcommand{\EPHP}{\mathrm{EPHP}}

\newcommand{\rPHP}{\mathrm{rPHP}}
\newcommand{\rwPHP}{\mathrm{rwPHP}}
\newcommand{\dwPHP}{\mathrm{dwPHP}}
\newcommand{\Refuter}{\textnormal{\textsc{Refuter}}}
\newcommand{\WrongProof}{\textnormal{\textsc{WrongProof}}}
\newcommand{\Rft}[2]{\mathsf{Ref}(#1 \subseteq #2)}

\newcommand{\Tseitin}{\mathrm{Tseitin}}

\newcommand{\Mono}{\mathsf{mono}}
\newcommand{\CriP}{\mathrm{CriP}}

\newcommand{\Res}{\mathsf{Res}}
\newcommand{\rRes}{\mathsf{rRes}}

\newcommand{\pf}{\mathrm{pf}}
\newcommand{\mistake}{\mathsf{mistake}}

\newcommand{\T}{\mathsf{T}}%
\renewcommand{\S}{\mathsf{S}}%
\newcommand{\PV}{\mathsf{PV}}
\newcommand{\APC}{\mathsf{APC}}

\newcommand{\XOR}{\mathsf{XOR}}

\newcommand{\set}[1]{\ensuremath{\mleft\{#1\mright\}}}

\newcommand{\calA}{\mathcal{A}}
\newcommand{\calB}{\mathcal{B}}
\newcommand{\calC}{\mathcal{C}}
\newcommand{\calD}{\mathcal{D}}

\newcommand{\calF}{\mathcal{F}}

\newcommand{\calO}{\mathcal{O}}
\newcommand{\calP}{\mathcal{P}}
\newcommand{\calQ}{\mathcal{Q}}
\newcommand{\calR}{\mathcal{R}}
\newcommand{\calS}{\mathcal{S}}
\newcommand{\calT}{\mathcal{T}}

\newcommand{\N}{\mathbb{N}}

\newcommand{\Domain}{\mathrm{Domain}}

\newcommand{\Reduce}{f}
\newcommand{\Decode}{g}

\newcommand{\Pudlak}{Pudl{\'{a}}k\xspace}
\newcommand{\Krajicek}{Kraj\'{\i}\v{c}ek\xspace}
\newcommand{\Jerabek}{Je\v{r}\'{a}bek\xspace}
\newcommand{\Muller}{M\"{u}ller\xspace}
\newcommand{\Garlik}{Garl{\'{\i}}k\xspace}

\newcommand{\Log}{\mathsf{Log}}

\newcommand{\dt}{\mathsf{dt}}
\newcommand{\width}{\mathrm{width}}

\newcommand{\FLB}{\mathcal{F}_{\mathsf{LB}}}

\newcommand{\IND}{\mathrm{IND}}%
\newcommand{\MIN}{\mathrm{MIN}}%

\usepackage{aliascnt} %

\begin{document}

\newgeometry{margin=0.8in}

\title{Finding Bugs in Short Proofs: \\The Metamathematics of Resolution Lower Bounds} 

\ifnum\Anonymity=1
\author{
	Jiawei Li\\ \small{UT Austin} \\ \small{\texttt{\href{mailto:davidlee@cs.utexas.edu}{davidlee@cs.utexas.edu}}}
	\and
	Yuhao Li\\ \small{Columbia University} \\ \small{\texttt{\href{mailto:yuhaoli@cs.columbia.edu}{yuhaoli@cs.columbia.edu}}}
	\and
	Hanlin Ren\\ \small{University of Oxford} \\ \small{\texttt{\href{mailto:hanlin.ren@cs.ox.ac.uk}{hanlin.ren@cs.ox.ac.uk}}}
}
\fi

\date{}

\maketitle
\pagenumbering{gobble}

\begin{abstract}
    We study the \emph{refuter} problems for proof complexity lower bounds. Suppose $\varphi$ is a hard tautology that does not admit any length-$s$ proof in some proof system $P$. In the corresponding refuter problem, we are given (query access to) a purported length-$s$ proof $\pi$ in $P$ that claims to have proved $\varphi$, and our goal is to find an invalid derivation step within $\pi$. As suggested by witnessing theorems in bounded arithmetic, the \emph{computational complexity} of these refuter problems is closely tied to the \emph{metamathematics} of the underlying lower bounds.

    We focus on refuter problems corresponding to lower bounds for \emph{resolution}, which is arguably the single most studied system in proof complexity. As a warm-up, we show that many refuter problems for resolution \emph{width} lower bounds are $\PLS$-complete. To capture the complexity of refuter problems for resolution \emph{size} lower bounds, we introduce a new class $\rwPHP(\PLS)$ in decision-tree $\TFNP$, which can be seen as a randomized version of $\PLS$.

    \begin{itemize}
        \item We show that the refuter problems for many resolution size lower bounds can be solved in $\rwPHP(\PLS)$, including the classic lower bound of Haken [TCS, 1985] for the pigeonhole principle. More generally, we identify a common proof technique that we call ``random restriction + width lower bound'', and present strong evidence that resolution lower bounds proved by this technique typically have refuter problems in $\rwPHP(\PLS)$.

        \item We then show that the refuter problem for \emph{any} resolution size lower bound is $\rwPHP(\PLS)$-hard, thereby demonstrating that the $\rwPHP(\PLS)$ upper bound mentioned above is tight. Informally speaking, this means that ``$\rwPHP(\PLS)$-reasoning'' is \emph{necessary} for proving \emph{all} resolution size lower bounds.
    \end{itemize}

    Interpreted in bounded arithmetic, our results show that the theory $\T^1_2(\alpha) + \dwPHP(\PV(\alpha))$ characterizes the ``reasoning power'' required to prove (the ``easiest'') resolution size lower bounds.
    
    As a corollary, we obtain surprisingly efficient proofs of resolution lower bounds. In particular, we show that many resolution size lower bounds can be proved in low-width \emph{random resolution} [\Pudlak--Thapen, CCC'17].
    
\end{abstract}

\newpage

{\small \tableofcontents}

\newpage

\pagenumbering{arabic}

\section{Introduction}
One of the earliest lower bounds in proof complexity was Haken's landmark result \cite{Haken85} that the pigeonhole principle requires exponential-size proofs in the resolution proof system. Since then, proof complexity has become a vibrant research area with substantial progress in establishing lower bounds for various proof systems, as well as the development of a wide range of lower bound techniques. However, despite decades of efforts, proving nontrivial lower bounds for stronger systems, such as Frege and Extended Frege, remains elusive. It is widely believed that proving lower bounds for Extended Frege is ``beyond our current techniques''\footnote{This belief is partly supported by the intuition that proving strong circuit lower bounds (e.g., $\NP\not\subseteq \P/_{\poly}$) seems to be a prerequisite for proving strong proof complexity lower bounds (e.g., for Extended Frege) \cite{Razborov15}. However, formalizing such connections has proven challenging \cite{PichS23,ArtecheKPS24}.}, but what does this even mean? How much, and in which directions, must our techniques expand, to enable us to prove lower bounds for stronger proof systems? These questions call for a study of the \emph{metamathematical} difficulty of proving lower bounds in proof complexity (see, e.g., \cite{PichS19, SanthanamT21}).

Inspired by recent works on the reverse mathematics of \emph{circuit lower bounds} \cite{CJSW21, Korten22, CTW23,ChenLiOliveira24}, we propose investigating the metamathematics of proof complexity lower bound through the computational lens of their \emph{refuter} problem.
To illustrate, consider the following total search problem: suppose we are given a resolution proof $\Pi$ that claims to prove the pigeonhole principle, yet its length is shorter than the lower bound established in \cite{Haken85}. By Haken’s result, $\Pi$ cannot be a valid resolution proof; it must contain an invalid derivation. The goal of the search problem is to locate such an error. We refer to this total search problem as the ``refuter problem''\footnote{This term is adopted from \cite{CTW23}, as will be discussed later.} corresponding to Haken's lower bound:

\begin{mdframed}[hidealllines=true,backgroundcolor=gray!10,skipabove=0.5em,skipbelow=-0.4em,innertopmargin=-0.3em]
	\small
\begin{problem}[Informal]\label{prob: refuter for PHP informal}
    Given (query access to) a subexponential-size resolution proof $\Pi$ that claims to be a proof of the pigeonhole principle, find an invalid derivation in $\Pi$.
\end{problem}
\end{mdframed}

For any proof complexity lower bound of the form ``the tautology $\phi$ requires proof length greater than $s$ in the proof system $P$'', we can define an associated search problem: Given a purported $P$-proof $\Pi$ of $\phi$ with length at most $s$, find an invalid derivation in $\Pi$. With the appropriate formalization (see \autoref{sec: settings}), these refuter problems are \emph{$\NP$ search problems} and are \emph{total} if and only if their underlying lower bounds hold. Therefore, their computational complexity can be studied using the theory of $\TFNP$ \cite{MegiddoP91}. As elaborated in \autoref{sec: more backgrounds}, the complexity of refuter problems reflects the metamathematical difficulty of proving the corresponding lower bounds, thereby providing a purely computational framework for analyzing the difficulty of proving such lower bounds.

In this paper, we initiate the study of refuter problems in proof complexity and take the first step by studying these problems for \emph{resolution} lower bounds. Resolution serves as a natural first step for exploring the metamathematics of proof lower bounds for two main reasons:

        \begin{enumerate}[(i)]
            \item First, resolution is a well-studied proof system, largely due to its fundamental connections to SAT-solving and automated theorem provers \cite{DavisP60,DavisLL62}. \Krajicek even estimates that ``there are perhaps more papers published about proof complexity of resolution than about all remaining proof complexity topics combined'' \cite[Chapter 13]{krajicek_proof_complexity}.
            \item Second, significant progress has already been made in proving lower bounds against resolution \cite{Haken85, Urquhart87, ChvatalS88,beame1996simplified, Ben-SassonW01}, suggesting that investigating the metamathematics of resolution lower bounds is a promising avenue.\label{item: why resolution 2}
        \end{enumerate}

We study several important resolution lower bounds, including those for the pigeonhole principle~\cite{Haken85,beame1996simplified}, Tseitin tautologies~\cite{Urquhart87,Schoning97}, and random CNF formulas~\cite{ChvatalS88}.
In our study, we introduce a new syntactic subclass of decision-tree $\TFNP$, denoted as $\rwPHP(\PLS)$, which can be thought of as a randomized version of $\PLS$, residing slightly above $\PLS$ in the $\TFNP$ hierarchy\footnote{The formal definition and  properties of $\rwPHP(\PLS)$ are presented in \autoref{sec: rwPHP(PLS)}.}. 

At a high level, we show that resolution \emph{width} lower bounds are captured by $\PLS$, while resolution \emph{size} lower bounds are captured by $\rwPHP(\PLS)$.

\begin{theorem}[Main Result; Informal]\label{thm: informal main result}

The refuter problems corresponding to the%
    \begin{enumerate} 
    
        \item resolution {width} lower bounds for the pigeonhole principle and Tseitin tautologies are $\PLS$-complete;
    
        \item resolution {size} lower bounds for the pigeonhole principle (\autoref{prob: refuter for PHP informal}), Tseitin tautologies, and random CNF formulas are $\rwPHP(\PLS)$-complete. 
    \end{enumerate} 

\end{theorem}

Our results are more comprehensive than those stated in \autoref{thm: informal main result}, and we defer the full formal presentation to \autoref{sec: our results}. Here, we highlight a few key insights:

\begin{itemize}
    \item {\bf Characterizing a common proof technique:} All the aforementioned resolution size lower bound proofs share a common strategy, which we call ``random restrictions + width lower bounds''. It is often the case that resolution lower bounds proven using this strategy have refuter problems in $\rwPHP(\PLS)$\footnote{An interesting exception is the general size-width tradeoff by Ben-Sasson and Wigderson \cite{Ben-SassonW01}; see \autoref{sec: what we failed to formalize} for further discussions.}. This implies that ``$\rwPHP(\PLS)$-reasoning'' is sufficient for implementing one of the most commonly employed techniques for proving resolution lower bounds. 

    \item {\bf Minimum reasoning for resolution lower bounds:} Complementing the above, we prove that for \emph{any} family of hard tautologies for resolution, the corresponding refuter problem (for size lower bound) is $\rwPHP(\PLS)$-hard. This establishes that our $\rwPHP(\PLS)$ upper bound is indeed tight. Notably, the hardness proof does not rely on the hard tautology being the pigeonhole principle.
    Consequently, this result carries an intriguing metamathematical implication: ``$\rwPHP(\PLS)$-reasoning'' is necessary for proving \emph{any} resolution lower bound.

    \item {\bf Consequences in bounded arithmetic:} \autoref{thm: informal main result} can also be interpreted as conservativeness results showing that a certain fragment of (relativized) bounded arithmetic, $\calT_{\Res} := \T^1_2(\alpha) + \dwPHP(\PV(\alpha))$, ``captures'' the minimum reasoning required for proving resolution size lower bounds. More precisely, $\calT_{\Res}$ is powerful enough to formalize many resolution lower bounds proved in the literature, including Haken's seminal lower bound for PHP~\cite{Haken85}, while at the same time, it is necessary for proving \emph{any} resolution lower bound. An interesting takeaway of this result is that it is consistent with every theory weaker than $\calT_{\Res}$ that \emph{resolution is p-bounded} (i.e., proves every tautology in polynomial size).

    We find the \emph{existence} of such a theory quite insightful. It is natural to speculate that there is a very powerful theory $\calT_{\sf EF}$ such that strong proof systems like Extended Frege ``appears p-bounded'' to every theory weaker than $\calT_{\sf EF}$. If true, this speculation would provide a strong barrier result in proof complexity.

    \item {\bf Surprisingly efficient proofs for resolution lower bounds:} Finally, translating our results into proof complexity, we exhibit surprisingly efficient proofs of resolution lower bounds. For instance, we show that low-width \emph{random resolution}~\cite{PudlakT19} can prove many resolution size lower bounds (including Haken's lower bound).\footnote{Although random resolution is not a Cook--Reckhow proof system unless $\P = \NP$~\cite{PudlakT19}, it is possible to define a fragment of random resolution that is Cook--Reckhow and that proves the aforementioned resolution size lower bounds. 
    } Previously, Cook and Pitassi~\cite{CookP90} proved Haken's lower bound in the theory ${\sf IPV}^\omega$, which can be interpreted as intuitionistic reasoning using polynomial-time concepts over the purported resolution proof. Our results suggest that $\AC^0$ concepts (and indeed much weaker ones) already suffice to prove the same lower bound.\footnote{Our formalization is different from~\cite{CookP90}. Informally speaking, \cite{CookP90} uses ``\emph{polynomial-time} reasoning'' to handle a \emph{polynomial-size} proof, while we use ``\emph{$\PH$-reasoning}'' to handle an \emph{exponentially long} proof.} %
\end{itemize}

\subsection{More Background}\label{sec: more backgrounds}

We now provide further background on bounded reverse mathematics, refuter problems, and the theory of $\TFNP$ to justify our methodology: the \emph{metamathematics} of the proof complexity lower bounds can --- and indeed \emph{should} --- be understood through the \emph{computational complexity} of their associated refuter problems within $\TFNP$.

\paragraph{Bounded reverse mathematics.} Reverse mathematics explores, for each mathematical theorem of interest, the minimal theory required to prove it. In bounded reverse mathematics~\cite{Cook07,Nguyen-PhD,Cook-Nguyen}, the theories considered come from \emph{bounded arithmetic}, which (roughly speaking) are logical theories formalizing the idea of ``reasoning within a complexity class $\mathsf{C}$''. The link between these logical theories and complexity classes makes bounded arithmetic, and hence bounded reverse mathematics, an effective framework for studying the metamathematics of complexity theory.

Indeed, there has been a long history of studying the (un)provability of lower bounds in the context of bounded arithmetic: In 1989, \Krajicek and \Pudlak investigated the unprovability of proof lower bounds \cite{KrajicekP89}, while Razborov studied the unprovability of circuit lower bounds in 1995 \cite{Razborov_feasible_mathematics_II,razborov1995unprovability}. Notably, many lower bounds for weak circuit classes and proof systems can be formalized in weak theories \cite{Razborov_feasible_mathematics_II, CookP90, MullerP20}, while some strong lower bounds are unprovable within them \cite{KrajicekP89, razborov1995unprovability, Krajicek97, Krajicek11, Pich15, PichS21, LiO23, CLO24b}.%

We take a different perspective from the aforementioned line of work: rather than asking whether lower bounds are provable in certain theories, our goal is to \emph{characterize} the exact reasoning power required to prove these lower bounds. That is, we seek to identify the \emph{minimal} theory $\calT$ that can prove the given lower bound and to establish the minimality of $\calT$ by showing that the axioms used in the proof are indeed \emph{necessary}. The necessity of axioms, i.e., deriving the axiom back from the theorem, is called a \emph{reversal} in reverse mathematics. %

\begin{mdframed}[hidealllines=true,backgroundcolor=gray!10,skipabove=0.3em,skipbelow=-0.4em,innertopmargin=-0.3em]
	\small
\begin{example}
    Recently, Chen, Li, and Oliveira \cite{ChenLiOliveira24} presented several notable reversals related to complexity lower bounds. In their work, they establish that variants of weak pigeonhole principles are \emph{necessary and sufficient} for proving various classical lower bounds. For instance, the fact that one-tape Turing machines require $\Omega(n^2)$ time to recognize palindromes \cite{Maass84} can be proved using the \emph{weak pigeonhole principle}; Moreover, \cite[Theorem 4.9]{ChenLiOliveira24} demonstrates a reversal, proving that this lower bound is, in fact, \emph{equivalent} to the weak pigeonhole principle. The work in \cite{ChenLiOliveira24} serves as one of the main inspirations of this paper.
\end{example}
\end{mdframed}

\paragraph{Refuter problems.} To investigate the metamathematics of a lower bound statement, we first write down the statement in forall-exists form:

\begin{itemize}
    \item Circuit lower bounds: Let $L$ be a hard language and $s$ be a size lower bound for $L$. The lower bound statement expresses that \emph{for every} circuit $C$ of size $s$, \emph{there exists} an input $x$ such that $L(x) \ne C(x)$.
    \item Proof lower bounds: Let $\phi$ be a tautology that is hard for some proof system $P$, and $s$ be a size lower bound for $\phi$. The lower bound statement expresses that \emph{for every} purported $P$-proof $\Pi$ of size $s$, \emph{there exists} an invalid derivation step in $\Pi$.
\end{itemize}
 In general, a statement
\begin{equation}
    \forall x\,\exists y~V(x, y) \label{eq: some forall sigma1b}
\end{equation}
would 
define a total  search problem of finding a valid $y$ given $x$ such that $V(x, y)$ holds; note that the statement is \emph{true} if and only if the search problem is \emph{total}. In our case, is a total search problem \autoref{prob: refuter for PHP informal} precisely because of the Resolution lower bounds against the pigeonhole principle proved by Haken and others~\cite{Haken85,beame1996simplified,Ben-SassonW01}.

This correspondence can be formally justified by the \emph{witnessing theorems} in bounded arithmetic. A witnessing theorem for a theory $\calT$ links it to a syntactic subclass $\mathsf{C}_\calT$ of $\TFNP$, and the theorem states that if \eqref{eq: some forall sigma1b} is provable in $\calT$, then the corresponding (total) search problem lies in the class $\mathsf{C}_\calT$.\footnote{This requires \eqref{eq: some forall sigma1b} to be a ``$\forall \Sigma_1^b$-sentence'', meaning that $|x|$ and $|y|$ are polynomially related and $V(x, y)$ is a deterministic polynomial-time relation.} For instance, Buss's witnessing theorem \cite{Buss85} states that if \eqref{eq: some forall sigma1b} is provable in $\S^1_2$, then the corresponding total search problem can be solved in polynomial time; Buss and \Krajicek \cite{BussKrajicek94} showed that if \eqref{eq: some forall sigma1b} is provable in $\T^1_2$, then the corresponding total search problem is solvable in $\PLS$ (polynomial local search).

Search problems corresponding to circuit lower bounds have already been studied in the literature \cite{GutfreundST07, Pich15, CJSW21, Korten22, CTW23, ChenLiOliveira24} and are termed ``refuter problems'' in \cite{CTW23}. We adopt this terminology and refer to the search problems associated with proof lower bounds as ``refuter problems'' as well.\footnote{In fact, \cite{CTW23} called these problems ``\emph{refutation} problems''. We choose to use ``\emph{refuter} problems'' to avoid confusion with the term ``refutation'' in proof complexity, which usually refers to a proof showing that a formula is unsatisfiable.}

\paragraph{Total search problems in $\NP$.} The above discussion suggests that the metamathematics of lower bounds can be understood through the computational complexity of their refuter problems. Since these problems are total search problems in $\NP$ (as long as the lower bounds are true), it is natural to adopt the methodology of $\TFNP$ while studying their complexity.

What is the ``methodology of $\TFNP$''? Since the seminal work of Megiddo and Papadimitriou \cite{MegiddoP91}, problems in $\TFNP$ have been categorized based on their \emph{proof of totality}. For instance, the class $\PLS$ captures $\NP$ search problems whose totality is provable from the principle ``every DAG has a sink'' \cite{DBLP:journals/jcss/JohnsonPY88}, while the class $\PPAD$ captures problems whose totality is provable from ``every DAG with an unbalanced node has another one'' \cite{Papadimitriou94}. Moreover, \emph{completeness} results play the same role as reversals in bounded reverse mathematics. For example, a pivotal result in this direction is the $\PPAD$-completeness of finding a Nash equilibrium in two-player games~\cite{chen2009settling,daskalakis2009complexity}. This result carries an intriguing metamathematical interpretation: Topological arguments (specifically, Brouwer's fixed point theorem \cite{Brouwer}) or methods akin to it are \emph{unavoidable} for proving the existence of Nash equilibrium~\cite{nash1951non}, which stands in stark contrast to the linear programming duality methods used for zero-sum games~\cite{von2023zero}.

The attentive reader may have already noticed that the above methodology shares a close resemblance to (bounded) reverse mathematics. This similarity can indeed be formally justified by the witnessing theorems mentioned earlier. (Another formal justification is that provability in (universal variants of) bounded arithmetic is equivalent to reducibility in $\TFNP$; see, e.g., \cite[Proposition 3.4]{Muller21}.) While reading this paper, it is useful to remember that all $\TFNP$ results established here can be translated into results in bounded arithmetic and vice versa, conveying the same underlying conceptual message.

\subsection{Our Settings}\label{sec: settings}

Before explaining our results, we first discuss the setting of (decision tree) $\TFNP$ and (relativized) bounded arithmetic in which our results take place. This sub-section is of preliminary nature, but we recommend reading (i.e., not skipping) it before proceeding to our results in~\autoref{sec: our results}. In particular, this sub-section introduces our formalization of lower bounds, which is different from previous works~\cite{CookP90, MullerP20} as the purported resolution proof is represented as an \emph{exponentially-long second-order} object.

We consider $\TFNP$ problems in the \emph{decision tree} model ($\TFNP^\dt$); this model is sometimes called ``type-$2$ $\TFNP$ problems'' \cite{BeameCEIP98} when the decision trees are uniform. In this model, we are given an input $x$ of length $N$ and we think of \emph{decision trees of $\polylog(N)$ depth} as ``efficient''. Each possible solution $o$ can be represented by $\polylog(N)$ bits, and there is an efficient procedure $\phi(x, o)$ that verifies whether $o$ is a valid solution for $x$. (That is, given the purported solution $o$, $\phi(x, o)$ makes only $\polylog(N)$ queries to $x$.) The goal is, of course, to find a solution $o$ such that $\phi(x, o)$ holds.

$\TFNP^\dt$ corresponds to \emph{relativized} bounded arithmetic where a new predicate $\alpha$ is added into the language. The predicate $\alpha$ is intuitively treated as an oracle (or an exponentially-long input). For example, $\PV(\alpha)$ captures reasoning using $\P^\alpha$-concepts, i.e., \emph{uniform and efficient decision trees over $\alpha$}.

\begin{mdframed}[hidealllines=true,backgroundcolor=gray!10,skipabove=0.3em,skipbelow=-0.4em,innertopmargin=0]
	\small
    \begin{remark}[Type-$1$ vs.~Type-$2$ $\TFNP$ Problems]
		In the literature, it is common to define a type-$1$ $\TFNP$ problem in terms of \emph{succinct encodings} of exponentially large objects. For example, a possible definition of a $\PLS$-complete problem is as follows: Given a ``neighborhood'' circuit $C:\{0, 1\}^n \to \{0, 1\}^{\poly(n)}$ and a ``potential function'' circuit $V:\{0, 1\}^n \to \{0, 1\}^{\poly(n)}$ that together encode a DAG on $2^n$ nodes, and also an active node (i.e., a node with non-zero out-degree), find a sink of this graph (i.e., a node with non-zero in-degree and zero out-degree). In contrast, the $\TFNP^\dt$ / type-$2$ $\TFNP$ problems that we consider simply treat $C$ and $V$ as oracles.
	
		Any separation of type-$2$ $\TFNP$ problems implies a separation of type-$1$ $\TFNP$ problems \emph{in a relativized world} \cite{BeameCEIP98}. For example, $\PLS^\dt\not\subseteq\PPA^\dt$ implies an oracle $O$ under which $\PLS^O\not\subseteq\PPA^O$.
	\end{remark}
\end{mdframed}

\subsubsection{Refuter Problems for Resolution Lower Bounds}\label{sec: intro refutation problems in TFNPdt}

This subsection formalizes the refuter problem for resolution lower bounds as a $\TFNP^\dt$ problem. We assume familiarity with the resolution proof system. In resolution, every line is a \emph{clause} (i.e., the disjunction of literals) and the only inference rule is the \emph{resolution rule}:
\[
	\begin{prooftree}
		\hypo{C\lor\ell}
		\hypo{D\lor \overline{\ell}}
		\infer2{C\lor D}
	\end{prooftree},
\]
where $C, D$ are clauses and $\ell$ is a literal. Sometimes, we will also allow the \emph{weakening} rule that replaces a clause with a consequence of it:
\[
	\begin{prooftree}
		\hypo{C}
		\infer1{C\lor D}
	\end{prooftree}.
\]
The \emph{size} of a resolution proof is the number of lines (i.e., clauses) in it. The \emph{width} of a resolution proof is the maximum width of any clause in it, where the \emph{width} of a clause is the number of literals in the clause. Basics about resolution can be found in any textbook on proof complexity, e.g., \cite[Section 5]{krajicek_proof_complexity}.

\paragraph{Size lower bounds for resolution.} Let $F$ be a tautology\footnote{A DNF $D$ is a \emph{tautology} if and only if the corresponding CNF $\lnot D$ is a \emph{contradiction}. A \emph{proof} of $D$ being a tautology is a \emph{refutation} of $\lnot D$ being a contradiction. For convenience, we will use the terms ``tautology/proof'' and ``contradiction/refutation'' interchangeably.} that is \emph{exponentially}-hard for resolution. For example, take $F$ to be the pigeonhole principle which does not have $c^n$-size resolution proofs for some absolute constant $c > 1$ \cite{Haken85}. The refuter problem, which we denote as 
\[\Refuter(s(F\vdash_\Res\bot) \le c^n),\]
is defined as follows. The input $\Pi$ is a purported length-$c^n$ resolution proof of $F$ represented as a list of $c^n$ \emph{nodes}, where each node consists of a clause in the resolution proof and the predecessors of this clause. (For example, if the clause in node $i$ is resolved from the clauses in node $j$ and node $k$, then the predecessor information would contain two integers $(j, k)$.) A \emph{valid solution} would be the index of any node $i\in [c^n]$ whose derivation is illegal: denoting $C_i$ the clause in node $i$, then there do not exist clauses $C, D$ and a literal $\ell$ such that
\[C_i = C\lor D, C_j = C\lor \ell, C_k = D \lor \overline{\ell}.\]

A more formal definition can be found in \autoref{sec: refutation problems}.

By Haken's lower bound mentioned above \cite{Haken85}, every purported resolution proof of length $c^n$ must contain an illegal derivation, thus the above problem is \emph{total}. Let $N := c^n \poly(n)$ denote the bit-length of the input resolution proof, then each node can be described in $\poly(n) \le \polylog(N)$ bits, hence there is an efficient decision tree that verifies whether a node $i$ is illegal and the above refuter problem is indeed in $\TFNP^\dt$.

We can also formalize resolution lower bounds in relativized bounded arithmetic as follows. We add a new symbol $\alpha$ into our language that encodes a length-$c^n$ resolution proof, i.e., for each $i\in[c^n]$, $\alpha(i)$ is the $i$-th bit of the proof. Fixing a hard tautology $F$, let $\mistake_F(n, \alpha, i)$ be a $\PV(\alpha)$ predicate that is true if $i < c^n$ and $\alpha$, interpreted as a length-$c^n$ resolution proof for $F$, makes an invalid derivation in the $i$-th step. Note that this only depends on a constant number of nodes in the proof, and each node is described in $\poly(n)$ bits, hence $\mistake_F(n, \alpha, i)$ is indeed computable in deterministic polynomial time with oracle access to $\alpha$. The $\forall\Sigma_1^b(\alpha)$-sentence\footnote{Roughly speaking, a $\forall\Sigma_1^b$-sentence (resp.~$\forall\Sigma_1^b(\alpha)$-sentence) is a sentence of the form
\[\forall x\,\exists y~\varphi(x, y),\]
where $|x|, |y|$ are polynomially related and $\varphi$ is a polynomial-time relation (resp.~polynomial-time relation with $\alpha$ oracle); these sentences naturally express problems in $\TFNP$ (resp.~$\TFNP^\dt$). The notation $n\in\Log$ means that $n$ is the length of some number, thus allowing one to reason about integers of magnitude $2^{\poly(n)}$ and strings of length $\poly(n)$. In our particular case, it allows the length of the purported proof to be exponential in $n$. These are standard notations in bounded arithmetic.}
\[\forall n\in\Log\,\exists i\le c^n~\mistake_F(n, \alpha, i)\]
expresses the totality of the refuter problem as defined above; the provability of this sentence in relativized bounded arithmetic corresponds to the complexity of the refuter problem in $\TFNP^\dt$.\footnote{\label{footnote: parameters in bounded arithmetic}As a technical detail, we can also allow $\alpha$ to take \emph{parameters} $\vec{z}$ that can be thought of as non-uniformity. That is, for each $i\in [c^n]$, $\alpha(\vec{z}, i)$ is the $i$-th bit of the proof. We consider the sentence
\[\forall n\in\Log, \vec{z}\,\exists i\le c^n~\mistake_F(n, \alpha(\vec{z}, \cdot), i)\]
which expresses that the proof encoded by $\alpha(\vec{z}, \cdot)$ is not a valid length-$c^n$ resolution proof for $F$. The power of many natural principles with and without parameters are very different (see e.g., \cite[Section 4.3]{IlangoLW23}).}

\paragraph{Width lower bounds for resolution.} In this paper, we also study the refuter problems corresponding to \emph{width} lower bounds for resolution. Let $F$ be a tautology without width-$w_F$ resolution proofs, the refuter problem for this width lower bound would be denoted as
\[\Refuter(w(F\vdash_\Res\bot) \le w_F).\]
The formalization of width lower bounds is essentially the same as that of size lower bounds, with the only difference that we now impose that every clause in the input resolution proof contains at most $w_F$ literals. This can be done \emph{syntactically} by only allocating $w_F$ literals to each node.%

\ifnum\Anonymity=0
\begin{mdframed}[hidealllines=true,backgroundcolor=gray!10,skipabove=0.3em,skipbelow=0,innertopmargin=0]
	\small
    \begin{remark}[Further motivations for refutation of width lower bounds]%
        Besides being interesting on their own, the complexity of refuting width lower bounds also serves as a stepping stone for understanding the complexity of refuting \emph{size} lower bounds.
    
        Although we have a fairly good understanding of resolution nowadays, size lower bounds for resolution have been an important open problem in history --- in fact, they are milestone achievements in proof complexity. Haken's lower bound \cite{Haken85} for the pigeonhole principle was a breakthrough at its time. But what is the underlying principle for this breakthrough lower bound? Does it correspond to any classical $\TFNP$ class such as $\PPP$, $\PLS$, $\PPAD$, $\PPA$, or $\CLS$? Towards its answer, it would be beneficial to dig into the \emph{proofs} of the resolution size lower bounds for $\PHP$.

        Haken's original paper \cite{Haken85} employed a ``bottleneck counting'' argument and the proof was quite involved. Beame and Pitassi~\cite{beame1996simplified} later introduced a new, simpler proof that elegantly \emph{reduced the size lower bound to a width lower bound} for (a monotone version of) resolution (see \autoref{subsection: width white-box} for more details). This size-width connection is not unique for $\PHP$. The groundbreaking paper by Ben-Sasson and Wigderson~\cite{Ben-SassonW01} established a generic size-width trade-off for resolution, which had a significant impact on the proof complexity community. Today, studying size-width trade-offs for various proof systems has become standard practice (see e.g.~\cite{clegg1996using,pitassi2012exponential,atserias2019size,sokolov2020semi}).

        Returning to $\PHP$ in the context of resolution, we know that reasoning about size lower bounds can, in some sense, be reduced to reasoning about width lower bounds (we will formalize this very soon!). Thus, understanding the refuter problem for width lower bounds seems like a prerequisite to understanding that for size lower bounds.%
    \end{remark}
\end{mdframed}

\fi

\subsubsection{Retraction Weak Pigeonhole Principles}\label{sec: intro rwPHP(PLS)}
This paper demonstrates that the complexity of refuter problems corresponding to resolution size lower bounds is tightly linked to the new complexity class $\rwPHP(\PLS)$. Therefore, we need to introduce this class before describing our results.

Here, ``$\rwPHP$'' stands for the \emph{retraction weak pigeonhole principle}:
\begin{quote}
	For any two functions $f:[N] \to [2N]$ and $g:[2N] \to [N]$, the function $f\circ g: [2N] \to [2N]$ cannot be the identity function.
\end{quote}
The term ``retraction'', borrowed from category theory \cite{Jerabek-independence}, means that the principle concerns a pair of functions $f, g$ where $g$ is a ``retraction''; the term ``weak'' indicates that the domain of $g$ ($[2N]$) is \emph{much} larger than its range ($[N]$). This principle, along with other variants of weak pigeonhole principles, is widely studied in the context of bounded arithmetic \cite{ParisWW88, Krajicek01b, MacielPW02, Thapen-PhD, Atserias03, Krajicek04a, Jerabek04, Jerabek-independence, ChenLiOliveira24} and total search problems \cite{KKMP21, Korten21, Korten22}; it is sometimes also called the ``witnessing weak pigeonhole principle ($\mathsf{WPHPWIT}$)'' \cite{Jerabek-APC1,ChenLiOliveira24} and ``\LossyCode'' \cite{Korten22}. Clearly, $\rwPHP$ corresponds to a $\TFNP^\dt$ problem: given (query access to) two functions $f: [N] \to [2N]$ and $g:[2N]\to [N]$, find an input $y\in [2N]$ such that $f(g(y)) \ne y$.

Let $\calP$ be a problem in $\TFNP^\dt$, then one can define a class $\rwPHP(\calP)$ capturing the retraction weak pigeonhole principle where, informally speaking, the retraction function $g$ can be computed in $\calP$. In the decision tree model, the inputs of $\rwPHP(\calP)$ consist of:\begin{enumerate}
	\item (the evaluation table of) a function $f:[N] \to [2N]$, and
	\item $2N$ instances of $\calP$, denoted as $\{I_y\}_{y\in [2N]}$, where each valid solution $ans$ of each $I_y$ is marked with an integer $g_{y, ans} \in [N]$.
\end{enumerate}
The goal is to find an integer $y\in [2N]$ along with a solution $ans$ of $I_y$ such that $f(g_{y, ans}) \ne y$. It is not hard to see that if $\calP \in \TFNP^\dt$ then $\rwPHP(\calP) \in \TFNP^\dt$ (\autoref{fact: rwPHP in TFNP}). Furthermore, $\rwPHP(\calP)$ can be solved by a simple \emph{randomized} algorithm given oracle access to any solver of $\calP$.

The class $\rwPHP(\PLS)$ is defined as the problems reducible to $\rwPHP(\calP)$ for a $\PLS$-complete problem $\calP$. It can be shown that $\rwPHP(\PLS)$ does not depend on the exact choice of the $\PLS$-complete problem $\calP$ (\autoref{fact: robustness of rwPHP(P)}).

\paragraph{Witnessing for $\T^1_2 + \dwPHP(\PV)$.} Although $\rwPHP(\PLS)$ seems to be new to the $\TFNP$ community, it already appeared implicitly in the literature of bounded arithmetic. This class captures the $\TFNP$ problems whose totality is provable in $\T^1_2 + \dwPHP(\PV)$. In other words, $\rwPHP(\PLS)$ corresponds to the \emph{witnessing theorem} for $\T^1_2 + \dwPHP(\PV)$ (just like how $\PLS$ corresponds to a witnessing theorem for $\T^1_2$ \cite{BussKrajicek94}). This was noticed in \cite{BussKT14} where they showed every $\forall\Sigma_1^b$-consequence of $\T^1_2 + \dwPHP(\PV)$ \emph{randomly} reduces to $\PLS$; in fact, the same argument implies a deterministic reduction to $\rwPHP(\PLS)$.%

\begin{mdframed}[hidealllines=true,backgroundcolor=gray!10,skipabove=0.3em,skipbelow=-0.4em,innertopmargin=0]
	\small
\begin{remark}[How Strong is $\rwPHP(\PLS)$?]~

    Since $\rwPHP(\PLS)$ can be seen as a randomized version of $\PLS$ (where the guarantee that ``most randomness is good'' is provided by the dual weak pigeonhole principle), its position in the $\TFNP^\dt$ hierarchy is roughly the same as, but slightly higher than $\PLS$. In particular, in the decision tree setting, it follows from the previous separations ($\PLS\not\subseteq \PPP$~\cite{GHMPRT22separation} and $\PLS\not\subseteq \PPA$~\cite{B-OM04}) that $\rwPHP(\PLS)$ is contained in neither $\PPP$ nor $\PPA$. Note that there is already a decision tree separation between $\PLS$ and the $\TFNP^\dt$ problem corresponding to $\rwPHP$ (which follows from a resolution width lower bound for $\rwPHP$ \cite[Proposition 3.4]{PudlakT19}), hence in the decision tree setting, $\rwPHP(\PLS)$ strictly contains $\PLS$.

    We also note that $\T^1_2(\alpha) + \dwPHP(\PV(\alpha))$ is a relatively weak theory in the realm of relativized bounded arithmetic.\footnote{The reader might have encountered claims in the literature that even weaker theories such as $\S^1_2$ or $\APC_1$ are ``strong'', so it might be confusing for a reader unfamiliar with bounded arithmetic that we are claiming $\T^1_2(\alpha) + \dwPHP(\PV(\alpha))$ as a ``weak'' theory. The reason is \emph{relativization}: In our formalization, the purported resolution proof $\alpha$ has \emph{exponential} size, and we are only allowed to reason about objects in $\PH$ (think of $\AC^0$ circuits over $\alpha$). This is much weaker than the setting where the proof $\alpha$ has \emph{polynomial} size and we are allowed to reason about polynomial-time concepts. This is roughly analogous to classifying the circuit class $\AC^0$ (i.e., relativized $\PH$) as ``weak'' and $\P/_{\poly}$ as ``strong''.} This theory is a subtheory of both $\T^2_2(\alpha)$ and \Jerabek's (stronger) fragment for approximate counting $\APC_2(\alpha)$ \cite{Jerabek-APC2}. It is also ``weak'' in the sense that unconditional unprovability results are known: it cannot prove the ordering principle \cite{AtseriasT14} and the pigeonhole principle \cite{PudlakT19}.
\end{remark}
\end{mdframed}

\subsection{Our Results}\label{sec: our results}

Our main results can be categorized into three parts: (1) bounded reverse mathematics ($\TFNP$ characterizations) for (several) resolution width lower bounds; (2) bounded reverse mathematics ($\TFNP$ characterizations) for (several) resolution size lower bounds; and (3) further applications in $\TFNP$ and proof complexity. We will describe the results related to width lower bounds first in \autoref{section: results Bounded Reverse Mathematics for Resolution Width Lower Bounds}, not only because they serve as prerequisites for the results regarding size lower bounds (discussed in \autoref{sec: our results on size lower bounds}), but also because the techniques therein find additional applications in $\TFNP$ and proof complexity (detailed in \autoref{section: our results in applications}). 

\subsubsection{Refuters for Resolution Width Lower Bounds}
\label{section: results Bounded Reverse Mathematics for Resolution Width Lower Bounds}

The main message in this subsection is that the refuter problems corresponding to resolution width lower bounds are complete for the well-studied class $\PLS$, the first syntactic subclass of $\TFNP$ introduced in the literature \cite{DBLP:journals/jcss/JohnsonPY88}.

We begin with the results related to the pigeonhole principle. The attentive reader may notice a subtle issue when formulating the refuter problem of width lower bound: $\PHP_{(n+1)\to n}$ already contains an axiom with width $n$, and the width lower bound for proving it is $n$ as well. Thus, the corresponding width refuter problem becomes trivial. To address this, we instead consider the width refuter problem for a \emph{constant-width analog} of $\PHP_{(n+1)\to n}$, called $\EPHP_{(n+1)\to n}$, which has constant-width axioms and an $n/3$ width lower bound as shown in \cite{Ben-SassonW01}. We characterize the complexity of its corresponding refuter problem:

\begin{namedtheorem}[\autoref{theorem: EPHP PLS complete}]
    $\Refuter(w(\EPHP \vdash_{\Res}\bot)< n/3)$ is $\PLS$-complete.
\end{namedtheorem}

A similar $\PLS$-completeness result also holds for Tseitin formulas (on expander graphs), where $e(G)$ below is the \emph{expansion} parameter of the graph $G$ (\autoref{def: expander graphs}).
\begin{namedtheorem}[\autoref{thm: Tseitin width lower bound refuter}]
    $\Refuter(w(\Tseitin\vdash_\Res\bot) < e(G))$ is $\PLS$-complete. 
\end{namedtheorem}

The techniques used in these results will be further extended to the refuter problems corresponding to black-box $\TFNP$ separations, specifically  $\PLS\not\subseteq\PPP$ and $\PLS\not\subseteq\PPA$, as described in \autoref{thm: TFNP ref PPP PPA PLS} below. 

To tackle \autoref{prob: refuter for PHP informal} though, we have to delve into the proofs of the exponential (size) lower bound. A \emph{monotonized} version of the \emph{width} lower bound plays a crucial role in the simplified proof by Beame and Pitassi~\cite{beame1996simplified}. In particular, they show that any resolution refutation of $\PHP_{(n+1)\to n}$ contains a clause $C$ with ``monotone width'' of at least $2n^2/9$ (see \autoref{subsection: width white-box}). We similarly characterize the complexity of its corresponding refuter problem (where the subscript $\Mono$ denotes the monotone analog of the width refuter problem; the formal definition is provided in \autoref{subsection: width white-box}):

\begin{namedtheorem}[\autoref{theorem: monoPHP PLS complete}]
    $\Refuter(w_{\Mono}(\PHP_{(n+1)\rightarrow n}\vdash_{\Res} \bot)<2n^2/9)$ is $\PLS$-complete.
\end{namedtheorem}
This result serves as a key step toward addressing the size refuter problem for the pigeonhole principle, which will be discussed in the next subsection.

The $\PLS$-hardness parts of all three results above stem from a \emph{unified and simple proof}, detailed in \autoref{lemma: width refuter lower bound}. Conversely, the $\PLS$-membership of these refuter problems is established by carefully analyzing the proofs in \cite{beame1996simplified, Ben-SassonW01} and demonstrating that ``$\PLS$-reasoning'' suffices to prove these lower bounds. (In fact, these proofs can be formalized in the theory $\T^1_2(\alpha)$, and the $\PLS$-membership follows directly from the witnessing theorem in \cite{BussKrajicek94}.)

\paragraph{A non-uniform universal $\PLS$-membership.}
Finally, we establish a \emph{universal} $\PLS$-membership result with respect to \emph{non-uniform} decision tree reductions: for any resolution width lower bound against every unsatisfiable CNF, \emph{as long as the lower bound is correct}, the corresponding refuter problem can be reduced to $\PLS$ under \emph{non-uniform} decision tree reductions. 

\begin{namedtheorem}[\autoref{lemma: width refuter upper bound}]
    Let $\calF$ be any (possibly non-uniform) family of unsatisfiable CNFs with polynomially many clauses, and let $w_0$ be any valid resolution width lower bound for $\calF$. Then there exists a (non-uniform) decision-tree reduction from $\Refuter(w(\calF \vdash_\Res\bot) < w_0)$ to $\PLS$.
\end{namedtheorem}

Both the formulation and proof of this result inherently require non-uniformity for at least two reasons: (1) it is computationally hard to check whether an arbitrarily given CNF is unsatisfiable, and (2) even assuming that the given CNF is unsatisfiable, it is hard to calculate the resolution width lower bound. See \autoref{subsection: width black-box} for further discussion.

\begin{mdframed}[hidealllines=true,backgroundcolor=gray!10,skipabove=0.3em,skipbelow=-0.4em,innertopmargin=0]
	\small
\begin{remark}[Uniform vs.~non-uniform reductions]
    Note that if one only cares about non-uniform reductions, then (the $\PLS$-membership parts of) \autoref{theorem: EPHP PLS complete} and~\autoref{thm: Tseitin width lower bound refuter} are merely special cases of \autoref{lemma: width refuter upper bound}. Nevertheless, we believe that the uniform $\PLS$-membership results in \autoref{theorem: EPHP PLS complete} and~\autoref{thm: Tseitin width lower bound refuter} are informative, as they actually show that the corresponding lower bounds can be formalized in $\T^1_2(\alpha)$; in fact, the \emph{code} of the Turing machine that implementing the uniform reduction to $\PLS$ effectively acts as a \emph{proof} of the width lower bound using a \emph{local search} argument. They are also crucial for the uniform $\rwPHP(\PLS)$-memberships for the size refuter problems.
    However, the decision tree reduction in \autoref{lemma: width refuter upper bound} seems to require $\exp(n)$ bits of non-uniformity, making it \emph{highly} non-uniform.

    On the other hand, the non-uniform reduction in \autoref{lemma: width refuter upper bound} implies an intriguing {proof complexity upper bound}: \emph{Small-width resolution can prove width lower bounds for resolution itself}! (See \autoref{section: our results in applications} for more details.) Uniformity is not required for this application, allowing us to derive more proof complexity upper bounds using \autoref{lemma: width refuter upper bound}: \emph{every} resolution width lower bound \emph{that is correct} can be proved in low-width resolution. (The size lower bound analog of \autoref{lemma: width refuter upper bound} remains  unknown, hence we can only show proof complexity upper bounds for tautologies encoding \emph{specific} resolution size lower bounds.)
\end{remark}
\end{mdframed}

\subsubsection{Refuters for Resolution Size Lower Bounds}\label{sec: our results on size lower bounds}

Our main message in this subsection is that the refuter problems corresponding to many resolution size lower bounds are complete for $\rwPHP(\PLS)$, the $\TFNP$ subclass introduced in \autoref{sec: intro rwPHP(PLS)}. Indeed, the theorems presented in this subsection suggest that $\rwPHP(\PLS)$ captures the complexity of proving \emph{the easiest-to-prove} size lower bounds for resolution. Our workflow is the same as before:
\begin{itemize}
    \item First, we show that for many notable resolution size lower bounds proven in the literature, the corresponding refuter problems reduce to $\rwPHP(\PLS)$. Specifically, we identify a common technique for proving resolution size lower bounds, which we call ``random restriction + width lower bounds'', and demonstrate that if a resolution size lower bound can be proven using it, then the corresponding refuter problem generally falls within $\rwPHP(\PLS)$.
    \item Next, we present a \emph{unified} $\rwPHP(\PLS)$-hardness result: the refuter problems for resolution size lower bounds are $\rwPHP(\PLS)$-hard, and the hardness proof \emph{does not} depend on the hard tautology considered. Thus, we conclude the $\rwPHP(\PLS)$-completeness of many refuter problems for resolution size lower bounds.
\end{itemize}
The $\rwPHP(\PLS)$-hardness of size lower bound refuters turns out to be more challenging than the $\PLS$-hardness of width lower bound refuters, as discussed in \autoref{sec: overview of lower bounds}.

We begin by showing that \autoref{prob: refuter for PHP informal} reduces to $\rwPHP(\PLS)$:
\begin{theorem}[Informal version of \autoref{thm: refuting Res lb for PHP is in rwPHP(PLS)}]\label{thm: informal main upper bound}
	There exists an absolute constant $c > 1$ and an efficient decision-tree reduction from the problem $\Refuter(s(\PHP_{(n+1)\to n}\vdash_\Res\bot) \le c^n)$ to $\rwPHP(\PLS)$.
\end{theorem}

In fact, we show that $\T^1_2(\alpha) + \dwPHP(\PV(\alpha))$ proves the sentence 
\[\forall n \in \Log~\exists i\le c^n~\mistake_\PHP(n, \alpha, i),\]
i.e., $\alpha$ is not a length-$c^n$ resolution proof for PHP, by formalizing the classical proofs in \cite{Haken85,CookP90,beame1996simplified}; \autoref{thm: refuting Res lb for PHP is in rwPHP(PLS)} then follows from the witnessing theorem for $\T^1_2(\alpha) + \dwPHP(\PV(\alpha))$. In the technical overview (\autoref{sec: overview of upper bounds}) and the main proof (\autoref{sec: upper bound for Refuter for Resolution}), we present the reduction from the refuter problem $\Refuter(s(\PHP\vdash_\Res\bot) \le c^n)$ to $\rwPHP(\PLS)$ directly, without relying on witnessing theorems. %

It turns out that a large variety of resolution size lower bounds can be proven using the paradigm of ``random restriction + width lower bounds,'' including those for $\XOR$-lifted formulas \cite{DantchevR03}, Tseitin formulas \cite{Urquhart87, Schoning97}, and random CNFs \cite{ChvatalS88}. We show that all these lower bounds have corresponding refuter problems in $\rwPHP(\PLS)$ (see \autoref{thm: XOR-lifting is in rwPHP(PLS)}, \autoref{thm: Res size lower bound for Tseitin}, and \autoref{thm: refuters for random k-CNF lower bounds}, respectively). These results provide strong evidence that $\rwPHP(\PLS)$ (or $\T^1_2(\alpha) + \dwPHP(\PV(\alpha))$) captures the ``complexity'' of this popular proof technique for resolution lower bounds.

We complement the above results by showing that for \emph{every} unsatisfiable family of CNFs $\{F_n\}$ that requires resolution size greater than $s_F(n)$, the corresponding refuter problem $\Refuter(s(F_n\vdash_\Res\bot) \le s_F(n))$ is hard for $\rwPHP(\PLS)$.

\begin{theorem}[Informal version of \autoref{thm: rwPHP(PLS)-hardness of refuters}]\label{thm: informal main lower bound}
	For every unsatisfiable family of CNF formulas $\{F_n\}$ and parameter $s_F(n)$ such that every resolution refutation of $F_n$ requires more than $s_F(n)$ clauses, there exists a decision tree reduction of depth $\poly(n)$ from $\rwPHP(\PLS)$ to $\Refuter(s(F_n\vdash_\Res\bot) \le s_F(n))$.\footnote{This theorem requires a mild technical condition that $s_F(n)$ should be moderately larger than the size of the $\rwPHP(\PLS)$ instance; see the formal statement in \autoref{thm: rwPHP(PLS)-hardness of refuters} for details.}
\end{theorem}

Note that \autoref{thm: informal main lower bound} holds for \emph{every} hard tautology, whereas the $\rwPHP(\PLS)$ upper bounds such as \autoref{thm: informal main upper bound} are only known to hold for some natural families of hard tautologies. For these natural tautologies, we establish a \emph{reversal} in the bounded reverse mathematics of proof complexity lower bounds: The power of ``$\rwPHP(\PLS)$-reasoning'' is \emph{sufficient} for implementing a popular proof strategy that can prove all these resolution lower bounds and, at the same time, is \emph{necessary} for proving \emph{any} resolution lower bound.

\begin{mdframed}[hidealllines=true,backgroundcolor=gray!10,skipabove=0.6em,skipbelow=-0.4em,innertopmargin=0]
	\small
    \begin{remark}
        We also note that \autoref{thm: informal main lower bound} requires decision tree depth $\poly(n)$ regardless of $s_F$, and is thus only considered ``efficient'' when $s_F = 2^{n^{\Omega(1)}}$. However, this is merely an artifact of our definition of ``efficiency'' in the decision tree setting, i.e., if the input length is $N$, then depth-$\polylog(N)$ decision trees are considered ``efficient''. In fact, even if $s_F = 2^{n^{o(1)}}$, each node in the purported length-$s_F$ resolution proof still requires $\poly(n)$ bits to represent, so it takes $\poly(n)$ query complexity to \emph{verify} a solution of the refuter problem. Therefore, it still makes sense \emph{in the particular setting of refuter problems} to consider a decision tree reduction \emph{efficient} if its query complexity is at most $\poly(n)$. We interpret \autoref{thm: informal main lower bound} to mean that ``$\rwPHP(\PLS)$-reasoning'' is necessary for proving \emph{not only subexponential but any moderately large} size lower bound for resolution.
    \end{remark}
\end{mdframed}

The proof of \autoref{thm: informal main lower bound} is heavily inspired by the $\NP$-hardness of automating resolution \cite{AtseriasM19} and the exposition of this result in \cite{RezendeGNPR021}. In these proofs, it was crucial to show that resolution cannot prove lower bounds against itself; in particular, \cite[Section 5]{RezendeGNPR021} showed that resolution requires a large (block-)width to prove resolution lower bounds. Notably, the proof in \cite{RezendeGNPR021} is by a reduction from $\rwPHP$, i.e., resolution cannot prove lower bounds against itself because resolution cannot prove $\rwPHP$. We strengthen these results by reducing a stronger problem ---$\rwPHP(\PLS)$ instead of $\rwPHP$ --- to the refuter problems, thereby obtaining a \emph{tight} characterization of these refuter problems.

Finally, our results provide an intriguing characterization of the provably total $\NP$ search problems in $\T^1_2 + \dwPHP(\PV)$ (see \autoref{cor: main reversal result}). That is:

\begin{quote}
	Just as ``every DAG has a sink'' characterizes the $\forall\Sigma_1^b$-consequences of $\T^1_2$ \cite{BussKrajicek94}, ``resolution requires $2^{\Omega(n)}$ size to prove PHP'' characterizes the $\forall\Sigma_1^b$-consequences of $\T^1_2 + \dwPHP(\PV)$.
\end{quote}

\subsubsection{Applications}
\label{section: our results in applications}

Besides being interesting in itself, our study of refuter problems also reveals several new insights into these well-studied proof complexity lower bounds and $\TFNP$ separations.
More specifically, we translate our results into different languages using the generic connection between $\TFNP^\dt$ and proof complexity via the \emph{false clause search} problem (see, e.g., \cite{RezendeGR22survey}): For an unsatisfiable CNF $F = C_1 \land \dots \land C_M$, the false clause search problem $\SearchCNF(F)$ is a $\TFNP^\dt$ problem where, given oracle access to an input $x \in \{0, 1\}^N$, the goal is to find a clause $C_i$ such that $C_i(x) = \False$. Any $\TFNP^\dt$ problem can be written as a false clause search problem for a family of low-width CNFs, and vice versa.
In particular, a family of unsatisfiable CNFs has low-width resolution refutations if and only if the corresponding false clause search problem reduces to $\PLS$ \cite{razborov1995unprovability}
(see also \cite[Section 8.2.2]{Pritish_Kamath_PhD} for an exposition). See \autoref{fig: application} for a diagram that summarizes the translations of our main results in different languages.

\begin{figure}[btph]
    \centering
    \begin{tikzpicture}[scale=1.1]
\tikzset{inner sep=0,outer sep=3}

\tikzstyle{a}=[inner sep=6pt, inner ysep=6pt,outer sep=0.5pt,
draw=black!20!white, fill=Cerulean!2!white, very thick, rounded corners=6pt, align=left]

\tikzstyle{b}=[inner sep=6pt, inner ysep=6pt,outer sep=0.5pt,
draw=white, fill=white, very thick, rounded corners=6pt, align=center]

\begin{scope}[yscale=1.145]
{
\large

\node[a, anchor = west] (PP) at (-8.1,5) {\autoref{thm: low width res prove res LB}, \ref{thm: low width rRes prove size lb}: \textcolor{blue}{\emph{Proof complexity}} of the \textcolor{blue}{\emph{proof lower bounds}}.};

\node[a, anchor = west] (CP) at (-9,3) {\autoref{lemma: width refuter upper bound}, \ref{thm: Tseitin width lower bound refuter}, \ref{thm: refuters for random k-CNF lower bounds}: \textcolor{orange}{\emph{Computational complexity}} of the \textcolor{red}{\textit{refuters}} for \textcolor{blue}{\emph{proof lower bounds}}.};

\node[a, anchor = west] (CTFNP) at (-7.59,1) {\autoref{thm: TFNP ref PPP PPA PLS}: \textcolor{orange}{\emph{Computational complexity}} of the \textcolor{red}{\textit{refuters}} for \textcolor{orange}{\emph{$\TFNP^{\dt}$ separations}}.};

\draw[<->, thick, black!70!black] 
    ([xshift=4.8cm, yshift=-0.5cm]PP.west) -- 
    ([xshift=4.8cm, yshift=-1.5cm]PP.west);

\draw[<->, thick, black!70!black] 
    ([xshift=8.8cm, yshift=-0.5cm]PP.west) -- 
    ([xshift=9.4cm, yshift=-1.5cm]PP.west);

\draw[->, thick, black!70!black] 
    ([xshift=2.8cm, yshift=-0.5cm]CP.west) -- 
    ([xshift=2.8cm, yshift=-1.5cm]CP.west);

\draw[<->, thick, black!70!black] 
    ([xshift=13.2cm, yshift=-0.5cm]CP.west) -- 
    ([xshift=13.2cm, yshift=-1.5cm]CP.west);

}
\node[b, anchor = west] (PLSRes) at (-7.7, 4) {$\PLS$ $=$ $\mathsf{Res}$ (\autoref{thm: pls characterize resolution})};

\node[b, anchor = west] (FCS) at (1.2, 4) {via \emph{false clause search}};

\node[b, anchor = west] (FCS) at (-6, 2) {via \autoref{lem: res_ref to pls_ref}};

\node[b, anchor = west] (FCS) at (-0.2, 2) {$\PLS$ $=$ $\mathsf{Res}$ (\autoref{thm: pls characterize resolution})};

\end{scope}

\end{tikzpicture}
    \caption{Translations of the main results in different languages. A one-way arrow represents an implication, and a two-way arrow indicates an equivalence.}
    \label{fig: application}
\end{figure}
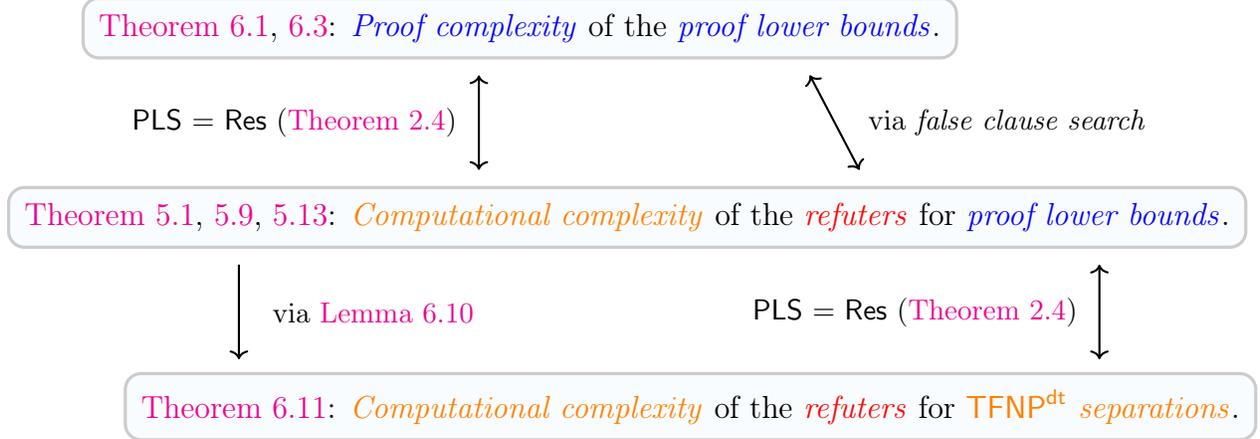

\paragraph{Proof complexity of proof lower bounds.} We first use our results to provide surprisingly efficient proofs for proof complexity lower bounds. Note that a proof complexity lower bound can be expressed by a family of CNFs $\FLB$ by formulating the corresponding refuter problem as a false clause search problem $\SearchCNF(\FLB)$ (see, e.g., \autoref{sec: proof complexity of proof complexity lower bounds}).

In particular, there exists a family of $\tilde{O}(w)$-width CNFs that encodes a width-$w$ resolution lower bound. Then, since $\PLS$ and low-width resolution are equivalent, \autoref{thm: black-box model PLS-completeness for width} implies the following \emph{upper bound} on the resolution width required to prove resolution lower bounds.

\begin{theorem}[Informal version of \autoref{thm: low width res prove res LB}]\label{thm: low width res proves res LB informal}
    Any width-$w$ resolution lower bound can be proved in resolution width $\tilde{O}(w)$.
\end{theorem}

We also use our $\rwPHP(\PLS)$ upper bounds to show that $\poly(n)$-width \emph{random resolution}~\cite{BussKT14, PudlakT19} can prove exponential-size resolution lower bounds (encoded as $\poly(n)$-width CNFs). In fact, using our results on random $k$-CNFs (\autoref{thm: refuters for random k-CNF lower bounds}), we can show that \emph{most} resolution size lower bounds are provable in low-width random resolution:

\begin{theorem}[Informal version of \autoref{thm: low width rRes prove size lb}]\label{thm: low width rRes prove size lb informal}
    With high probability over a random $k$-CNF $F$, the resolution size lower bound $s(F\vdash_\Res\bot) > 2^{\Omega(n)}$ can be proved in \emph{random resolution} width of $\poly(n)$.
\end{theorem}

These results stand in stark contrast with \Garlik's result \cite{Garlik19} that tautologies encoding any resolution size lower bounds are hard for resolution: We show that either switching to width lower bounds (\autoref{thm: low width res proves res LB informal}) or considering \emph{random} resolution (\autoref{thm: low width rRes prove size lb informal}) makes these lower bound tautologies easy to prove!\footnote{In our \autoref{sec: proof complexity of proof complexity lower bounds}, the lower bound tautologies use \emph{binary encoding}, where (e.g.)~the predecessors of every node are encoded by $O(\log N)$ bits. In contrast, \Garlik \cite{Garlik19} uses \emph{unary encoding} where for every pair of nodes $(i, j)$ (a minor detail is that \cite{Garlik19} requires $i$ to be ``one level above'' $j$), there is a Boolean variable $x_{i, j}$ indicating whether $i$ is a predecessor of $j$. As \cite{Garlik19} pointed out, a resolution lower bound for the unary-encoded refutation statements implies a similar lower bound for the binary-encoded refutation statements. On the other hand, since we are proving \emph{width} upper bounds and the unary encoding already results in large-width CNFs, we can only afford to use binary encoding (see \autoref{remark: formalization of resolution lower bounds}).}

\paragraph{Complexity of refuting black-box $\TFNP$ separations.}

    We also consider the refuter problem for \emph{black-box $\TFNP$ separations}. Let $\mathsf{A}, \mathsf{B}$ be two $\TFNP^\dt$ classes such that $\mathsf{A} \nsubseteq \mathsf{B}$. Informally, $\Rft{\mathsf{A}}{\mathsf{B}}$ is the class of problems reducible to the following kind of ``refuter'' problems:
    The input is a purported decision tree reduction from $\mathsf{A}$ to $\mathsf{B}$, and the solution is a short witness showing that the reduction is wrong. The refuter problems for $\TFNP^\dt$ separations also lie in $\TFNP^\dt$, as their totality follows from the correctness of the black-box separation $\mathsf{A} \nsubseteq \mathsf{B}$. 
    
    The complexity of such refuter problems measures the strength of the arguments used for black-box separation results. For example, the following corollary conveys a simple but often overlooked fact: \emph{when separating a syntactic $\TFNP^\dt$ subclass $\mathsf{A}$ from $\mathsf{B}$, it is necessary to incur the totality principle of $\mathsf{B}$.}

    \begin{namedtheorem}[\autoref{cor: tfnp_ref: general_lb}]
        For any two $\TFNP^\dt$ classes $\mathsf{A}, \mathsf{B}$ such that $\mathsf{A} \nsubseteq \mathsf{B}$, $\mathsf{B} \subseteq \Rft{\mathsf{A}}{\mathsf{B}}.$
    \end{namedtheorem}

    Due to the connection between $\TFNP^\dt$ and proof complexity, the refuter problem for each black-box $\TFNP$ separation naturally aligns with a corresponding refuter problem for a proof complexity lower bound.
    In particular, we build a uniform reduction from the refuter problems for separations from $\PLS$ to the refuter problems for resolution width lower bounds (\autoref{lem: res_ref to pls_ref}), because showing a $\TFNP^\dt$ subclass $\mathsf{A}$ is not in $\PLS$ is essentially showing a resolution width lower bound for the formula expressing the totality of $\mathsf{A}$. 
    
    Note that the false clause search problem for $\EPHP$ and $\Tseitin$ are in $\PPP$ and $\PPA$ respectively. Therefore, using our characterization of the resolution width refuter for $\EPHP$ (\autoref{thm: white-box PLS-completeness for width}) and $\Tseitin$ (\autoref{thm: Tseitin width lower bound refuter}), we conclude that \emph{it is necessary and sufficient to use {local search} principle to separate $\PPP$ and $\PPA$ from $\PLS$ in the black-box setting.}

    \begin{theorem}[Informal version of \autoref{thm: TFNP ref PPP PPA PLS}]
        $\Rft{\PPP}{\PLS} = \Rft{\PPA}{\PLS} = \PLS$.
    \end{theorem}

\subsection{Discussions, Speculations, and Future Directions}\label{sec: discussions}
This paper initiates a research program that attempts to understand, for every proof system $\calP$ of interest, the metamathematics of proving lower bounds against $\calP$ through the lens of refuter problems.\footnote{We believe the similar research program for circuit lower bounds would also be fruitful, which has already started since \cite{CJSW21, Korten22, ChenLiOliveira24} if not earlier. We limit our discussions to proof lower bounds here.} Our results on resolution suggest that this is a promising direction. There are a plethora of future research directions, both regarding ``weak'' systems (where we already know strong lower bounds against $\calP$-proofs) and ``strong'' ones (where we are still struggling to prove non-trivial lower bounds against $\calP$).

\paragraph{Weak proof systems.} It might be feasible to characterize the complexity of refuter problems for weak proof systems. How does the complexity of refuting lower bounds for $\calP$ compare with $\calP$ itself (or, more precisely, the $\TFNP^\dt$ subclass corresponding to $\calP$ \cite{BFI23})? In the case that $\calP$ is resolution, our work shows that the complexity of refuting \emph{width} lower bounds for $\calP$ is exactly $\calP$ itself (i.e., $\PLS$), and the complexity of refuting \emph{size} lower bounds is a randomized version of $\calP$ (i.e., $\rwPHP(\PLS)$). Thus, it seems reasonable to conjecture that for ``weak'' proof systems, the complexity of proving lower bounds against them is not much higher than themselves.

Moreover, the proof complexity of proof complexity lower bounds is intimately connected to the hardness of automatability of proof systems, see e.g., \cite{AtseriasM19, GoosKMP20, Bell20, RezendeGNPR021, ItsyksonR22, Garlik24, Papamakarios24}. We expect that a thorough understanding of the former would help make progress on the latter as well.

\paragraph{Strong proof systems.} The situation for strong proof systems seems much more mysterious. For strong proof systems $\calP$ (think of $\calP$ being Frege or Extended Frege), it is even unclear whether there should be an ``easiest-to-prove'' lower bound for $\calP$ (which would correspond to a syntactic subclass $\calC(\calP)\subseteq \TFNP^\dt$ that characterizes the complexity of proving lower bounds for $\calP$). Even if such a $\calC(\calP)$ exists, it is unclear if it is captured within our current landscape of $\TFNP^\dt$.\footnote{Note that the question of where $\calC(\calP)$ sits in the $\TFNP^\dt$ hierarchy is merely a restatement of the open problem of determining the proof complexity of proof complexity lower bounds for $\calP$. For example, $\calC(\calP)$ is a subclass of $\mathsf{PTFNP}$ \cite{GoldbergP18} if and only if \textsc{q-eff} (the proof system underlying the definition of $\mathsf{PTFNP}$) can prove lower bounds for $\calP$.}

This suggests the following possibility: The reason that we have not been able to prove lower bounds for $\calP$ is that $\calC(\calP)$ is a very complicated class, far beyond our current understanding of $\TFNP^\dt$ and bounded arithmetic. An even more speculative hypothesis would be that the proof systems $\calP$ for which we are able to prove lower bounds are exactly those where $\calC(\calP)$ is not ``much'' higher than $\calP$ themselves. We hope that future work will determine to what extent these hypotheses are correct.

The case of $\AC^0[p]$-Frege (where $p$ is a prime) is of particular interest. Although strong lower bounds for $\AC^0[p]$ \emph{circuits} have been known for decades \cite{Razborov87, Smolensky87}, we have not yet succeeded in turning these circuit lower bounds into proof complexity lower bounds against $\AC^0[p]$-Frege (see, e.g., \cite{MacielP96, buss2015collapsing}). The paper \cite{BussIPRS97} laid out a research program towards $\AC^0[p]$-Frege lower bounds by studying weaker algebraic proof systems such as the Nullstellensatz \cite{BeameIKPP94} and Polynomial Calculus \cite{clegg1996using, Razborov98_PC}. After a few decades, we have become proficient at proving lower bounds against such algebraic proof systems, but lower bounds against $\AC^0[p]$-Frege remain elusive. Is it because the refuter problems corresponding to $\AC^0[p]$-Frege lower bounds are \emph{fundamentally different} from those for the weaker algebraic proof systems? Does our metamathematical $\TFNP^\dt$ perspective bring new insights to this long-standing open question?

\subsection{Further Related Works}\label{sec: related works}

\paragraph{Refuter problems for circuit lower bounds.} Our study of the refuter problems for proof lower bounds is strongly influenced by the line of work on refuter problems for circuit lower bounds. Chen, Jin, Santhanam, and Williams \cite{CJSW21} call a lower bound \emph{constructive} if the corresponding refuter problem can be solved in deterministic polynomial time, and they argued that constructivity is a desirable aspect of lower bounds. Chen, Tell, and Williams \cite{CTW23} showed that for many lower bounds against randomized computational models, their refuter problems characterize derandomizing $\textrm{pr-}\BPP$. The main result of Korten \cite{Korten22} can also be seen as the ${\sf WPHPWIT}$-hardness of refuter problems for one-tape Turing machine lower bounds. Pich and Santhanam \cite{PichS23} showed how to turn proof complexity lower bounds into circuit lower bounds, assuming the refuter problem for the (conjectured) lower bound $\SAT\not\in \P/_{\poly}$ is ``provably easy'' in a certain sense. Finally, the results of Chen, Li, and Oliveira \cite{ChenLiOliveira24} can be interpreted as the $\PWPP$- and ${\sf WPHPWIT}$-completeness of various refuter problems.

It is also worth mentioning that Ebtehaj \cite{ebtehaj2023variants} studied the refuter problems for $\calA\not\subseteq\BPP$ for each (type-$1$) subclass $\calA\subseteq\TFNP$ that is indeed hard. However, \cite{ebtehaj2023variants} did not obtain any completeness results for such refuter problems.

\paragraph{Unprovability of complexity upper bounds.} In parallel to the investigation of unprovability of complexity lower bounds, there is another line of work showing the unprovability of complexity \emph{upper} bounds in fragments of bounded arithmetic~\cite{CookK07, KrajicekO16, BydzovskyKO20, BydzovskyM20, CarmosinoKKO21, AtseriasBM23}. For example, \Krajicek and Oliveira~\cite{KrajicekO16} proved that Cook's theory $\PV$ cannot prove $\P\subseteq \SIZE[n^k]$, and Atserias, Buss, and \Muller~\cite{AtseriasBM23} proved that the theory $\mathsf{V}^0_2$ cannot prove $\NEXP\subseteq \P/_\poly$. These results are equivalent to the \emph{consistency} of lower bounds with fragments of bounded arithmetic, thus in some sense representing progress towards proving circuit lower bounds.\footnote{The ``conventional wisdom'' seems to believe that the complexity lower bounds are true (for discussions, see \url{https://rjlipton.com/conventional-wisdom-and-pnp/}, accessed Mar 14, 2026). Hence, unprovability of complexity lower bounds can be seen as the difficulty for proving this ``conventional wisdom'', while unprovability of complexity upper bounds represents progress towards proving it. One should keep in mind that the opposite opinion makes equal sense: for a believer of complexity \emph{upper} bounds, the unprovability of these upper bounds indicates the difficulty of confirming their belief, while the unprovability of lower bounds implies progress towards it!} Indeed, \cite{CarmosinoKKO21} presented a general framework for showing such consistency results by proving lower bounds against circuits with a certain uniformity condition called ``LEARN-uniformity'', and the techniques employed in many of these papers are inspired by uniform circuit lower bounds such as \cite{SanthanamW14}.

\paragraph{Witnessing theorems.} $\TFNP$ and bounded arithmetic are connected through \emph{witnessing theorems}: each theory is associated with the class of $\TFNP$ problems whose totality is provable in this theory. Perhaps the best-known witnessing theorem is Buss's one \cite{Buss85}: every $\NP$ search problem provably total in $\S^1_2$ can be solved in deterministic polynomial time. The class $\PLS$ and its generalizations such as $\mathsf{CPLS}$ capture the $\NP$ search problems provably total in higher levels of bounded arithmetic hierarchy \cite{BussKrajicek94, KrajicekST07, skelley2011provably, PudlakT12}; in this sense, witnessing theorems also provide a systematic method for defining new syntactic subclasses of $\TFNP$. Other witnessing theorems considered in the literature include \cite{KolodziejczykNT11, BeckmannB17, KolodziejczykT22}. Our paper contributes to this line of research by characterizing the class of $\NP$ search problems provably total in $\T^1_2 + \dwPHP(\PV)$ by the refuter problems corresponding to many resolution lower bounds, in particular the problem $\Refuter(s(\PHP_{(n+1)\to n}\vdash_\Res \bot) < c^n)$.

\paragraph{Comparison with the consistency search problem.} We note that the refuter problem looks superficially similar to $\WrongProof$, the \emph{consistency search} problem for proof systems \cite{BeckmannB17, GoldbergP18, Pudlak20}. Let $\calP$ be a proof system, $\WrongProof(\calP)$ is the $\TFNP^\dt$ problem that given as input a purported $\calP$-proof $\Pi$ of \emph{an incorrect statement}, asks for the location of an invalid derivation in $\Pi$.

Although both $\WrongProof$ and our refuter problems take a purported proof as input and ask for an invalid derivation in the proof, we think that these two problems are fundamentally different, because they have different \emph{reasons of totality}. Roughly speaking, the totality of $\WrongProof$ is proved by the \emph{soundness} of $\calP$, and the totality of $\Refuter$ is guaranteed by \emph{lower bound proofs}. We elaborate on this in \autoref{sec: discuss wrong proof}.

Another (superficial) similarity between these two problems is that both problems are used to characterize the provably total $\NP$ problems in bounded arithmetic. The consistency search problems for Frege and Extended Frege characterize the $\forall\Sigma_1^b$-consequences of $\mathsf{U}^1_2$ and $\mathsf{V}^1_2$ respectively \cite{BeckmannB17}, while in this paper we show that the refuter problem for resolution (with a suitable hard tautology) characterizes the $\forall\Sigma_1^b$-consequences of $\T^1_2 + \dwPHP(\PV)$.

\subsection*{Paper Organization}
The main body of the paper is structured so that the initial sections primarily focus on presenting the complexities of refuter problems of the pigeonhole principle. Specifically, \autoref{sec: prelim} contains the necessary preliminaries and formal definitions. \autoref{sec: PHP} presents the complexity upper bound of the refuters for $\PHP_{(n+1)\to n}$. 
\autoref{sec: general lower bounds} complements the previous section by showing universal hardness results for refuting any narrow or short resolution proofs. In \autoref{sec: more upper bounds}, we extend our results to many other formulas, including $\XOR$-lifted formulas, Tseitin formulas, and random $k$-CNFs. Finally, \autoref{sec: applications} introduces two novel applications of our results in understanding the proof complexity of proof complexity lower bounds and the complexity of black-box $\TFNP$ separations.

\section{Technical Overview}

\subsection{Refuter Problems in \texorpdfstring{$\rwPHP(\PLS)$}{rwPHP(PLS)}}\label{sec: overview of upper bounds}

In this subsection, we explain how the lower bound proof in \cite{CookP90, beame1996simplified} yields a reduction from the problem $\Refuter(s(\PHP_{(n+1)\to n}\vdash_\Res\bot) \le c^n)$ to $\rwPHP(\PLS)$. As mentioned before, this is essentially a formalization of the lower bound proof in $\T^1_2(\alpha) + \dwPHP(\PV(\alpha))$, and the reduction follows from the witnessing theorem for this theory. However, this subsection will describe the reduction without invoking the witnessing theorem (nor does the formal proof in \autoref{sec: upper bound for Refuter for Resolution} use the witnessing theorem). We hope that by opening up the black box of the witnessing theorem, it would become clearer how each component in the proof corresponds to a component in the reduction to $\rwPHP(\PLS)$.

The proof of \cite{CookP90, beame1996simplified} consists of two components:\begin{itemize}
	\item (Random restrictions) First, we carefully design a distribution of random restrictions $\calR$ under which the following holds. (1) With high probability over $\rho\gets\calR$, any fixed size-$c^n$ resolution proof will simplify to a resolution proof of \emph{width}\footnote{In fact, in the case of PHP, we obtain a resolution proof of small \emph{monotone width}. We omit the distinction between width and monotone width in the overview and refer the reader to \autoref{sec: PHP} for details.} at most $w$ under $\rho$ (for some parameter $w$); (2) the pigeonhole principle $\PHP_{(n+1)\to n}$ remains to be the pigeonhole principle (of a slightly smaller size $\PHP_{(n'+1)\to n'}$) under any restriction $\rho\in\calR$. Moreover, $\calR$ is the uniform distribution over some set of restrictions; we abuse notation and use $\calR$ to also denote this set.
	\item (Width lower bound) Then, we invoke the width lower bound for the pigeonhole principle and show that resolution cannot prove $\PHP_{(n'+1)\to n'}$ in width $w$.
\end{itemize}
	
\def\ellcomp{\ell_{\sf comp}}
Given a resolution proof $\Pi$ of size $c^n$, the fact that most restrictions simplify $\Pi$ into a small-width proof can be shown by a \emph{compression argument}: given a clause $C_i\in \Pi$ and a restriction $\rho \in \calR$ that does \emph{not} shrink $C_i$ into a clause of width $\le w$, one can describe $\rho$ in $\ellcomp$ bits for some small $\ellcomp$. In what follows, it suffices if $\ellcomp \le \log|\calR| - \log(c^n) - 1$, which is indeed the case under a suitable choice of parameters. Note that for comparison, if such a clause $C_i$ were not known, it would require, information-theoretically, at least $\log|\calR|$ bits to encode any restriction $\rho \in \calR$. This compression argument implies that a random $\rho\gets \calR$ shrinks a fixed clause w.p.~$\ge 1-\frac{1}{2c^n}$, thus the existence of a good $\rho$ shrinking the whole proof $\Pi$ follows from a union bound over the $c^n$ clauses in $\Pi$.
	
For every restriction $\rho$, one can compute a proof $\Pi|_\rho$ with each clause $C\in \Pi$ replaced by $C|_\rho$, the restriction of $C$ under $\rho$; if $\width(C|_\rho) > w$, we truncate $C|_\rho$ to force its width to be at most $w$. Since $\Pi|_\rho$ is a width-$w$ resolution proof, it follows from the width lower bound that it does not prove $\PHP_{(n'+1)\to n'}$.

Our reduction from $\Refuter(s(\PHP_{(n+1)\to n}\vdash_\Res \bot) \le c^n)$ to $\rwPHP(\PLS)$ works as follows.\begin{itemize}
	\item Let $N := |\calR| / 2$. The function $f:[N] \to [2N]$ takes as inputs $(i, s)$ where $i\in[c^n]$ denotes a node in $\Pi$ and $s$ is the compressed description of a random restriction $\rho$ that fails to simplify $C_i$ to width $w$ (note that this takes $\log(c^n) + \ellcomp \le \log N$ bits), and outputs the standard encoding of $\rho$ (in $\log|\calR| = \log(2N)$ bits). It is easy to see that every restriction $\rho$ outside the range of $f$ would be a good restriction (that successfully shrinks every clause in $\Pi$ into width $w$).
	\item For each $\rho \in \calR \cong [2N]$, $\Pi|_\rho$ is a width-$w$ resolution proof. By \autoref{thm: mono PLS-completeness}, we can reduce the problem of finding an illegal derivation in $\Pi|_\rho$ to $\PLS$. Call this instance $I_\rho$.
	\item Finally, let $i\in [c^n]$ be an illegal derivation in $\Pi|_\rho$ (that can be found in $\PLS$). There are two reasons that the $i$-th step is illegal in $\Pi|_\rho$: first, it might already be an invalid derivation in $\Pi$; second, the width of $C_i|_\rho$ might be greater than $w$, thus the error happens when we truncate $C_i|_\rho$ to width $w$. In the second case, let $s$ be the $\ellcomp$-bit description of $\rho$ given that it does not simplify $C_i$ to width $\le w$, and $g_{\rho, i} := (i, s)$, then $f(g_{\rho, i}) = \rho$.
\end{itemize}

It follows that once we found any $\rho$ and $i$ such that $i$ is an answer for $I_\rho$ and $f(g_{\rho, i}) \ne \rho$, then the $i$-th step is illegal for the first reason stated above, i.e., the $i$-th step in $\Pi$ is also invalid. See \autoref{fig: rwPHP Upper Bound} for a high-level overview of this construction.

\begin{figure}[t]
    \centering
    \includegraphics[trim={0 8em 0em 5em}, clip, width=0.98\linewidth]{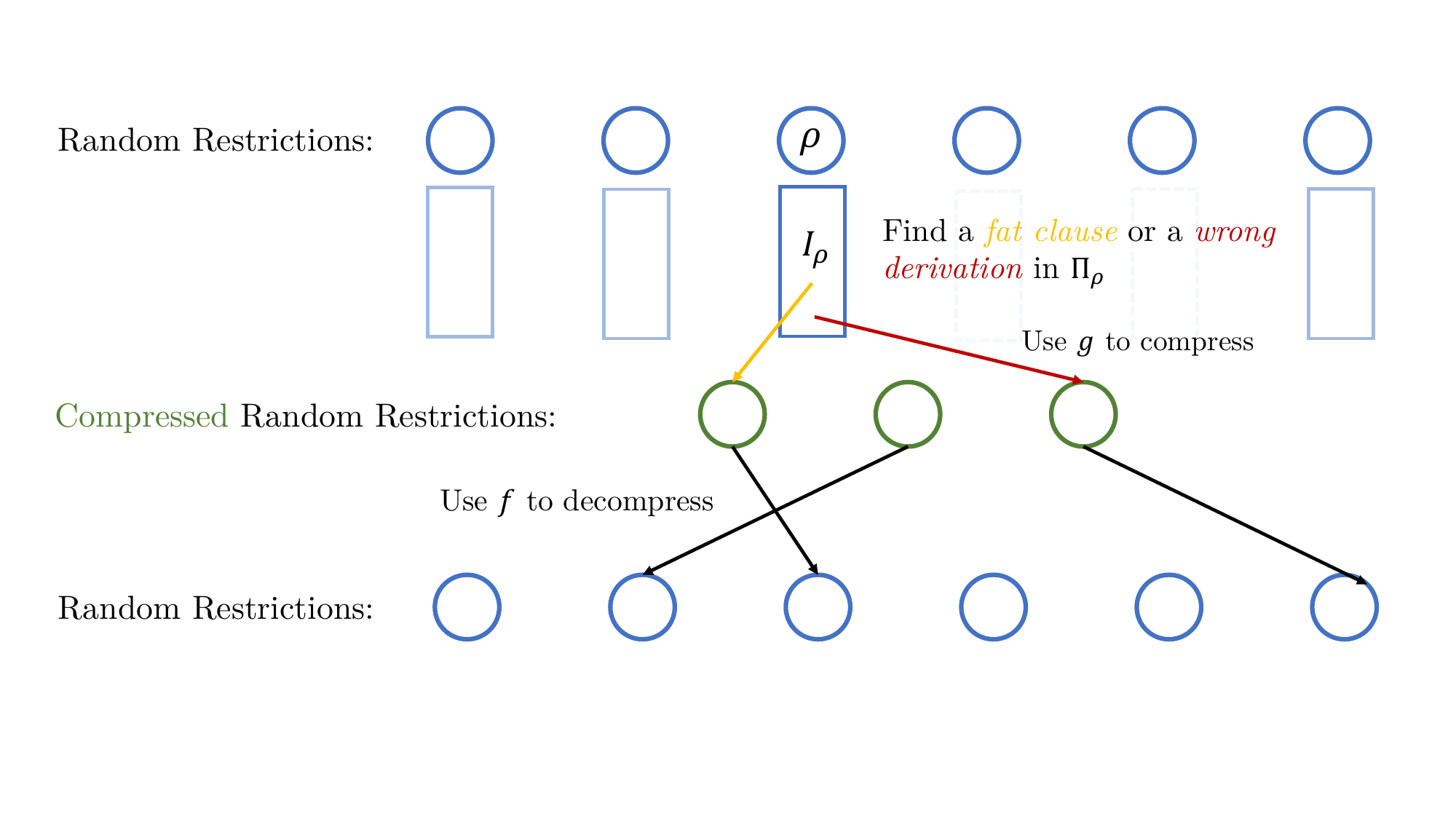}
    \caption{The $\rwPHP(\PLS)$ instance constructed from $\Refuter(s(\PHP_{(n+1)\to n}\vdash_\Res \bot) \le c^n)$.}
    \label{fig: rwPHP Upper Bound}
\end{figure}

\paragraph{Random restrictions + width lower bounds.} It turns out that the above proof template that combines random restrictions and width lower bounds is very popular in proving resolution lower bounds. Given a hard tautology $F$, we design a family of restrictions $\calR$ such that (1) Any fixed \emph{short} resolution proof will simplify to a \emph{narrow} resolution proof under $\calR$, and (2) even after a random restriction in $\calR$, $F$ remains hard for narrow resolution proofs. Note that the family $\calR$ is usually carefully chosen according to the hard tautology $F$; e.g., $\calR$ corresponds to \emph{partial matchings} when $F = \PHP$ \cite{beame1996simplified} and corresponds to \emph{random edge sets} when $F = \Tseitin$ \cite{Schoning97}.

As mentioned before, this proof strategy is capable of proving resolution size lower bounds for various hard tautologies, and we can use a similar argument as the above to show that the refuter problems corresponding to these resolution size lower bounds are in $\rwPHP(\PLS)$. This includes XOR-lifted formulas (\autoref{sec: rwPHP(PLS) for XOR-lifting}), Tseitin tautologies (\autoref{sec: tseitin}), and random $k$-CNFs (\autoref{sec: random k CNF}).

In fact, it is quite intuitive to formalize ``random restrictions + width lower bounds'' in $\T^1_2(\alpha) + \dwPHP(\PV(\alpha))$. Roughly speaking, we first use $\dwPHP(\PV(\alpha))$ to formalize the compression argument and show that most random restrictions will shrink the resolution proof (represented by $\alpha$) into a narrow one; then we use $\Sigma_1^b(\alpha)\text{-}\MIN$ (which is available in $\T^1_2(\alpha)$) to prove a resolution width lower bound.

\subsection{Refuter Problems are \texorpdfstring{$\rwPHP(\PLS)$}{rwPHP(PLS)}-Hard}\label{sec: overview of lower bounds}

In this subsection, we explain the ideas behind the reduction from $\rwPHP(\PLS)$ to the refuter problems for resolution size lower bounds. In fact, a reduction from $\rwPHP$ to the refuter problems already appeared in~\cite{RezendeGNPR021}, which provides a streamlined proof of the celebrated $\NP$-hardness of automating resolution~\cite{AtseriasM19}. It turns out that with minor modifications, the same proof can be adapted to reduce not only $\rwPHP$ but also $\rwPHP(\PLS)$ to the refuter problems, thereby proving \autoref{thm: rwPHP(PLS)-hardness of refuters}. Hence, the remainder of this subsection will focus on the $\rwPHP$-hardness result from~\cite{RezendeGNPR021}; the complete $\rwPHP(\PLS)$-hardness result can be found in \autoref{sec: rwPHP(PLS)-hardness}.

There is a clear intuition behind the reduction: suppose $\rwPHP$ \emph{were} false, i.e., there are functions $f:[N]\to [2N]$ and $g:[2N]\to [N]$ such that $f\circ g: [2N]\to [2N]$ is the identity function, then every unsatisfiable CNF $F$ \emph{would} have a resolution refutation of size $\poly(N, n)$. Of course, the ground truth is that such functions $f$ and $g$ should not exist, but a weak proof system (such as resolution itself) might not be aware of this. Suppose the weak system ``thinks'' that such a pair of functions $(f, g)$ \emph{might} exist, and it can construct a short resolution refutation of $F$ from $(f, g)$, then the weak system should also ``think'' that $F$ \emph{might} have a short resolution refutation. In summary, if it is hard to refute the existence of $(f, g)$ (which means proving $\rwPHP$), then it is also hard to prove that $F$ does not have a short resolution refutation.

Now, our task becomes the following. We live in a strange world where there is a \emph{surjection} from $[N]$ to $[2N]$; given an \emph{arbitrary} unsatisfiable CNF $F$, we want to construct a $\poly(N, n)$-size resolution refutation of $F$. Consider the size-$2^{O(n)}$ brute-force resolution refutation for every unsatisfiable CNF, which is represented by the following proof tree.\begin{itemize}
	\item The root (level $0$) of the tree contains the empty clause $\bot$.
	\item For each level $1\le i\le n$, each clause $C$ at level $i-1$ is resolved from the two clauses $C\lor x_i$ and $C\lor \overline{x}_i$, both of which sits in level $i$. The clauses $C\lor x_i$ and $C\lor\overline{x}_i$ are the two \emph{children} of $C$. Note that each clause at level $i$ has a width of exactly $i$.
	\item Finally, every clause $C$ at level $n$ corresponds to an assignment $x_C\in\{0, 1\}^n$ which is the only assignment falsifying $C$. Since $F$ is unsatisfiable, there is an axiom of $F$ that $x_C$ falsifies. Clearly, $C$ is a \emph{weakening} of this axiom.
\end{itemize}
\begin{figure}[h]
    \centering
    \includegraphics[width=0.8\linewidth]{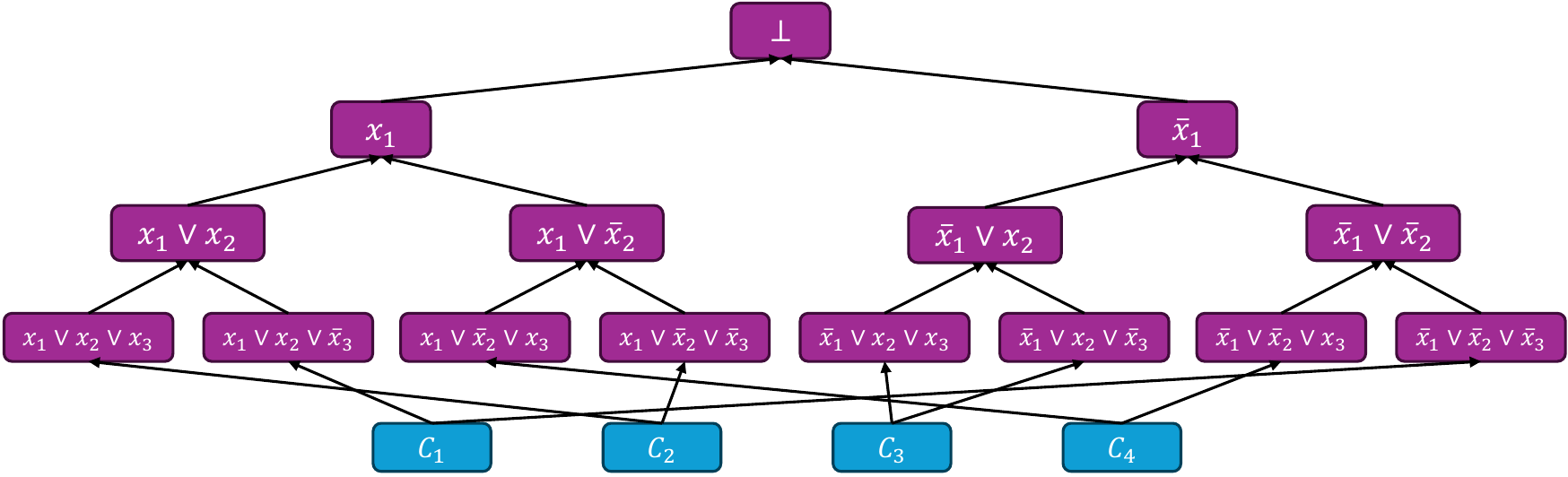}
    \caption{The brute-force resolution proof for a CNF $F = C_1 \land C_2 \land C_3 \land C_4$ when $n = 3$.}
\end{figure}

We now construct a shorter resolution refutation using the surjection from $[N]$ to $[2N]$. We guarantee that in our short refutation, each level never contains more than $N$ clauses; this implies that our resolution refutation is of size $O(N\cdot n)$. Consider level $i$ where $1\le i\le n$. If level $i-1$ contains at most $N$ clauses, then level $i$ contains at most $2N$ clauses: for each clause $C_j$ in level $i-1$, there are two clauses $C'_{2j} := C\lor x_i$ and $C'_{2j+1} := C\lor \overline{x}_i$ in level $i$. However, since there is a \emph{surjection} from $[N]$ to $[2N]$, it is possible to pick $N$ clauses among these $2N$ ones such that each of the $2N$ clauses appears in these $N$ ones! (The $j$-th picked clause ($j\in [N]$) is $C'_{f(j)}$; the clause $C'_j$ ($j\in [2N]$) appears as the $g(j)$-th picked clause.) Now that level $i$ also contains at most $N$ clauses, we can proceed to the next level and so on.

We stress again that the ground truth is, of course, that there do not exist functions $f:[N]\to [2N]$ and $g:[2N]\to [N]$ such that $f\circ g: [2N]\to [2N]$ is the identity function. However, the point is that given any step in the above resolution refutation that is an invalid derivation, we can pinpoint a ``witness'' number $x\in [2N]$ such that $f(g(x)) \ne x$.

The above describes the intuition behind the decision tree reduction from $\rwPHP$ to the refuter problems of resolution size lower bounds presented in \cite{RezendeGNPR021}. Our reduction from $\rwPHP(\PLS)$ to the refutation problems proceeds in the same way, except that now $g$ is only a function computable in $\PLS$. Compared with \cite{AtseriasM19, RezendeGNPR021}, our proof only has one more component: showing that these $\PLS$ instances can also be embedded into the above resolution refutation. We refer the reader to the formal proof in \autoref{sec: rwPHP(PLS)-hardness} for details.

\section{Preliminaries}\label{sec: prelim}

The first three subsections present standard preliminaries and can be skipped if the reader is familiar. However, the last two subsections introduce new concepts and it is highly recommended to read through (i.e., not skip) them. In particular, \autoref{sec: refutation problems} introduces the refuter problems for resolution lower bounds as $\TFNP^\dt$ problems, and \autoref{sec: rwPHP(PLS)} defines and discusses the subclass $\rwPHP(\PLS)$.

We use $0$-indexing: $[n] = \{0, 1, \dots, n-1\}$. For functions $f: \calA \to \calB$ and $g: \calB \to \calC$, their \emph{composition} $g\circ f$ is defined as
\[\forall x\in\calA, (g\circ f)(x) = g(f(x)).\]

\subsection{Pigeonhole Principle}
Let $m > n$, the \emph{pigeonhole principle} (PHP) states that there is no way to send $m$ pigeons into $n$ holes such that different pigeons are sent to different holes. This is expressed as the following unsatisfiable CNF $\PHP_{m\to n}$. (In the definition below, think of $x_{ij} = 1$ if pigeon $i$ goes to hole $j$.)
\begin{restatable}[$\PHP_{m\rightarrow n}$]{definition}{DefPHP}\label{def: PHP}
    $\PHP_{m\rightarrow n}$ is the conjunction of the following set of clauses:
    \begin{itemize}
        \item $\bigvee_{j\in[n]} x_{ij}$ for every pigeon $i\in[m]$;
        \item $\overline{x}_{ij}\vee\overline{x}_{i'j}$ for every two different pigeons $0\leq i<i'\leq m -1$ and every hole $j\in[n]$.
    \end{itemize}
\end{restatable}

The seminal work of Haken \cite{Haken85} proved that any resolution proof of $\PHP_{(n+1)\to n}$ requires $2^{\Omega(n)}$ size. The proof of this classical theorem has been simplified by several follow-up works \cite{CookP90,beame1996simplified,Ben-SassonW01}.

\subsection{Decision Tree \texorpdfstring{$\TFNP$}{TFNP}}
\label{sec: prelim: blackbox tfnp}

Let $\calO = \{O_N\}_{N}$ be a family of solution spaces. A \emph{search} problem $\calP$ is a family of sets $\{P_N\}_{N}$, where each $P_N$ is a subset of $\{0, 1\}^N \times O_N$. Let $x\in\{0, 1\}^N$ be an \emph{input} to $\calP$, we say that $o\in\calO_N$ is a \emph{solution} of $x$ if $(x, o) \in P_N$. We say $\calP$ is \emph{total} if every $x\in \{0, 1\}^*$ has at least one solution. We sometimes abuse the notation by calling an individual relation $P_N$ a search problem, and implicitly assume that there is a sequence $\{P_N\}_{N}$.

We study total search problems in the \emph{decision tree} model. In this model, we think of the input $x\in \{0, 1\}^N$ as very long and can only be accessed by querying individual bits. An algorithm (i.e., decision tree) is \emph{efficient} if it only makes $\polylog(N)$ many queries. We will typically consider search problems where $|O_N| \le 2^{\polylog(N)}$, so efficient algorithms will be able to handle solutions $o\in O_N$ in their entirety. A search problem $\calP$ is in $\FNP^\dt$ if given (oracle access to) an input $x \in \{0, 1\}^N$ and a solution $o\in O_N$, there is an efficient decision tree $T_o$ for deciding whether $(x, o) \in P_N$. The class $\TFNP^\dt$ consists of all \emph{total} search problems in $\FNP^\dt$.

For example, an important $\TFNP^\dt$ problem in this paper is the problem $\Iter$, defined as follows.%

\begin{mdframed}[hidealllines=true,backgroundcolor=gray!10]
    \begin{center}
        \textbf{Problem $\Iter$}
    \end{center}
    
    \underline{Input:} A function $S:[N] \to [N]$.
    
    \underline{Output:} A number $x\in [N]$ is a valid solution if one of the following holds:
    \begin{itemize}
        \item $x = 0$ and $S(0) = 0$;
        \item $S(x) < x$; or
        \item $S(x) > x$ and $S(S(x)) = S(x)$.
    \end{itemize}
\end{mdframed}

It is easy to check that $\Iter$ is in $\FNP^\dt$: Given an output $x$ and oracle access to the function $S:[N] \to [N]$, one can verify whether $x$ is a valid solution by querying at most $2$ entries of $S$; namely $S(x)$ and $S(S(x))$. Since each entry can be represented by at most $\log N$ bits, the query complexity of verifying solutions for $\Iter$ is $\polylog(N)$. On the other hand, the totality of $\Iter$ expresses the following fact: every DAG has a sink. It turns out that we will also frequently use a reversed version of $\Iter$ for simplicity, whose equivalence to $\Iter$ is easy to see: Given a function $S:[N]\to [N]$ such that $S(N-1) < N-1$, find some $x\in[N]$ such that 1) either $S(x) > x$ or 2) $S(x) < x$ and $S(S(x)) = S(x)$.

\begin{definition}[Decision tree reductions]
	Let $\calP, \calQ$ be two $\TFNP^\dt$ problems, and $d(N)$ be a parameter (typically $\polylog(N)$). A \emph{depth-$d$} decision tree reduction from $\calP$ to $\calQ$ consists of two functions $(\Reduce, \Decode)$, where each output bit of $\Reduce, \Decode$ can be computed from the input $x$ by a depth-$d$ decision tree:\begin{itemize}
		\item $\Reduce:\{0, 1\}^N\to \{0, 1\}^{M(N)}$ maps an input $x$ of $\calP$ to an input $\Reduce(x)$ of $\calQ$.%
		\item $\Decode$ maps any valid solution of $\Reduce(x)$ (as an instance of $\calQ$) into a valid solution of $x$ (as an instance of $\calP$). %
	\end{itemize}
	We say the reduction is \emph{uniform} if both $\Reduce$ and $\Decode$ can be computed by uniform Turing machines with query access to $x$. We allow $M(N)$ to be super-polynomial in $N$, but we require $M(N) \le \exp(d(N))$.%
\end{definition}

Usually, for two $\TFNP^\dt$ problems $\calP, \calQ$ we say $\calP$ can be (\emph{many-one}) reduced to $\calQ$ if there is a $\polylog(N)$-depth decision tree reduction from $\calP$ to $\calQ$.

The class $\PLS$\footnote{In this paper, most of the times when we mention a syntactic subclass of $\TFNP$ (such as $\PLS$) we mean the decision tree version of it (i.e., $\PLS^\dt$), and it should be easy to figure out whether we mean the decision tree version or the Turing machine version of this subclass from the context. Therefore, for convenience, we drop the superscript $\dt$ when we express syntactic subclasses of $\TFNP^\dt$. We still preserve the superscript $\dt$ in ``$\TFNP^\dt$'' when we want to emphasize that the underlying model is decision tree $\TFNP$.} is the class of problems in $\TFNP^\dt$ that has a depth-$\polylog(N)$ reduction to $\Iter$. (Note that $\PLS$ was originally defined differently \cite{DBLP:journals/jcss/JohnsonPY88}; the $\PLS$-completeness of $\Iter$ was shown in \cite{morioka2001classification}.)

The inputs of most $\TFNP^\dt$ problems introduced in this paper will be partitioned into blocks; for example, the input of $\Iter$ consists of $N$ blocks where each block consists of $\log N$ bits describing an integer in $[N]$. It will be more convenient to work with the \emph{block-depth} of decision trees, which is the number of different \emph{blocks} that a decision tree queries. For example, solutions of $\Iter$ can be verified in block-depth $2$. The problems in this paper will have block size $\polylog(N)$, hence $\polylog(N)$ block-depth is equivalent to $\polylog(N)$ (bit-)depth. However, we will upper bound the complexity of our decision trees by block-depth for convenience. Although the distinction of depth and block-depth does not make an essential difference in this paper, many interesting lifting theorems and non-automatability results are recently proved using the notion of block-depth (or block-width) \cite{AtseriasM19, GoosKMP20, RezendeGNPR021}. It might be beneficial to have bounds on block-depth, which is usually sharper as the decision trees we construct tend to query many bits in the same block.

We assume all the $\TFNP^\dt$ problems discussed in this paper are \emph{paddable}, i.e., for any $N < M$, solving an instance of size $N$ could always be efficiently reduced to solving an instance of size $M$ of the same problem. Most of the common $\TFNP^\dt$ problems can be easily formulated in a paddable way.\footnote{There is a similar notion called \emph{instance extension}, which is defined in \cite{B-OM04}.} 
    
\subsubsection{Connection to Proof Complexity}\label{sec: prelim: connection to proof complexity}
There is a generic connection between $\TFNP^\dt$ and propositional proof complexity via the \emph{false clause search} problem (see, e.g.~\cite{RezendeGR22survey,BFI23}).

\begin{definition}\label{def: false clause search problem}
    For an unsatisfiable CNF $F := C_1 \land \dots \land C_m$, $\SearchCNF(F)$ is the search problem in which an assignment $x$ to $F$ is given via query access, and a solution is a clause $C_i$ of $F$ falsified by $x$.

    Define $\SearchCNF(\calF)$ for a family of formula $\calF = \{F_n\}_{n\in\N}$ as $\{\SearchCNF(F_n)\}_{n\in\N}$ accordingly.
\end{definition}

    When the width of $F$ is $\polylog(n)$, where $n$ is the number of variables in $F$, $\SearchCNF(F)$ is a $\TFNP^\dt$ problem. In the other direction, for any $\TFNP^\dt$ problem $R_n \in \{0, 1\}^n \times O_n$, it can be equivalently written as $\SearchCNF(F_n)$ for some CNF $F_n$ of $\polylog(n)$ width.
    More specifically, let $\{T_o\}_{o \in O_n}$ be the set of efficient decision trees for verifying solutions, $\neg T_o(x)$ can be written as a low-width CNF stating that any accepting path in $T_o$ is falsified by $x$. We then take
    \begin{equation}\label{equ: search F}
        F_n = \bigwedge_{o \in O_n} \neg T_o(x),
    \end{equation}
    and it is easy to see the equivalence between $\SearchCNF(F_n)$ and $R_n$ by definition.

    Informally, we say a proof system $P$ is characterized by a syntactical $\TFNP^{\dt}$ subclass $\mathsf{C}$ if for any family of formula $\calF = \{F_n\}$, $P$ has a \emph{small} proof of $\calF$ if and only if $\SearchCNF(\calF) \in \mathsf{C}$. 
    Buss, Fleming, and Impagliazzo~\cite{BFI23} showed that any well-behaved\footnote{Here, a proof system is well-behaved if it is closed under decision tree reduction, and it can prove its own soundness.} proof system $P$ is characterized by a $\TFNP^{\dt}$ subclass $\mathsf{C}$, and vice versa. In particular, resolution is characterized by $\PLS$.
    \begin{theorem}[Folklore]\label{thm: pls characterize resolution}
        Let $\calF = \{F_n\}$ be a family of unsatisfiable formula,
        $\SearchCNF(\calF) \in \PLS$ if and only if $\calF$ have a $\polylog(n)$-width resolution refutation.
    \end{theorem}

\subsection{Bounded Arithmetic}
We introduce the theories $\T^1_2$, $\T^1_2 + \dwPHP(\PV)$, as well as their relativized versions. A more comprehensive introduction of bounded arithmetic (including the theories $\S^i_2$ and $\T^i_2$) can be found in \cite{Krajicek_BA_PL_CT}.

The language of bounded arithmetic consists of the following symbols
\[L_{\rm BA} := \{0, 1, +, \cdot, <, =, \lfloor \cdot/2\rfloor, |\cdot|, \#\}.\]
Here, the intended meaning of $|a|$ is the \emph{bit-length} of the binary number $a$, i.e.,
\[|a| := \begin{cases}\lceil\log_2(a+1)\rceil & \text{if }a>0;\\0 & \text{if }a=0.\end{cases}\]
The intended meaning of $\#$ (``smash'') is
\[x\#y := 2^{|x|\cdot |y|};\]
roughly speaking, this symbol is used to create objects whose size is \emph{polynomial}, instead of only \emph{linear}, in the length of its inputs. These symbols are governed by a list of $32$ axioms called $\rm BASIC$, each of which asserts some basic fact about the intended meanings of these symbols. For instance:
\begin{equation}
    a\le b \to a\le b+1.\tag{axiom 1 in $\rm BASIC$}
\end{equation}
The complete list of $\rm BASIC$ axioms can be found in \cite[Definition 5.2.1]{Krajicek_BA_PL_CT}.

A \emph{bounded quantifier} is a quantifier of the form
\[\forall y < t(\vec{x}) \quad\text{or}\quad \exists y < t(\vec{x})\]
for some term $t$. Formally, they are defined as abbreviations:
\begin{align*}
    \forall y < t(\vec{x})~\varphi(\vec{x}, y) := &\, \forall y~(y < t(\vec{x}) \to \varphi(\vec{x}, y));\\
    \exists y < t(\vec{x})~\varphi(\vec{x}, y) := &\, \exists y~(y < t(\vec{x}) \land \varphi(\vec{x}, y)).
\end{align*}
A \emph{sharply bounded quantifier} is a quantifier of the form
\[\forall y < |t(\vec{x})|\quad\text{or}\quad \exists y < |t(\vec{x})|.\]
That is, the domain of possible values of $y$ is bounded by the length of a term. Intuitively, sharply bounded quantifiers are ``feasible'' because, thinking of $t(\vec{x})$ as the description of a polynomial-size object, there are only polynomially many possibilities of $y$ and they can be enumerated in polynomial time.

A formula is \emph{sharply bounded} if all quantifiers in it are sharply bounded quantifiers. A \emph{$\Sigma_1^b$-formula} is a formula constructed from sharply bounded formulas using $\land$, $\lor$, sharply bounded quantifiers, and \emph{existential bounded quantifiers} (``$\exists y < t(\vec{x})$''). It can be shown that the languages defined by $\Sigma_1^b$-formulas are exactly those computed in $\NP$.

The power of theories in bounded arithmetic comes from their \emph{induction axioms}. Let $\Phi$ be a class of formulas, then $\Phi\text{-}\IND$ is the following axiom schema
\[(\phi(0) \land \forall x~(\phi(x) \to \phi(x+1))) \to \forall x~\phi(x)\]
for every $\phi\in \Phi$. The definition of $\T^1_2$ is:
\[\T^1_2 := {\rm BASIC} + \Sigma_1^b\text{-}\IND.\]
That is, when reasoning in $\T^1_2$, it is allowed to use induction axioms over $\Sigma_1^b$ formulas (i.e., $\NP$ languages).

It is equivalent, and sometimes more convenient to replace $\Sigma_1^b\text{-}\IND$ with $\Sigma_1^b\text{-}\MIN$, the \emph{minimization principle} over $\Sigma_1^b$ formulas. For a set of formulas $\Phi$, the axiom schema $\Phi\text{-}\MIN$ consists of
\[\phi(a) \to \exists x\le a \forall y < x ~ (\phi(x) \land \lnot \phi(y))\]
for every $\phi\in \Phi$. Equivalently, when reasoning in $\T^1_2$, it is allowed to use the fact that there exists a \emph{smallest} $x$ such that $C(x) = 1$, whenever $C$ is a polynomial-size \emph{nondeterministic circuit} and we know some $y$ such that $C(y) = 1$.

The theory $\PV$ is an equational theory defined by Cook \cite{Cook75} to capture polynomial-time reasoning. It contains a function symbol for every polynomial-time algorithm, introduced inductively using Cobham's recursion-theoretic characterization of polynomial time \cite{cobham1964intrinsic}. More detailed treatments about $\PV$ can be found in \cite{Krajicek_BA_PL_CT, Cook-Nguyen, ChenLiOliveira24}. In the literature, it is common to also use $\PV$ to denote the set of function symbols in $\PV$ (which corresponds to functions computable in polynomial time).

The \emph{dual weak pigeonhole principle} over $\PV$ functions, denoted as $\dwPHP(\PV)$, is the following axiom schema
\[\forall a > 1 \exists v < a^2 \forall u < a~f(u) \ne v\]
for every $\PV$-function $f$ with parameters\footnote{This is the standard terminology in bounded arithmetic that means $f$ might depend on some other parameter not shown above. The parameter can be thought of as non-uniformity; cf.~\autoref{footnote: parameters in bounded arithmetic}.}. Roughly speaking, this means that if we have a polynomial-size circuit $f: \{0, 1\}^n \to \{0, 1\}^{2n}$ (think of $a = 2^n$ above), then there exists some $v \in \{0, 1\}^{2n}$ that is not in the range of $C$. We note that the choice of $a^2$ above is somewhat arbitrary, as $\dwPHP(\PV)$ with various parameters are equivalent over $\S^1_2 \subseteq \T^1_2$ \cite{ParisWW88, Jerabek04}.

To summarize, when reasoning in the theory $\T^1_2 + \dwPHP(\PV)$, one is allowed to use the following two axiom schemas:
\begin{description}
    \item [{($\Sigma_1^b\text{-}\MIN$)}] For a polynomial-size nondeterministic circuit $C$ and some $y$ such that $C(y) = 1$, there exists a \emph{smallest} $x$ such that $C(x) = 1$.
    \item [{($\dwPHP(\PV)$)}] For a polynomial-size circuit $C: \{0, 1\}^n \to \{0, 1\}^{2n}$, there exists a string $y \in \{0, 1\}^{2n}$ that is not in the range of $C$.
\end{description}

Finally, the \emph{relativized} theories $\T^1_2(\alpha)$ and $\T^1_2(\alpha) + \dwPHP(\PV(\alpha))$ are simply their unrelativized counterparts with a new unary relation symbol $\alpha$ added into the language $L_{\rm BA}$. (One can think of $\alpha$ as an oracle that encodes an exponentially-long input; for example, $\alpha(i)$ might encode the $i$-th bit of an exponentially-long resolution proof according to some canonical encoding.) The class of $\Sigma_1^b(\alpha)$ formulas and axioms $\Sigma_1^b(\alpha)\text{-}\MIN$ and $\dwPHP(\PV(\alpha))$ are relativized in a straightforward way. There are no other axioms involving $\alpha$ except for the induction axioms and dual weak pigeonhole principles. To summarize:
\begin{itemize}
    \item When reasoning in $\T^1_2(\alpha)$, it is allowed to use $\Sigma_1^b(\alpha)\text{-}\MIN$, i.e., for any polynomial-size nondeterministic \emph{oracle} circuit $C^\alpha$ and input $y$ such that $C^\alpha(y) = 1$, there exists a \emph{smallest} input $x$ such that $C^\alpha(x) = 1$.
    \item When reasoning in $\T^1_2(\alpha) + \dwPHP(\PV(\alpha))$, it is additionally allowed to use the fact that for any polynomial-size \emph{oracle} circuit $C^\alpha: \{0, 1\}^n \to \{0, 1\}^{2n}$, there exists some $y \in \{0, 1\}^{2n}$ that is not in the range of $C^\alpha$.
\end{itemize}

\subsection{Refuter Problems for Resolution Lower Bounds}\label{sec: refutation problems}

We provide formal definitions of the refuter problems in the decision tree model. We begin by defining \emph{resolution refutations}; the definition is adapted from \cite[Section 3.1]{RezendeGNPR021}.

\begin{definition}\label{def: Resolution refutation}
    Let $F$ be an unsatisfiable CNF with $n$ variables and $m$ clauses; the clauses in $F$ will be called \emph{axioms} and will be denoted as $C_{-m}, \dots, C_{-1}$ for convenience. A \emph{resolution refutation} of $F$ is a sequence of nodes $C_0, C_1, \dots, C_{L-1}$, where each node $C_i$ contains the following information.
    \begin{itemize}
        \item A set of literals among $\{x_1, x_2, \dots, x_n, \overline{x}_1, \overline{x}_2, \dots, \overline{x}_n\}$. Abusing notation, we also denote the clause consisting of the disjunction of these literals by $C_i$.
        \item A \emph{tag} which is one of the following: ``resolution'' or ``weakening''.
        \item Two integers $-m \le j, k < i$ and a variable $a\in \{1, 2, \dots, n\}$ if the tag is ``resolution''. This means that $C_i$ is obtained from the clauses $C_j$ and $C_k$ by \emph{resolving} the variable $x_a$.
        \item One integer $-m \le j < i$ if the tag is ``weakening''. This means that $C_i$ is a \emph{weakening} of $C_j$.
    \end{itemize}

    The resolution refutation is valid if the following is true for every $1\le i\le L$:\begin{itemize}
        \item If $C_i$ is marked ``resolution'', then there are clauses $D$ and $E$ such that $C_j = x_a\lor D$, $C_k = \overline{x}_a \lor E$, and $C_i = D \lor E$. 
        \item If $C_i$ is marked ``weakening'', then there is a clause $D$ such that $C_i = C_j \lor D$.
        \item Finally, $C_{L-1} = \bot$ (i.e., contains no literals).
    \end{itemize}

    The \emph{length} or \emph{size} of the refutation is $L$, and the \emph{width} of the refutation is the maximum integer $w$ such that every clause $C_i$ ($-m \le i < L$) in the refutation contains at most $w$ literals.
\end{definition}

Resolution is \emph{complete} and \emph{sound}: a CNF $F$ has a resolution refutation (of whatever length) if and only if it is unsatisfiable.

Each node in the resolution refutation would be a \emph{block}; therefore, when we say a decision tree over a resolution refutation has block-depth $d$, we mean that it only queries (potentially all information in) $d$ nodes of the refutation.

Next, we define the refuter problems.

\begin{definition}\label{def: refutation problem for Res}
    Let $\calF=\set{F_n}_{n\in\N}$ be a family of unsatisfiable CNFs where every $F_n$ requires resolution of width greater than $w_n$ and size greater than $s_n$.
    \begin{itemize}
        \item An input to the problem $\Refuter(w(F_n \vdash_{\Res} \bot) \leq w_n)$ is a purported resolution refutation of $F_n$ with width at most $w_n$. (It is easy to syntactically guarantee that the width of the input refutation is at most $w_n$ by allocating only $w_n$ literals for each node.)
        \item An input to the problem $\Refuter(s(F_n \vdash_{\Res} \bot) \leq  s_n)$ is a purported resolution refutation of $F$ with at most $s_n$ clauses.
    \end{itemize}
    The outputs of these problems consist of only one index $i$, which means the node $C_i$ does not satisfy the validity conditions defined in \autoref{def: Resolution refutation}. We will call such nodes \emph{invalid derivations} or \emph{illegal derivations}.
\end{definition}

Note that each node can be described in $\poly(w, \log n, \log L)$ bits where $w$ is the width of the resolution refutation. Hence, in the typical parameter regime, we will consider resolution refutations whose length is exponential in its width ($L = 2^{w^{\Omega(1)}}$), so that the access to each block is ``efficient'', i.e., only needs to query polylogarithmic many bits. %
In particular, the typical parameter regime for size lower bounds is exponential ($L = 2^{n^{\Omega(1)}}$); polynomial width lower bound ($w = n^{\Omega(1)}$) is considered in \autoref{sec: PHP} and \autoref{sec: more upper bounds}, while polylogarithmic width lower bound is considered in \autoref{sec:app:tfnp}. We also assume $L = 2^{n^{O(1)}}$, so that the proof is not extremely redundant.

Since there is a decision tree of block-depth at most $3$ verifying whether a given node in the resolution refutation is an invalid derivation, the refuter problems defined above are in $\FNP^\dt$. Moreover, if the resolution lower bounds ($w_0$ or $s_0$) are indeed true, then the refuter problems defined above are total. Hence, $\Refuter(\cdot)$ is a natural family of problems in $\TFNP^\dt$.

\subsection{\texorpdfstring{$\calP$}{P}-Retraction Weak Pigeonhole Principle}\label{sec: rwPHP(PLS)}

Recall that $\rwPHP$, the \emph{retraction weak pigeonhole principle}, is the following principle:

\begin{fact}
	Let $g:[2M]\to [M]$ and $f:[M] \to [2M]$ be two functions. Then there must exist some $y\in[2M]$ such that $f(g(y)) \ne y$. 
\end{fact}

Roughly speaking, for any $\TFNP$ class $\calP$, $\rwPHP(\calP)$ is the retraction weak pigeonhole principle where the retraction ($g:[2M] \to [M]$) is a (multi-valued) function computable in $\calP$. For example:

\begin{mdframed}[hidealllines=true,backgroundcolor=gray!10]
    \begin{center}
        \textbf{Problem $\rwPHP(\PLS)$}
    \end{center}
    
    \underline{Input:} Let $M\leq N/2$. The input consists of the following functions:
    \begin{compactitem}
		\item $f:[M]\to [N]$ is a purported ``surjection'';
		\item for each $y\in[N]$, $I_y := (L,S_y)$ is an instance of $\Iter$, where $S_y:[L]\to [L]$; and
		\item $g_y:[L]\rightarrow [M]$ maps solutions of $I_y$ to integers in $[M]$.
    \end{compactitem}

    \underline{Output:} A number $y\in [N]$ and $ans\in[L]$ such that $ans$ is a solution of the $\Iter$ instance $g_y$ and $f(g_y(ans)) \ne y$.
\end{mdframed}

In general, for a $\TFNP^\dt$ problem $\calP$, we define $\rwPHP(\calP)$ by replacing each $I_y$ in the above definition with an instance of $\calP$. For convenience, we sometimes use the same notation $\rwPHP(\calP)$ to indicate the $\TFNP$ subclass that contains any $\TFNP$ problem reducible to the problem $\rwPHP(\calP)$. It will be clear from the context whether $\rwPHP(\calP)$ refers to the problem or the class in the remainder of the paper.

\begin{fact}\label{fact: rwPHP in TFNP}
	$\rwPHP(\calP) \in \TFNP^\dt$.
\end{fact}
\begin{proof}
	To verify a solution $(y, ans)$, check that $ans$ is a valid solution for $I_y$ and that $f(g_y(ans)) \ne y$. 
	
	The totality (i.e., existence of solutions) can be argued as follows. For each $y\in[N]$, let $ans(y)$ be a solution (say the lexicographically first one) of $I_y$; since $\calP$ is a total problem, $ans(y)$ exists. Let $g'(y) := g_y(ans(y))$. By the retraction weak pigeonhole principle, there exists some $y\in[N]$ such that $f(g'(y)) \ne y$. It follows that $(y, ans(y))$ is a valid solution for $\rwPHP(\calP)$.
\end{proof}

\begin{fact}\label{fact: rwPHP(P) is hard for both rwPHP and P}
	There is a depth-$1$ decision tree reduction from $\calP$ to $\rwPHP(\calP)$ and a depth-$1$ decision tree reduction from $\rwPHP$ to $\rwPHP(\calP)$.
\end{fact}
\begin{proof}
	To reduce $\calP$ to $\rwPHP(\calP)$: let $I$ be a $\calP$ instance. Define $f(x) = 1$ as a trivial function; for each $y\in[N]$, define the instance $I_y := I$; for every possible answer $ans$ of $I_y$, let $g_y(ans) = 1$. Clearly, for any answer $(y, ans)$ of the $\rwPHP(\calP)$ instance, $ans$ itself would be a valid answer of the $\calP$ instance $I$.
	
	To reduce $\rwPHP$ to $\rwPHP(\calP)$: let $f:[M]\to[N]$ and $g:[N]\to[M]$ be an $\rwPHP$ instance. Fix any (say trivial) $\calP$ instance $I$. For each $y\in[N]$, define the instance $I_y := I$; for every possible answer $ans$ of $I_y$, let $g_y(ans) = g(y)$. For any answer $(y, ans)$ of the $\rwPHP(\calP)$ instance, since $f(g_y(ans)) \ne y$, it follows that $f(g(y)) \ne y$, and hence $y$ is a valid answer of the $\rwPHP$ instance $(f, g)$.
\end{proof}

\begin{fact}
    $\rwPHP(\PLS)$ is not in $\PLS$ in the black-box setting. 
\end{fact}

\begin{proof}
    Classical techniques (such as Prover-Delayer games) show that the totality of $\rwPHP$ requires resolution width $\Omega(M)$ to prove (see also \cite{PudlakT19, RezendeGNPR021}), which means that any decision tree of depth $o(M)$ cannot reduce $\rwPHP$ to $\PLS$. It follows from \autoref{fact: rwPHP(P) is hard for both rwPHP and P} that there is a black-box separation between $\rwPHP(\PLS)$ and $\PLS$ itself.
\end{proof}

\begin{fact}\label{fact: robustness of rwPHP(P)}
	Let $\calP$ and $\calQ$ be $\TFNP^\dt$ problems. If there is a depth-$d$ decision tree reduction from $\calP$ to $\calQ$, then there is a depth-$d$ decision tree reduction from $\rwPHP(\calP)$ to $\rwPHP(\calQ)$.
\end{fact}
\begin{proof}
	Let $f:[M] \to [N]$, $\{I_y\}_{y\in [N]}$, and $\{g_y\}_{y\in [N]}$ be an instance of $\rwPHP(\calP)$. Define an instance $(f', \{I'_y\}, \{g'_y\})$ of $\rwPHP(\calQ)$ as follows. The function $f' := f$ stays the same; each $I'_y$ is obtained by running the reduction from $\calP$ to $\calQ$ on $I_y$; for each possible answer $ans'$ of each $I'_y$, we run the reduction to obtain an answer $ans$ of $I_y$, and return $g'_y(ans') := g_y(ans)$.
	
	Finally, let $(y', ans')$ be a valid answer of the instance $(f', \{I'_y\}, \{g'_y\})$. Let $y := y'$ and run the reduction to obtain an answer $ans$ of $I_y$ from the answer $ans'$ of $I'_y$, then $(y, ans)$ is a valid answer of $(f, \{I_y\}, \{g_y\})$.
\end{proof}

It follows from \autoref{fact: robustness of rwPHP(P)} that, for example, the class $\rwPHP(\PLS)$ can be defined from \emph{any} complete problem for $\PLS$. %

\paragraph{Amplification for $\rwPHP$.} In this paper, when we talk about $\rwPHP$, we always think about a purported ``surjection'' $f:[M]\to [N]$ where $N = 2M$. This is without loss of generality since the complexity of $\rwPHP$ does not depend significantly on the relationship between $M$ and $N$ (unless they are too close to each other). This is also true for $\rwPHP(\calP)$ provided that $\calP$ is closed under \emph{Turing reductions}. We note that many interesting subclasses of $\TFNP$ (such as $\PLS$, $\PPA$, $\PPAD$, and $\PPADS$) are indeed closed under Turing reductions \cite[Section 6]{BussJ12}, with the notable exception of $\PPP$ in the black-box model \cite{PPP-Turing}.

\begin{theorem}[Informal]\label{informal thm: amplification of rPHP}
    Suppose $\calP$ is closed under Turing reductions. Let $\rPHP_{M\to N}(\calP)$ denote the problem $\rwPHP(\calP)$ where the given ``purported surjection'' is from $[M]$ to $[N]$.\footnote{The notation $\rPHP_{M\to N}$ means ``retraction pigeonhole principle from $M$ pigeons to $N$ holes''. As ``weak'' conventionally refers to the case where $M=2N$, the retraction pigeonhole principle with $M$ and $N$ specified are called $\rPHP_{M\to N}$ instead of $\rwPHP_{M\to N}$.} Then there is an efficient decision tree reduction from $\rPHP_{M \to (M+M/\polylog(M))}(\calP)$ to $\rPHP_{M\to M^{100}}(\calP)$.
\end{theorem}

We remark that amplification of weak pigeonhole principles is a well-known fact in bounded arithmetic \cite{ParisWW88, Thapen-PhD, Krajicek04a, Jerabek04, Jerabek-independence, ChenLiOliveira24} and total search problems \cite{Korten21, Korten22}. Since the proof follows from standard arguments in the literature, we postpone it to \autoref{sec: amplification for wPHP}.

\subsubsection{Witnessing for \texorpdfstring{$\T^1_2 + \dwPHP(\PV)$}{T12 + dwPHP(PV)}}\label{sec: rwPHP(PLS) from witnessing}
\def\Sol{\mathrm{Sol}}

As mentioned in the introduction, every $\TFNP$ problems whose totality can be proved in $\T^1_2 + \dwPHP(\PV)$ reduces to $\rwPHP(\PLS)$. This is an easy corollary of \cite{BussKrajicek94} and \cite[Lemma 2.1]{AtseriasT14}; for completeness we present a proof here.

\begin{theorem}
    Suppose that $\phi(x) := \exists y<t_\phi(x)~\psi(x, y)$ is a $\Sigma_1^b(\alpha)$-sentence and 
    \[\T^1_2(\alpha) + \dwPHP(\PV(\alpha)) \vdash \forall x~\phi(x).\]
    Then the $\TFNP$ problem corresponding to $\phi$ is in $\rwPHP(\PLS)$.
\end{theorem}
\begin{proof}
It is shown in \cite[Lemma 2.1]{AtseriasT14} that there is a term $t = t(x)$ and a function symbol $f\in \PV(\alpha)$ such that
\begin{equation}\label{eq: lemma 1 of AT14}
	\T^1_2(\alpha) \vdash \forall x~(t > 2 \land \forall v<t^2 \exists u<t~(f_x(u) = v \lor \phi(x))).
\end{equation}
(In fact, this follows from standard manipulations underlying Wilkie's witnessing theorem for $\S^1_2 + \dwPHP(\PV)$; see e.g., \cite[Proposition 14]{Jerabek04}.)

To parse \eqref{eq: lemma 1 of AT14}, note that $f_x:[t]\to [t^2]$ defines a ``stretching'' function and hence cannot be surjective. The non-existence of $u$ such that $f_x(u) = v$ exactly means that $v$ is not in the range of $f_x$; hence given \eqref{eq: lemma 1 of AT14}, $\forall x~\phi(x)$ follows from the dual weak pigeonhole principle.

The following problem is in $\PLS$ by the witnessing theorem for $\T^1_2$ \cite{BussKrajicek94}. On input $(x, v)$, find either some $u < t(x)$ such that $f_x(u) = v$ or some $y < t_\phi(x)$ such that $\psi(x, y)$ holds (we call such $y$ a ``certificate'' for $x$). Now it is at least easy to see that there is a \emph{randomized} reduction from the $\TFNP$ problem corresponding to $\phi(x)$ to $\PLS$: Let $v\gets [t^2]$ be random (which is a non-output of $f_x$ w.h.p.), then the above $\PLS$ procedure finds some $y$ that is a certificate for $x$.

Working slightly harder, we can see that the above is actually a reduction to $\rwPHP(\PLS)$:\begin{itemize}
	\item The ``purported surjection'' is the function $f_x:[t]\to [t^2]$.
	\item For each $v\in [t^2]$, there is a $\PLS$ instance $I_v$ which captures the problem of given $(x, v)$ outputting either $u \in f_x^{-1}(v)$ or a certificate $y$ for $x$.
\end{itemize}
Hence, given any $v\in [t^2]$ and solution $ans$ of $I_v$ such that $ans$ does not contain the information of $u\in [t]$ such that $f(u) = v$, it must be the case that $ans$ leads to a certificate for $x$. It follows that the $\TFNP$ problem corresponding to $\phi$ reduces to $\rwPHP(\PLS)$.
\end{proof}

On the other hand, it is easy to see that $\T^1_2 + \dwPHP(\PV)$ proves the totality of $\rwPHP(\PLS)$ and the proof relativizes. Let $(f, \{I_y\}, \{g_y\})$ be an instance of $\rwPHP(\PLS)$. By $\dwPHP(f)$, there exists $y\in [N]$ that is not in the range of $f$. Since $\T^1_2$ proves the totality of $\PLS$, it also proves the existence of a solution $ans$ for $I_y$. Note that since $y$ is not in the range of $f$, we have in particular that $f(g_y(ans)) \ne y$. Hence, $(y, ans)$ is a valid solution for this $\rwPHP(\PLS)$ instance.

\def\cri{\mathrm{cri}}
\def\Cri{\mathrm{Cri}}
\def\mono{\mathsf{mono}}

\section{Refuters for the Pigeonhole Principle}\label{sec: PHP}

In this section, we study the refuter problems where the family of hard tautologies is the Pigeonhole principle $\PHP_{(n+1)\rightarrow n}$. Our main results are the $\PLS$-memberships of width refuter problems for (variants of) $\PHP_{(n+1)\rightarrow n}$ and the $\rwPHP(\PLS)$-memberships of the size refuter problems for $\PHP_{(n+1)\rightarrow n}$. Looking ahead, in \autoref{sec: general lower bounds}, we will establish \emph{universal} $\PLS$-hardness for width refuters and $\rwPHP(\PLS)$-hardness for size refuters, thereby characterizing their complexities in $\TFNP^\dt$.

\subsection{Refuters for Narrow Resolution Proofs}
\label{subsection: width white-box}
Historically, proving size lower bounds for resolution has been challenging and considered milestones in proof complexity. The honor of the first super-polynomial size lower bounds belongs to the Pigeonhole Principle ($\PHP_{(n+1)\rightarrow n}$). However, in terms of width lower bound, $\PHP$ is not satisfactory: $\width(\PHP_{(n+1)\rightarrow n}\vdash_{\Res}\bot)=n$, but there is already an \emph{axiom} in $\PHP$ that has width $n$. Therefore, studying the complexity of finding a wide clause is uninteresting, as one of the widest clauses appears directly in the axiom and can be easily located. In what follows, we consider the width refuter problem of a \emph{constant-width analog} of $\PHP_{(n+1)\rightarrow n}$, namely, the nondeterministic extension $\EPHP_{(n+1)\rightarrow n}$, defined below. However, we will come back to $\PHP_{(n+1)\rightarrow n}$ shortly after $\EPHP_{(n+1)\rightarrow n}$ and examine the refuter problem for the so-called ``monotone width'' of $\PHP_{(n+1)\rightarrow n}$. The monotone width lower bound for $\PHP_{(n+1)\to n}$ will also serve as a key component in the study of size refuters for $\PHP_{(n+1)\to n}$.

\begin{definition}[$\EPHP_{(n+1)\rightarrow n}$]\label{def: EPHP}
    $\EPHP_{(n+1)\rightarrow n}$ is the same as $\PHP_{(n+1)\rightarrow n}$ except that we replace every clause $\bigvee_{j\in[n]} x_{ij}$ by a $3$-CNF nondeterministic extension; that is, by the following $n+2$ clauses: 
    \[\overline{y}_{i,0}, (y_{i,0}\vee x_{i,1}\vee\overline{y}_{i,1}),(y_{i,1}\vee x_{i,2}\vee\overline{y}_{i,2}),\cdots,(y_{i,n-1}\vee x_{i,n}\vee\overline{y}_{i,n}),y_{i,n},\]
    where $y_{i,0},\cdots,y_{i,n}$ are newly introduced variables.
\end{definition}

Width lower bounds for $\EPHP_{(n+1)\to n}$ were proved by Ben-Sasson and Wigderson \cite{Ben-SassonW01}.
\begin{theorem}[{\cite[Theorem 4.9]{Ben-SassonW01}}]\label{claim: width lower bound for EPHP}
	Any resolution refutation of $\EPHP_{(n+1)\to n}$ contains a clause $C$ with $w(C) \ge n/3$.
\end{theorem}

\begin{theorem}\label{thm: white-box PLS-completeness for width}
    $\Refuter(w(\EPHP \vdash_{\Res}\bot)< n/3)$ is in $\PLS$. In particular, there is a uniform decision tree reduction of block-depth $3$ from the refuter problem to $\Iter$.
\end{theorem}
\begin{proof}
    We will reduce it to an instance of reversed $\Iter$. This reduction is an analog of a \emph{constructive} version of width lower bound proofs by Beame and Pitassi~\cite{beame1996simplified}, which we will use again later in the proof of \autoref{thm: mono PLS-completeness}. We call $(x_{i,j})$ \emph{original} variables and $(y_{i,j})$ \emph{extension} variables. 

    We consider a set of functions over all the variables, including $n+1$ pigeon functions $\{EP_i\}$ where 
    \begin{equation}
    EP_i := \overline{y}_{i,0} \wedge(y_{i,0}\vee x_{i,1}\vee\overline{y}_{i,1})\wedge(y_{i,1}\vee x_{i,2}\vee\overline{y}_{i,2})\wedge\cdots\wedge(y_{i,n-1}\vee x_{i,n}\vee\overline{y}_{i,n})\wedge y_{i,n},\end{equation}
    and $n^2(n+1)/2$ hole functions $\set{H^j_{(i, i')}}$, where $H^j_{(i,i')}=\overline{x}_{i,j}\vee \overline{x}_{i',j}$. It is easy to see that the semantic meaning of $y_{i, j}$ is whether ``the index of the hole that pigeon $i$ goes into belongs to $\set{1,\dots,j}$'' or not. 

    We say an assignment $\alpha$ over all variables (including the extension variables) is \emph{$\ell$-critical} if $EP_{\ell}(\alpha)=0$ but all other functions are 1 under this assignment, namely $EP_i(\alpha)=1$ for all $i\neq \ell$ and $H^j_{(i,i')}(\alpha)=1$ for all $j\in[n],i,i'\in[n+1]$. If we ignore the extension variables, $\alpha$ is essentially a complete matching between pigeons and holes, except pigeon $i$ is not going anywhere.
    Given the definition of $\ell$-critical assignments, we define a complexity measure for a clause $C$, denoted by $\cri(C)$:
    \[\cri(C)\coloneqq \mleft|\set{\ell : \exists \ell\text{-critical assignment } \alpha \text{ such that } C(\alpha)=0}\mright|.\]

    Note that $\cri$ has four important properties:
    \begin{enumerate}[(I)]
        \item $\cri(\bot)=n+1$;\label{item: property I of cri}
        \item $\cri(EP_i)=1$ for all $i$ and $\cri(H^j_{i,i'})=0$ for all $j,i,i'$;\label{item: property II of cri}
        \item $\cri$ is subadditive with respect to resolution derivation, namely, if $C$ is resolved from $A$ and $B$, then $\cri(C)\leq \cri(A)+\cri(B)$;\label{item: property III of cri}
        \item if $C$ is obtained from a weakening of $A$, then $\cri(C)\leq \cri(A)$. \label{item: property IV of cri}
    \end{enumerate}
    We first show that $\cri(\cdot)$ can be computed in polynomial time. Then we show that any clause $C_i$ such that $n/3\leq \cri(C_i)\leq 2n/3$ will give us a solution. The $\PLS$-membership follows from that the standard $1/3$-$2/3$ trick can be implemented via a reduction to reversed $\Iter$.
    
    \begin{lemma}\label{lem: cri EPHP efficient}
        For any clause $C$, $\cri(C)$ can be computed in $\poly(n)$ time.
    \end{lemma}
    \begin{claimproof}
        Fix any clause $C$. We will enumerate $\ell$ and check the existence of $\ell$-critical assignments. The only part that we need to be careful is how we deal with the extension variables. 

        Imagine that we maintain a complete bipartite graph with $n+1$ nodes on the left and $n$ nodes on the right. We will iteratively delete some edges from this graph based on the requirement of $\ell$-critical assignments. We will show that the existence of $\ell$-critical assignment can be reduced to the existence of a perfect matching of the final graph.

        For an $\ell$-critical assignment, $EP_{\ell}$ needs to be $0$, and all other functions need to be $1$, then we have that $x_{\ell,j}$ needs to be $0$ for all $j$ (which means pigeon $\ell$ cannot go into any hole). This is because if pigeon $\ell$ were matched with some hole, some other pigeons would have no place. So we delete edges between $(\ell,j)$ for all $j$.

        For some $\alpha$ being an $\ell$-critical assignment, we have $C(\alpha)=0$. This means all literals that appeared in $C$ are fixed to be $0$ in the search of $\alpha$. If $C$ contains a literal $\overline{x}_{\ell,j}$ for some $j$, then we directly conclude that $\ell$-critical assignment does not exist (due to the argument above).
    
        Now assume that $C$ does not contain any literal $\overline{x}_{\ell,j}$. For every literal $x_{i,j}$ in $C$, in order to falsify $C$, $x_{i,j}$ is going to be $0$, so we delete the edge $(i,j)$. For every literal $\overline{x}_{i,j}$ in $C$, $x_{i,j}$ is going to be 1, meaning that pigeon $i$ is going to be matched with hole $j$, so we delete the edge $(i,j')$ for every $j'\neq j$ and $(i',j)$ for every $i'\neq i$.

        For every literal $y_{i,j}$ in $C$, $y_{i,j}$ is going to be $0$, so we delete edges $(i,j')$ for all $j'\leq j$. For every literal $\overline{y}_{i,j}$ in $C$, $y_{i,j}$ is going to be $1$, so we delete edges $(i,j')$ for all $j'>j$. 

        We can conclude that there is an $\ell$-critical assignment if and only if the remaining bipartite graph has a perfect matching between pigeons $[n+1]\setminus\set{\ell}$ and holes $[n]$.
    \end{claimproof}

    \noindent \textbf{Reduction to $\Iter$:} 	
    The reversed $\Iter$ instance is defined by the following function $S:[L] \to [L]$. For every $i\in[L]$:
    \begin{itemize}
        \item If $\cri(C_{i})<\frac{2n}{3}$, then $S(i) = i$.
        \item Otherwise, if $C_i$ is a weakening from $C_j$, then let $S(i)=j$.
        \item Finally, if $C_{i}$ is resolved from $C_{j}$ and $C_{k}$: If $\cri(C_{j}) \ge \cri(C_{k})$, then $S(i) = j$; otherwise $S(i) = k$.
    \end{itemize}

    It is easy to see that this reduction can be implemented in block-depth $3$: for example, if $C_i$ is resolved from $C_j$ and $C_k$, then one only needs to read the $i$-th, $j$-th, and $k$-th node in the resolution refutation.

    Note that when we find any solution $i$ of this reversed $\Iter$ instance, it satisfies $S(i)<i$ and $S(S(i))=S(i)$. This means $\cri(C_i)\geq 2n/3$ but $\cri(C_{S(i)})<2n/3$. Thus we have $\cri(C_{S(i)})\in[n/3,2n/3]$.

    Now it remains for us to show that any $C$ such that $n/3\leq \cri(C)\leq 2n/3$ has width at least $n/3$.
For the sake of contraction, assume that $\width(C)< n/3$. This implies that for at most $n/3$ pigeon $i$, $C$ has some variable related to $i$, namely, some $x_{i,j}$ or $y_{i,j}$. Since $\cri(C)\geq n/3$, we know that there exists $\ell$ such that there is an $\ell$-critical assignment $\alpha$ for $C$ but $C$ has no variables of the form $x_{\ell,j}$ or $y_{\ell,j}$. 

    On the other hand, since $\width(C)< n/3$ and $\cri(C)\leq 2n/3$, we know that there is another index $\ell'$ such that $\ell'$ is not critical to $C$ and $C$ has no variables of the form $x_{\ell',j}$ or $y_{\ell',j}$. Let $k$ be the hole that is matched with $\ell'$ in $\alpha$. Consider the following assignment $\alpha'$: we start from $\alpha' := \alpha$, flip $x_{\ell,k}$ from $0$ to $1$, and flip $x_{\ell',k}$ from $1$ to $0$. We further flip all $y_{\ell,j}$ and $y_{\ell',j}$ correspondingly. We have that $EP_{\ell}(\alpha')=1$ and $EP_{\ell'}(\alpha')=0$. Since $C$ doesn't contain any variables related to pigeons $\ell$ and $\ell'$, $C(\alpha')$ is still $0$. This constructs a witness that $\ell'$ is critical to $C$, a contradiction. This finishes the proof.
\end{proof}

It is easy to see that the above reduction is actually a formalization of the width lower bound in $\T^1_2(\alpha)$. Let $n\in\Log$, $\alpha(\cdot, \cdot)$ be an oracle that encodes a purported resolution proof of $\EPHP_{(n+1)\to n}$ with width less than $n/3$, where the second input is a parameter (non-uniformity); that is, each $z$ corresponds to a resolution proof $\pi_z$ and $\alpha(i, z)$ is the $i$-th bit of $\pi_z$. We can syntactically enforce every clause in the proof to have width at most $n/3$. Let $\mistake_{\EPHP}^{\rm w}(n, \alpha, z, i)$ denote the $\PV(\alpha)$ predicate stating that the $i$-th derivation step in $\pi_z$ is invalid (the superscript ``$\rm w$'' stands for ``width''). Then we have:

\begin{theorem}
    \[\T^1_2(\alpha) \vdash \forall n \in\Log\,\forall z\,\exists i~ \mistake_{\EPHP}^{\rm w}(n, \alpha, z, i).\]
\end{theorem}
\begin{proof}[Proof Sketch]
    Reason in $\T^1_2(\alpha)$. Assuming $\forall i~\lnot\mistake_{\EPHP}^{\rm w}(n, \alpha, z, i)$, we will derive a contradiction. A minor technical issue is that we need a $\PV$-definition of the function $\cri$ such that the properties \ref{item: property I of cri}-\ref{item: property IV of cri} are true, and that any clause $C$ with $n/3 \le \cri(C) \le 2n/3$ has width at least $n/3$. This follows from the formalization of bipartite matching algorithms in $\PV$ \cite{Le-Cook}.\footnote{This annoying detail would disappear if we consider the universal variant $\forall\T^1_2(\alpha)$ since these properties are indeed true universal sentences in the standard $\PV$ model $\mathbb{N}$ and thus are included in the axioms of $\forall\T^1_2(\alpha)$. As pointed out in \cite{Muller21}, it is the provability in \emph{universal} variants of relativized bounded arithmetic that captures reducibility among type-$2$ $\TFNP$ problems.}
    
    Let $C_1, C_2, \dots, C_L$ denote the purported resolution refutation encoded by $\alpha(\cdot, z)$. Using $\Sigma_1^b(\alpha)\text{-}\MIN$ (which is available in $\T^1_2(\alpha)$), there is a smallest integer $i$ such that $\cri(C_i) \ge n/3$. By \ref{item: property II of cri}, $C_i$ cannot be an axiom of $\EPHP_{(n+1)\to n}$. By \ref{item: property IV of cri}, $C_i$ cannot be a weakening of any clause $C_j$ ($j < i$), as this would contradict the minimality of $i$. Hence $C_i$ is resolved from $C_j$ and $C_k$ for some $j, k < i$. By \ref{item: property III of cri}, $\cri(C_i) \le \cri(C_j) + \cri(C_k) \le 2n/3$. Hence the width of $C_i$ is at least $n/3$, a contradiction.
\end{proof}

\paragraph{Monotonized resolution and its width refuter.} The first exponential size lower bound for resolution was proven by Haken~\cite{Haken85} for the pigeonhole principle $\PHP$. Haken used the so-called ``bottleneck counting'' argument and the proof was quite involved. A much simpler proof was found by Beame and Pitassi~\cite{beame1996simplified}, where one of the crucial ingredients is the following lemma.

\begin{lemma}[\cite{beame1996simplified}]\label{claim: width lower bound for Res(PHP)}
	Any resolution refutation of $\PHP_{(n+1)\to n}$ contains a clause $C$ with $w(\mono(C)) \ge 2n^2/9$.
\end{lemma}

Here, for a clause $C$, $\mono(C)$ is the ``monotonized'' version of $C$, which is obtained from $C$ by replacing every negated variable $\overline{x}_{ij}$ with the set of variables $x_{i'j}$ for all $i'\neq i$.\footnote{The notion of $\mono(C)$ is tailored to the hard tautology $\PHP_{(n+1)\to n}$. The proof of \cite{beame1996simplified} only considers assignments $x \in \{0, 1\}^{(n+1)n}$ that defines a bijective mapping from $n$ of the pigeons to all $n$ holes; it is easy to see that every clause $C$ is equivalent to $\mono(C)$ w.r.t.~such ``critical'' assignments $x$.} Given \autoref{claim: width lower bound for Res(PHP)}, we could define another variant of the width refuter problem that concerns the original $\PHP$ tautologies with no extension variables: we wish to find a clause $C$ such that $\Mono(C)$ has a large width, in particular, $w(\Mono(C))\geq 2n^2/9$. The problem is denoted as $\Refuter(w_{\Mono}(\PHP\vdash_{\Res} \bot)<2n^2/9)$ and is defined as follows.

\begin{definition}[$\Refuter(w_{\Mono}(\PHP\vdash_{\Res} \bot))$]\label{def: mono width refuter for PHP}
Consider the tautology $\PHP_{(n+1)\rightarrow n}$ and let $w := 2n^2/9$. The input instance $\Pi$ is a purported resolution proof for $\PHP_{(n+1)\rightarrow n}$ that consists of clauses $C_{-k},\dots,C_{-1},C_0,\dots,C_{L-1}$, where the first $k:=n+1+n^2(n+1)/2$ clauses are axioms from $\PHP_{(n+1)\rightarrow n}$ and the last clause $C_L=\bot$.

A solution of the given instance is one of the following:
\begin{itemize}
	\item an index $i\in[L]$ such that $\Mono(C_i)$ has at least $w$ literals;\footnote{Here, unlike the formalization of width lower bounds in \autoref{def: refutation problem for Res}, it is unclear how to syntactically enforce that \emph{the monotonized version} of every clause has width $<w$. Therefore we include the clauses with large monotone width as solutions.} or
	\item an index $i$ such that $C_i$ is an invalid derivation. 
\end{itemize}
\end{definition}

The width refuter problem of monotonized resolution proof may not seem natural in the first place. However, converting resolution proof into monotonized resolution proof is an elegant ingredient in the \emph{size lower bounds} of $\PHP$ in the proof of Beame and Pitassi~\cite{beame1996simplified}. Ultimately, our main motivation is for the size refuter of $\PHP$, and the $\PLS$-membership of $\Refuter(w_{\Mono}(\PHP\vdash_{\Res} \bot)<2n^2/9)$ is a key step of showing the $\rwPHP(\PLS)$-membership of the size refuter ($\Refuter(s(\PHP\vdash_{\Res} \bot)<1.01^n)$). Indeed, the $\rwPHP(\PLS)$-membership in \autoref{sec: upper bound for Refuter for Resolution} uses the following theorem as a black box:

\begin{restatable}{theorem}{thmrefutermonoPHPPLScomplete}\label{thm: mono PLS-completeness}
    $\Refuter(w_{\Mono}(\PHP_{(n+1)\rightarrow n}\vdash_{\Res} \bot)<2n^2/9)$ is in $\PLS$. In particular, there is a uniform decision tree reduction of block-width $3$ from this refuter problem to $\Iter$.
\end{restatable}

The proof is in fact simpler than that of \autoref{thm: white-box PLS-completeness for width}, and is a proper constructive translation of the proof by Beame and Pitassi~\cite{beame1996simplified}.
\begin{proof}[Proof of \autoref{thm: mono PLS-completeness}]

    We call an assignment $\alpha$ to be $\ell$-critical if $\alpha$ is a perfect matching between pigeons and holes, except pigeon $\ell$ is not going anywhere. Formally, for every $i\neq \ell$ we have $x_{i 1}\vee\cdots\vee x_{i n}$, and for every $i,i',j$ we have $\overline{x}_{ij}\vee \overline{x}_{i'j}$ under $\alpha$.

    Given the definition of $\ell$-critical assignments, we define $\cri(\mono(C))$ as follows:
    \[\cri(\mono(C))\coloneqq |\set{\ell: \exists \ell\text{-critical assignment }\alpha\text{ that falsifies }\mono(C)}|.\]

    Again, $\cri$ has the following four important properties:
    \begin{itemize}
        \item $\cri(\bot)=n+1$;
        \item $\cri(\mono(C))\leq 1$ for all axioms $C$ of $\PHP_{(n+1)\rightarrow n}$;
        \item $\cri$ is subadditive with respect to resolution derivation, namely, if $C$ is derived from $A$ and $B$, then $\cri(\mono(C))\leq \cri(\mono(A))+\cri(\mono(B))$.
        \item If $C$ is obtained from a weakening of $A$, then $\cri(\mono(C))\leq \cri(\mono(A))$. 
    \end{itemize}
    This time, since we are only concerned with the monotonized version of a clause $C$ and there are no extension variables, it is easier to show that $\cri(\mono(C))$ can be computed in polynomial time.

    \begin{claim}
        For any clause $C$, $\cri(\mono(C))$ can be computed in polynomial time.
    \end{claim}
    \begin{claimproof}
        Fix a pigeon $\ell$ and we want to check if there is an $\ell$-critical assignment for $C$. We maintain a complete bipartite graph and delete all edges between $(\ell,j)$ for all hole $j$. If a variable $x_{ij}$ appears in $C$, we delete the edge $(i,j)$. Then $\ell$-critical assignment exists if and only if there is a perfect matching between pigeons $[n+1]\setminus\set{\ell}$ and holes $[n]$.
    \end{claimproof}

    \begin{claim}\label{claim: width lower bound from cri}
	For any clause $C$, the width of $\mono(C)$ is at least $\cri(\mono(C)) \cdot (n - \cri(\mono(C)))$.
    \end{claim}
\begin{claimproof}
    Let $D\coloneqq \mono(C)$. Let $\CriP(D)$ be the critical pigeons to $D$, i.e., the set of pigeons $\ell\in [n+1]$ such that there exists an $\ell$-critical assignment falsifying $D$. Then $\cri(\mono(C)) = |\CriP(D)|$. Let $u_1\in \CriP(D)$ and $u_2\not\in \CriP(D)$. Since $u_1\in \CriP(D)$, there is a $u_1$-critical assignment $\alpha$ that falsifies $D$. Suppose that $u_2$ is mapped to the hole $v_2$ in the assignment $\alpha$. Let $\beta$ denote the assignment obtained from $\alpha$ by mapping $u_1$ into $v_2$ and not mapping $u_2$ anywhere. Then $\beta$ is a $u_2$-critical assignment. Since $u_2\not\in \CriP(D)$, $\beta$ satisfies $D$. However, there is only one variable that appears positively in $\beta$ but negatively in $\alpha$: namely, $x_{u_1, v_2}$. Since $D$ is monotone, the literal $x_{u_1, v_2}$ appears in $D$. Repeating this argument for every $u_1 \in \CriP(D)$ and $u_2\not\in \CriP(D)$, we can see that the width of $D$ is at least $\cri(D)\cdot (n - \cri(D))$.
\end{claimproof}

    Given the lemma above, it remains for us to show that finding a clause $C$ such that $\cri(\mono(C))\in[n/3,2n/3]$ belongs in $\PLS$. This can be implemented by the standard $1/3$-$2/3$ trick in a potential function way, which is exactly the same as the argument used in the proof of \autoref{thm: white-box PLS-completeness for width}.
\end{proof}

\subsection{Refuters for Short Resolution Proofs}
\label{sec: upper bound for Refuter for Resolution}
In this subsection, we investigate \autoref{prob: refuter for PHP informal}, i.e., the refuter for the following classic resolution \emph{size} lower bound:
\begin{theorem}[{\cite{Haken85}}]\label{thm: Res lb for PHP}
	Any resolution refutation of $\PHP_{(n+1)\to n}$ requires at least $2^{\Omega(n)}$ clauses.
\end{theorem}

We show that this refuter problem is in $\rwPHP(\PLS)$. As mentioned in \autoref{sec: our results on size lower bounds}, this is done by carefully following the proofs of \autoref{thm: Res lb for PHP}. We follow the simplified proof by Beame and Pitassi~\cite{beame1996simplified}, which consists of two steps: a (monotone) width lower bound and a random restriction argument. As the required width lower bound was already studied in \autoref{thm: mono PLS-completeness}, we focus on the random restriction argument here.

Let $X := [n+1]$ denote the set of pigeons and $Y := [n]$ denote the set of holes. Recall that the \emph{monotone} version of a clause $C$, denoted as $\mono(C)$, is obtained from $C$ by replacing every negated variable $\overline{x}_{ij}$ with the set of variables $x_{ij'}$ for all $j'\ne j$. %

\begin{definition}\label{def: matching restriction}%
Let $t < n$ be a parameter, $\pi:X \to Y$ be a size-$t$ \emph{matching}, i.e., a partial injective function with $|\Domain(\pi)| = t$. This matching induces a restriction that sets:
\begin{itemize}
	\item the variable $x_{u, \pi(u)}$ to be $1$, for every $u\in\Domain(\pi)$;
	\item the variable $x_{u, v}$ to be $0$, for every $u\in\Domain(\pi)$ and $v\in Y\setminus\{\pi(u)\}$; and
	\item the variable $x_{u', \pi(u)}$ to be $0$, for every $u\in \Domain(\pi)$ and $u'\in X\setminus \{u\}$.
\end{itemize}

\end{definition}

Suppose that $C_0, C_1, \dots, C_{L-1}$ is a resolution refutation of $\PHP_{(n+1) \to n}$. For each clause $C_i$, let $\pi(C_i)$ denote the sub-clause of $C_i$ under the above restriction. That is, if $\pi$ sets some variable in $C_i$ to $1$ then $\pi(C_i) = 1$; otherwise $\pi(C_i)$ is obtained from $C_i$ by removing every variable set to $0$ by $\pi$. Note that the above restriction transforms the unsatisfiable CNF $\PHP_{(n+1)\to n}$ into the unsatisfiable CNF $\PHP_{(n-t+1) \to (n-t)}$. Then, $\pi(C_0), \pi(C_1), \dots, \pi(C_{L-1})$ is a resolution refutation for $\PHP_{(n-t+1) \to (n-t)}$. Now we claim that the width of the resolution refutation becomes small after being restricted by a random matching $\pi$.

\begin{claim}\label{claim: size-width tradeoff for Resolution}
    If a size-$t$ matching $\pi$ is chosen uniformly at random over all possible size-$t$ matchings, then with probability at least $1/2$, it holds that for every clause $i\in [L]$, $w(\mono(\pi(C_i))) < W$, where $W := (n+1)^2(1-(1/2L)^{1/t})$.
\end{claim}
\begin{proof}
	Fix $i\in[L]$, we show that the probability over $\pi$ that $w(\mono(\pi(C_i))) > W$ is at most $1/(2L)$. The claim then follows from a union bound.
	
	Choose the matching $\pi$ round by round. There are $t$ rounds, where in each round, we choose an unmatched $u\in X$ and an unmatched $v\in Y$ uniformly at random and match them. If, for the current partial matching $\pi$, we have $w(\mono(\pi(C_i))) \ge W$, then the probability that $x_{u, v} \in \mono(\pi(C_i))$ is at least $W / (n+1)^2$. If this is the case, then $\mono(\pi(C_i))$ will become $1$ (the always-true clause) after we set $\pi(u) \gets v$, thus it gets ``killed.'' It follows that the probability that $C_i$ never gets killed is at most $(1-W/(n+1)^2)^t \le 1/(2L)$.
\end{proof}

Combining \autoref{claim: width lower bound for Res(PHP)} and \autoref{claim: size-width tradeoff for Resolution}, we obtain the following size lower bound:
\begin{theorem}\label{thm: size lb for resolution}
	Any resolution refutation of $\PHP_{(n+1)\to n}$ requires more than $L := 1.01^n$ clauses.
\end{theorem}
\begin{proof}
	Let $t := n/10$, then $W = (n+1)^2(1-(1/2L)^{1/t}) \le \frac{2}{9}(n-t)^2$. If there is a resolution refutation of $\PHP_{(n+1)\to n}$ of size at most $L$, then by \autoref{claim: size-width tradeoff for Resolution}, there is a resolution refutation of $\PHP_{(n-t+1) \to (n-t)}$ of monotone width at most $\frac{2}{9}(n-t)^2$. This contradicts \autoref{claim: width lower bound for Res(PHP)}.
\end{proof}

To derive an upper bound for the complexity of refuter that corresponds to \autoref{thm: size lb for resolution}, we need a \emph{constructive} version of \autoref{claim: size-width tradeoff for Resolution}. We start by setting up an encoding for the partial matchings and random restrictions that will make it easier to describe our reductions.

\def\Seq{\mathcal{SEQ}}
\def\Bad{\mathcal{BAD}}
\def\seq{\mathsf{seq}}
\def\bad{\mathsf{bad}}

A size-$t$ matching can be described by an edge-sequence $(u_0, v_0), (u_1, v_1), \dots, (u_{t-1}, v_{t-1})$, where for each $j\in[t]$, $u_j \in [n-j+1]$ and $v_j \in [n-j]$. The first edge in this matching connects the $u_0$-th node in $X$ and the $v_0$-th node in $Y$ (the first node is the $0$-th), the second edge connects the $u_1$-th unused node in $X$ (i.e., $u_1$-th node in $X\setminus \{u_1\}$) and the $v_1$-th unused node in $Y$, and so on.\footnote{Note that the edges are \emph{ordered}, hence each matching corresponds to $t!$ different edge-sequences. In what follows, we will talk about edge-sequences instead of matchings.} The space of all possible edge-sequences is denoted by 
\[\Seq := ([n+1]\times [n])\times ([n]\times [n-1])\times \dots \times ([n-t+2]\times [n-t+1]).\]

On the other hand, fix a clause $C_i$ such that $w(\mono(C_i)) \ge W$. Say an edge-sequence $s$ is \emph{bad} for $C_i$ if $w(\mono(s(C_i))) \ge W$, where $\pi_s(C_i)$ is the restriction of $C_i$ under the matching corresponding to $\pi_s$. 
As we argued in \autoref{claim: size-width tradeoff for Resolution}, the number of bad edge-sequences for each $C_i$ is small; we set up another encoding to justify this fact. Any bad edge-sequence can be encoded as a sequence $(e_0, e_1, \dots, e_{t-1})$,\footnote{To avoid confusion, we use ``edge-sequence'' to denote elements in $\Seq$ and ``sequence'' to denote elements in $\Bad$.} where for each $j\in[t]$, $e_j \in [(n-j+1)(n-j) - W]$. The first edge $(u_0, v_0)$ in this matching is the $e_0$-th edge, in the lexicographical order, that is not a literal in $\mono(C_i)$; the second edge $(u_1, v_1)$ is the $e_1$-th edge that still can be chosen (we cannot choose any edge touching either $u_0$ or $v_0$) and is not a literal in the current $\mono(\pi_s(C_i))$; and so on. If $s$ is bad, then $w(\mono(\pi_s(C_i)))$ never goes below $W$, therefore at the $j$-th stage, there are at most $(n-j+1)(n-j) - W$ possible edges to choose. Hence, the space of all possible sequences encoding bad edge-sequences is:
\[\Bad = [(n+1)n-W] \times [n(n-1)-W] \times \dots \times [(n-t+2)(n-t+1)-W].\]

The following calculation corresponds to a \emph{union bound} over all $C_i$ that the number of bad edge-sequences is small:
\begin{align}
	\frac{|\Bad|\cdot L}{|\Seq|} =&\, L\cdot \prod_{j=0}^{t-1}\frac{(n+1-j)(n-j)-W}{(n+1-j)(n-j)}\nonumber\\
	\le&\, L\cdot (1-W/(n+1)^2)^t\nonumber\\
	\le&\, 1/2.\label{eq: union bound over BAD}
\end{align}

Fix a clause $C_i$. For every sequence $b\in\Bad$, let $\seq(C_i, b) \in \Seq$ denote the bad edge-sequence for $C_i$ corresponding to $b$; if $\seq(C_i, b)$ does not exist\footnote{This may happen when, for example, $w(\mono(C_i))$ is much larger than $W$ and $b_0 > (n+1)n-w(\mono(C_i))$.}, then we denote $\seq(C_i, b) = \bot$. Conversely, any $s\in\Seq$ is either bad for $C_i$ or not; if $s$ is bad for $C_i$, then denote $b := \bad(C_i, s)$ as the sequence $b\in\Bad$ corresponding to $s$; otherwise we say $\bad(C_i, s) := \bot$. We need the fact that:
\begin{fact}\label{fact: seq and bad map to each other}
	Let $s\in\Seq$ be bad for the clause $C_i$, then $\seq(C_i, \bad(C_i, s)) = s$.
\end{fact}

Now we are ready to establish the $\rwPHP(\PLS)$ upper bound for the refuter of \autoref{thm: size lb for resolution}.

\begin{theorem}\label{thm: refuting Res lb for PHP is in rwPHP(PLS)}
	There is a uniform decision tree reduction of block-depth $3$ from $\Refuter(\Res(\PHP) > 1.01^n)$ to $\rwPHP(\PLS)$.
	
	That is, there is a uniform decision tree reduction of block-depth $3$ such that the following holds:
	\begin{itemize}
		\item given a resolution refutation $\Pi = (C_0, C_1, \dots, C_{L-1})$ for $\PHP_{(n+1)\to n}$, where $L \le 1.01^n$, the reduction computes an instance $(f, \{I_y\}, \{g_y\})$ of $\rwPHP(\PLS)$;
		\item given any valid answer for $(f, \{I_y\}, \{g_y\})$, one can compute an invalid derivation $C_i\in\Pi$ in $\poly(n)$ time.
	\end{itemize}
\end{theorem}
\begin{proof}
	Let $M := |\Bad|\cdot L$ and $N := |\Seq|$, then from \autoref{eq: union bound over BAD}, we have $M\le N/2$. We will identify numbers in $[M]$ with pairs $(i, b)$ where $i\in[L]$ and $b\in\Bad$, and identify numbers in $[N]$ with edge-sequences in $\Seq$. The instance $(f, \{I_y\}, \{g_y\})$ is defined as follows:\begin{itemize}
		\item [{($f$)}] For every $x\in [M]$, we interpret $x$ as a pair $(i, b)$ where $i\in[L]$ and $b\in\Bad$. If $\seq(C_i, b) \ne \bot$, then we let $f(x) := \seq(C_i, b)$; otherwise let $f(x) := 0$ (the choice $0$ is arbitrary).
		\item [{($I_y$)}] Fix $y\in [N] = \Seq$. The edge-sequence $y$ defines a size-$t$ partial matching $\pi_y$, which induces a resolution refutation $\Pi|_y = (\pi_y(C_0), \pi_y(C_1), \dots, \pi_y(C_{L-1}))$ of $\PHP_{(n-t+1) \to (n-t)}$. We treat $\Pi|_y$ as a purported resolution refutation with monotone width $< W$; by \autoref{thm: mono PLS-completeness}, the problem of finding an invalid derivation in $\Pi|_y$ reduces to $\Iter$ via a decision tree of block-width $3$. Let $I_y$ be the $\Iter$ instance obtained by this reduction.
		\item [{($g_y$)}] Fix $y\in [N] = \Seq$ and an answer $ans$ of the $\Iter$ instance $I_y$. Given $ans$, we can compute a clause that is either invalid or too fat; we then compute $g_y(ans)$ from this clause.
		
		More precisely, we can compute an integer $i\in[L]$ such that either $\width(\mono(\pi_y(C_i)))\ge W$, or $\pi_y(C_i)$ corresponds to an invalid derivation in $\Pi|_y$. In the second case, $C_i$ is also an invalid derivation in $\Pi$, thus we can set $g_y(ans)$ to be an arbitrary value (say $0$). In the first case, we can set $g_y(ans) := (i, \bad(C_i, y))$.
	\end{itemize}

	The block-depth of the decision trees computing $f$, $\{I_y\}$, and $\{g_y\}$ are $1$, $3$, and $2$ respectively. Clearly, the decision trees are uniform.

    Finally, let $(y, ans)$ be a valid solution for the $\rwPHP(\PLS)$ instance $(f, \{I_y\}, \{g_y\})$ (i.e., $f(g_y(ans)) \ne y$). The edge-sequence $y\in\Seq$ corresponds to a size-$t$ partial matching $\pi_y$ and from $ans$ we can read off a clause $i \in [L]$ such that either (1) $\width(\mono(\pi_y(C_i))) \ge W$ or (2) $\pi_y(C_i)$ is an invalid derivation in $\Pi|_y$. If (1) holds, then
    \[f(g_y(ans)) = f(i, \bad(C_i, y)) = \seq(C_i, \bad(C_i, y)) = y\]
    by \autoref{fact: seq and bad map to each other}. This contradicts $(y, ans)$ being a valid solution for $\rwPHP(\PLS)$. It follows that (2) happens and we have found an invalid derivation (namely $C_i$) in $\Pi$.
\end{proof}

As mentioned in \autoref{sec: our results on size lower bounds}, the above arguments are essentially a formalization of Haken's lower bounds in the theory $\T^1_2(\alpha) + \dwPHP(\PV(\alpha))$; the decision tree reduction in \autoref{thm: refuting Res lb for PHP is in rwPHP(PLS)} follows from the witnessing theorem for $\T^1_2(\alpha) + \dwPHP(\PV(\alpha))$ (see \autoref{sec: rwPHP(PLS) from witnessing}). In what follows, we make this formalization explicit:

\begin{theorem}\label{thm: formalize Haken in T12 + dwPHP(PV)}
	Let $n\in\Log$, $\alpha(\cdot, \cdot)$ be an oracle that encodes a purported length-$1.01^n$ resolution proof of $\PHP_{(n+1) \to n}$, where the second input is a parameter. Let $\mistake_\PHP(n, \alpha, z, i)$ denote the $\PV(\alpha)$ predicate stating that the $i$-th derivation step in $\alpha(\cdot, z)$ is invalid. Then
	\[\T^1_2(\alpha) + \dwPHP(\PV(\alpha)) \vdash \forall n\in\Log\,\forall z\,\exists i\in [1.01^n] ~ \mistake_\PHP(n, \alpha, z, i).\]
\end{theorem}
\begin{proof}[Proof Sketch]
	Reason in $\T^1_2(\alpha) + \dwPHP(\PV(\alpha))$; assuming $\forall i\in [1.01^n]~\lnot\mistake_\PHP(n, \alpha, z, i)$, we will derive a contradiction. We still use $C_i$ to denote the $i$-th clause of the resolution proof, noticing that given $i$ and $z$, $C_i$ can be computed by a $\PV(\alpha)$ function. We also use our previous notation such as $\Bad$ and $\Seq$, and previous parameters $t := n/10$, $L := 1.01^n$, and $W := (n+1)^2(1-(1/2L)^{1/t}) \le \frac{2}{9}(n-t)^2$.

	First, we use $\dwPHP(\PV(\alpha))$ to select a good random restriction under which each $C_i$ becomes a small-width clause. This random restriction will be encoded as an edge-sequence $s\in\Seq$. Consider the function
	\[\overline{\bad_z}(i, b) := \bad(C_i, b),\]
	where $b\in\Bad$ is any sequence encoding a bad edge-sequence. Clearly, $\overline{\bad_z}$ is a function symbol in $\PV(\alpha)$. By $\dwPHP(\PV(\alpha))$, there is an edge-sequence $s\in\Seq$ such that for every $i\in[L]$ and $b\in\Bad$, $\overline{\bad_z}(i, b)\ne s$.
	
	Next, we apply $s$ to each clause $C_i$; denote $\pi_s(C_i)$ the restriction of $C_i$ under the matching corresponding to $s$. By our choice of $s$, for every $i\in[L]$, we have $w(\mono(\pi_s(C_i))) < W \le \frac{2}{9}(n-t)^2$. By \autoref{claim: width lower bound from cri}, we have $\cri(\mono(\pi_s(C_i))) > \frac{2(n-t)}{3}$ or $\cri(\mono(\pi_s(C_i))) < \frac{n-t}{3}$ for every $i\in[L]$. %
	
	Then we invoke \autoref{claim: width lower bound for Res(PHP)} to show that the sequence $\pi_s(C_0), \pi_s(C_1), \dots, \pi_s(C_{L-1})$ is not a valid resolution proof for $\PHP_{(n-t+1)\to (n-t)}$. Note that this is the step where we use the power of $\T^1_2(\alpha)$. Since $\cri(\mono(\pi_s(C_{L-1}))) = \cri(\bot) = n+1$, by $\Sigma_1^b(\alpha)$-$\MIN$ (which is available in $\T^1_2(\alpha)$), there is a smallest integer $i \le L-1$ such that $\cri(\mono(\pi_s(C_i))) > \frac{2(n-t)}{3}$.\begin{itemize}
		\item If $\pi_s(C_i)$ is an axiom, then $\cri(\mono(\pi_s(C_i))) \le 1$, which is a contradiction.
		\item If $\pi_s(C_i)$ is a weakening of $\pi_s(C_j)$ where $j < i$, then $\cri(\mono(\pi_s(C_j))) \le \cri(\mono(\pi_s(C_i)))$, contradicting the minimality of $i$.
		\item If $\pi_s(C_i)$ is a resolution of $\pi_s(C_j)$ and $\pi_s(C_k)$ where $j, k < i$, then
  \[\cri(\mono(\pi_s(C_i))) \le \cri(\mono(\pi_s(C_j))) + \cri(\mono(\pi_s(C_k))).\]
  However, this means either $\cri(\mono(\pi_s(C_j)))$ or $\cri(\mono(\pi_s(C_k)))$ is at least $\frac{n-t}{3}$, a contradiction.
	\end{itemize}

	Now, since $\pi_s(C_0), \pi_s(C_1), \dots, \pi_s(C_{L-1})$ is not a valid resolution proof for $\PHP_{(n-t+1) \to (n-t)}$, we have that $C_0, C_1, \dots, C_{L-1}$ is not a valid resolution proof of $\PHP_{(n+1)\to n}$ either. This finishes the proof.
\end{proof}

\section{Hardness of Refuting Resolution Proofs}
\label{sec: general lower bounds}

In this section, we provide two hardness results for refuter problems: namely, the $\PLS$-hardness for resolution width refuters (\autoref{lemma: width refuter lower bound}) and the $\rwPHP(\PLS)$-hardness for resolution size refuters (\autoref{thm: rwPHP(PLS)-hardness of refuters}). Notably, our hardness results hold for any family of hard tautologies \emph{as long as the lower bounds are true}. This means that ``$\PLS$-reasoning'' is necessary for proving \emph{any} non-trivial resolution width lower bounds and ``$\rwPHP(\PLS)$-reasoning'' is necessary for proving \emph{any} non-trivial resolution size lower bounds. Of course, this also implies that $\PLS$ and $\rwPHP(\PLS)$ are the \emph{best possible} upper bounds for the refuter problems for \emph{any} non-trivial unsatisfiable family of CNFs (as what we obtained for $\PHP_{(n+1)\rightarrow n}$ and more upper bounds in \autoref{sec: more upper bounds}). %

\subsection{Hardness of Refuting Narrow Resolution Proofs}
\label{section: PLS hard width}

We show that the refuter problems for any \emph{true} resolution width lower bounds are $\PLS$-hard. In fact, this holds even for unsatisfiable CNF families that only contain a \emph{single} CNF of \emph{constant} size:

\begin{restatable}{theorem}{lemmawidthrefuterlowerbound}\label{lemma: width refuter lower bound}
Let $F$ be any unsatisfiable CNF with a non-trivial resolution width lower bound, i.e., $w(F\vdash_{\Res} \bot)> \width(F)$. Let $\calF:=\set{F}$ and $w_0:=\width(F)$. Then there is a (uniform) decision-tree reduction of block-depth $2$ from $\Iter$ to $\Refuter(w(F\vdash_{\Res} \bot) \le w_0)$.
\end{restatable}
\begin{proof}
    We show a straightforward reduction from the reversed $\Iter$ to the width refuter. 

    Note that $F$ is a fixed CNF so it can be seen as constant size. Hence we can check in constant time that $F$ is unsatisfiable and $\width(F\vdash_{\Res} \bot)> \width(F)$.
	
	Let $S:[L]\rightarrow[L]$ such that $S(L-1)<L-1$ be any instance of reversed $\Iter$. We will construct a purported resolution refutation $\Pi$ for $F$ such that any invalid derivation in $\Pi$ corresponds to an answer for $S$. Let $k$ be the number of axioms in $F$. The resolution refutation $\Pi$ consists of nodes $C_{-k}, \dots, C_{-1}, C_0, \dots, C_{L-1}$, where $C_{-k}, \dots, C_{-1}$ are the axioms of $F$.

	For every $i$ such that $S(i)=i$, we let $C_i=C_{-k}$ and define $C_i$ to be a \underline{weakening} from $C_{-k}$. This is a valid derivation (and $C_i$ will not be used anymore). The clauses written in all other nodes in $C_1,\dots,C_{L-1}$ will be $\bot$. The weakening rules applied among these nodes will encode the successor pointer $S$:\begin{itemize}
		\item For every solution $i$ of the reversed $\Iter$ instance (i.e., $i$ such that either $S(i) > i$ \emph{or} ($S(i) < i$ and $S(S(i)) = S(i)$)), the weakening rule applied for $C_{i}$ will be \emph{invalid}. More specifically, let $C_{i}$ be a \underline{weakening} from $C_{-k}$, then $C_{i}$ becomes a solution of the $\Refuter(w(F\vdash_{\Res} \bot)\le w_0)$ instance.
		\item For every $i$ such that $S(i)<i$ and $S(S(i))<S(i)$, we let $C_i$ be a \underline{weakening} from $C_{S(i)}$. Since both $C_i$ and $C_{S(i)}$ are $\bot$, this is a valid derivation.
	\end{itemize}
 
	This finishes the construction, and the correctness follows from the following two facts immediately: (1) there are no nodes whose width is larger than $\width(F)$; (2) a resolution derivation is invalid if and only if it is a solution of the given reversed $\Iter$ instance. The block-depth of our reduction is $2$, as we only need to query $S(i)$ and $S(S(i))$.
\end{proof}

Note that the reduction in \autoref{lemma: width refuter lower bound} also works for a family of CNFs $\{F_n\}_{n\in\N}$ with non-trivial width lower bound. Therefore, combined with the $\PLS$-membership results in \autoref{subsection: width white-box}, we obtain:

\begin{theorem}\label{theorem: EPHP PLS complete}
    $\Refuter(w(\EPHP \vdash_{\Res}\bot)< n/3)$ is $\PLS$-complete.
\end{theorem}

Note that every clause generated in the reduction in \autoref{lemma: width refuter lower bound} has monotone width $O(n)$. Hence the same proof also shows the $\PLS$-hardness of \emph{monotone} width refuters (as in \autoref{thm: mono PLS-completeness}):
\begin{theorem}\label{theorem: monoPHP PLS complete}
    $\Refuter(w_{\Mono}(\PHP_{(n+1)\rightarrow n}\vdash_{\Res} \bot)<2n^2/9)$ is $\PLS$-complete.
\end{theorem}

\subsection{Hardness of Refuting Short Resolution Proofs}\label{sec: rwPHP(PLS)-hardness}

In this section, we show the $\rwPHP(\PLS)$-hardness of refuters for resolution \emph{size} lower bounds. In particular, for any family of unsatisfiable CNF formulas $\{F_n\}_{n\in \N}$ that requires resolution size $>s_F(n)$, if $s_F(n)$ is not too small, then $\rwPHP(\PLS)$ reduces to the problem $\Refuter(s(F_n\vdash_\Res\bot) \le s_F(n))$. This result and our $\rwPHP(\PLS)$ upper bounds (Theorems~\ref{thm: refuting Res lb for PHP is in rwPHP(PLS)}, \ref{thm: XOR-lifting is in rwPHP(PLS)}, \ref{thm: Res size lower bound for Tseitin}, and \ref{thm: refuters for random k-CNF lower bounds}) complement each other by showing that $\rwPHP(\PLS)$ is the tightest complexity class in all these results.

Recall that an $\rwPHP(\PLS)$ instance consists of $\mleft(f,\set{I_y}_{y\in[2M]},\set{g_y}_{y\in[2M]}\mright)$, where:
\begin{itemize}
	\item $f:[M]\to [2M]$ is a purported ``surjection'';
	\item for each $y\in[2M]$, $I_y := (L,S_y)$ is an instance of $\Iter$, where $S_y:[L]\to [L]$; and
	\item $g_y:[L]\rightarrow [M]$ maps solutions of $I_y$ to integers in $[M]$.
\end{itemize}

We now state and prove the main theorem of this subsection.

\begin{figure}[!th]
\begin{subfigure}{\textwidth}
    \centering
    \includegraphics[width=0.65\textwidth]{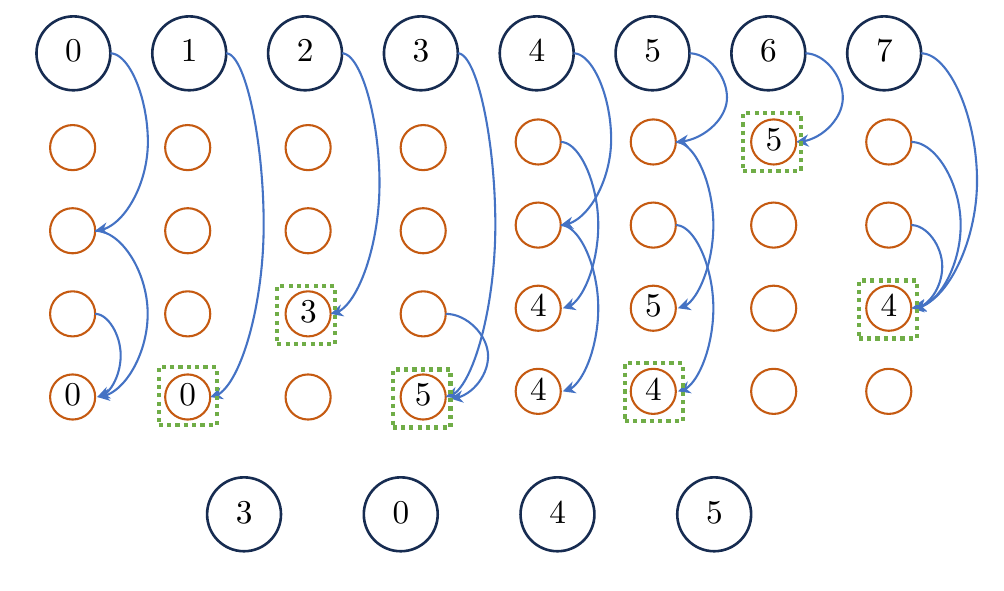}
    \caption{This is an $\rwPHP(\PLS)$ instance with $M=4$, compressing the top $2M$ elements to the bottom $M$ elements. The bottom four points represent the function $f: [M] \to [2M]$; i.e., in this example, $f(0)=3,f(1)=0,f(2)=4,f(3)=5$. Every column is a $\PLS$ instance (every vertex without an outgoing edge has a self-loop). Every sink is a solution of the corresponding $\PLS$ instance, on which we have $g_y:[L]\mapsto [M]$. The number on every sink represents $f(g_y(\cdot))$, which if different from $y$, would be a solution of the whole $\rwPHP(\PLS)$ instance. In this figure, every solution is marked with a dotted green box. \label{subfigure: monotone function} }
\end{subfigure}
\begin{subfigure}{\textwidth}
    \centering
    \includegraphics[width=0.7\textwidth]{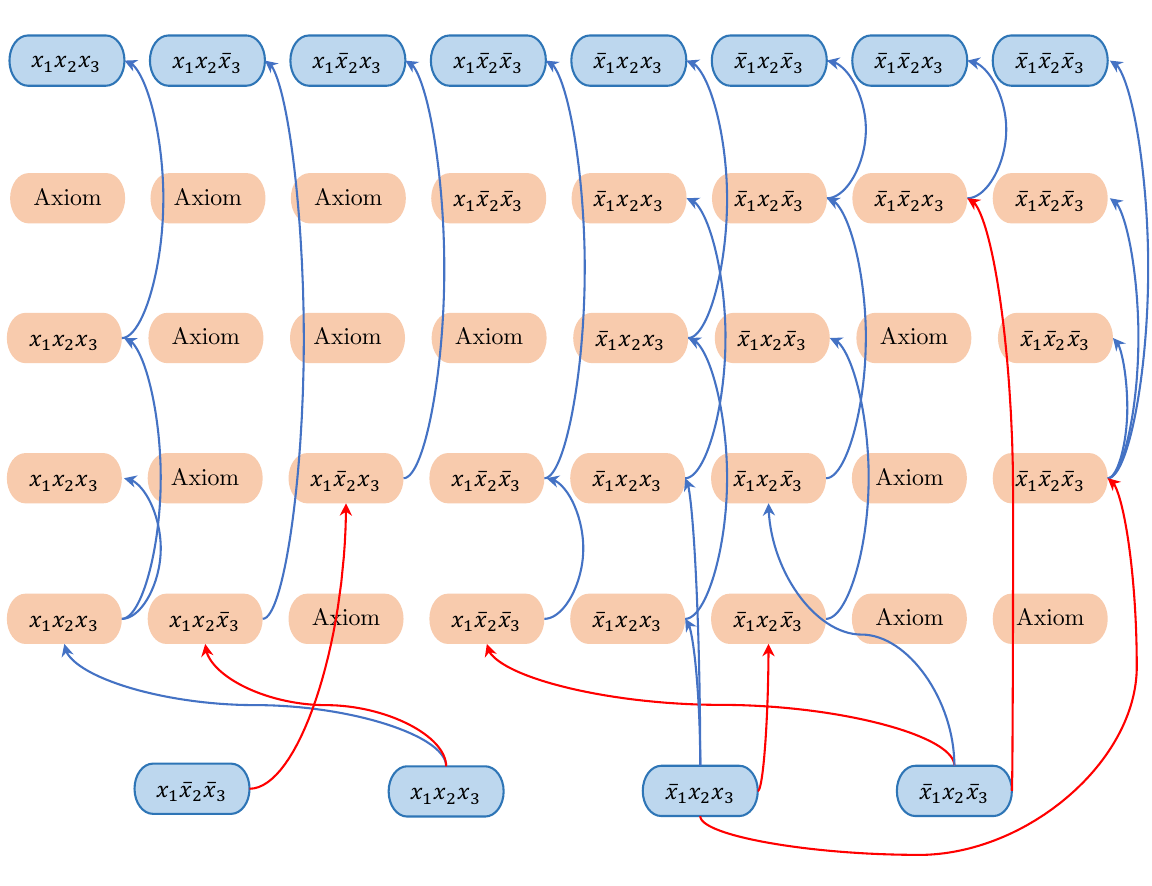}
    \caption{Part of the constructed resolution derivation $\Pi$. Initially, we have $2M=8$ clauses. At the bottom, we have $M=4$ clauses, which exactly correspond to the $3$rd, $0$th, $4$th, and $5$th clauses above. For every node $a$ in a $\PLS$ instance, if it is a self-loop, then we let the clause be a \underline{weakening} from some axiom (and it would never be used again). If $a$ is a solution of the $\PLS$ instance, we let it be the \underline{weakening} of clause $g_y(a)$ at the bottom. Otherwise, we let it be the \underline{weakening} of $S_y(a)$. All blue arrows are valid weakenings and \emph{all red arrows are invalid weakenings}. The invalid weakenings here will be the (only) solutions to the refuter problem.}
\end{subfigure}

\caption{The gadget to embed an $\rwPHP(\PLS)$ instance.}
\end{figure}

\begin{figure}[!ht]
    \centering
    \includegraphics[width=0.9\linewidth]{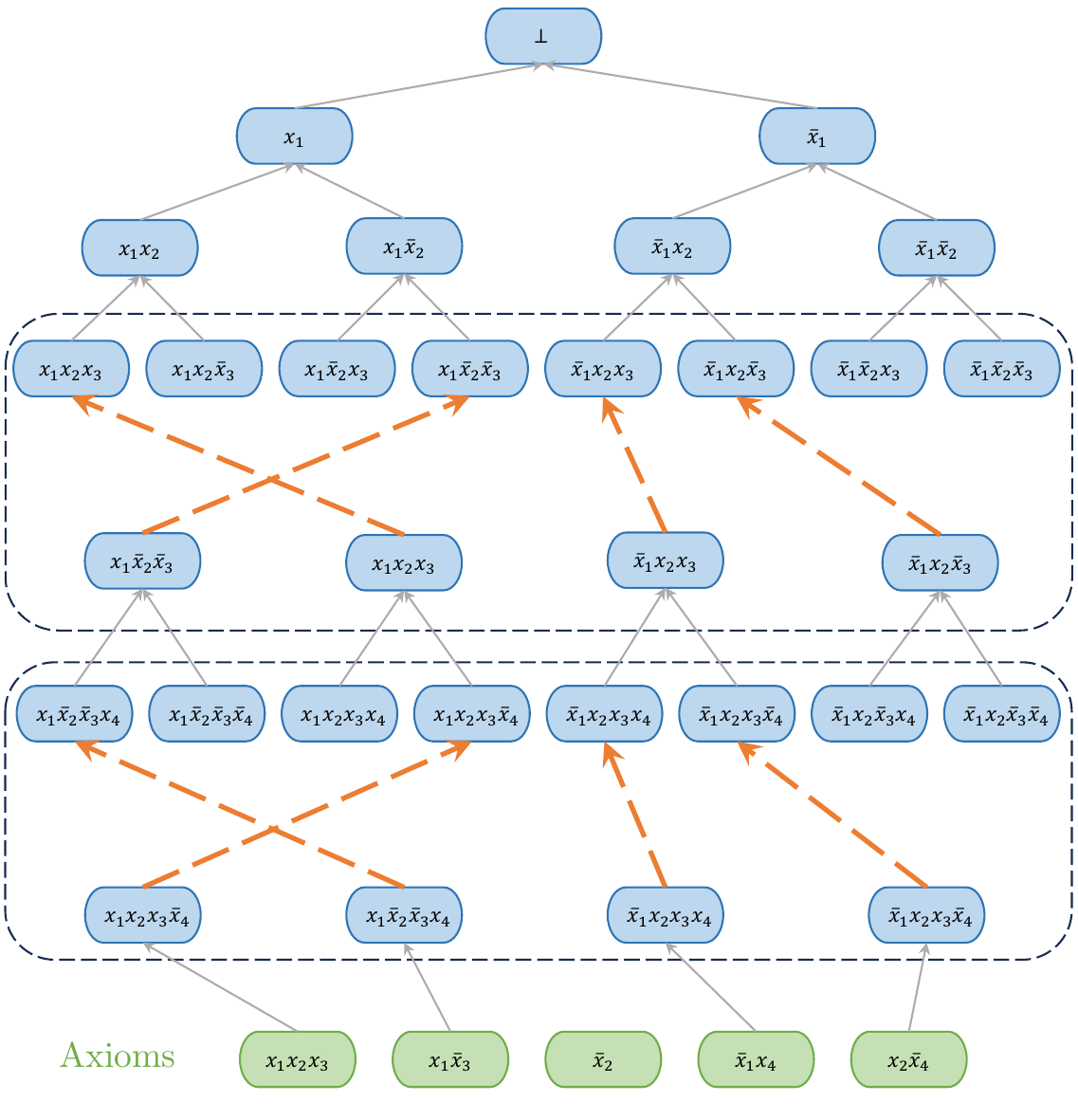}
    \caption{An illustration of our reduction from $\rwPHP(\PLS)$ to the size refuter problem. All gray arrows are valid resolution derivations (and the last layer is weakening from axioms). Every dashed box uses the gadget to embed an $\rwPHP(\PLS)$ instance that enforces every layer to have at most $2M$ clauses. Thus the only possible invalid derivations are those inside the gadget which, once found, would directly imply a solution of the original $\rwPHP(\PLS)$ instance. The overall reduction will produce a purported resolution refutation of size $O(nLM+|F|)$.}
\end{figure}

\begin{restatable}{theorem}{thmrwPHPPLShardnessofrefuters}\label{thm: rwPHP(PLS)-hardness of refuters}
	There is a universal constant $C\ge 2$ such that the following holds. Let $L, M\geq 1$ be the parameters of $\rwPHP(\PLS)$ instances and $n\geq 1$.
 
	For every unsatisfiable CNF formula $F$ over $n$ variables and parameter $s_F \ge C\cdot (nLM + |F|)$ such that every resolution refutation of $F$ requires more than $s_F$ clauses, there is a decision tree reduction of block-depth $O(n)$ from a $\rwPHP(\PLS)$ instance to a $\Refuter(s(F\vdash_\Res\bot) \le s_F)$ instance.
\end{restatable}
\begin{proof}
	Let $\mleft(f,\set{I_y}_{y\in[2M]},\set{g_y}_{y\in[2M]}\mright)$ be an instance of $\rwPHP(\PLS)$ and we will reduce it to an instance of $\Refuter(s(F\vdash_\Res\bot) \le s_F)$. Our goal is to construct a size-$s_F$ resolution refutation $\Pi$ for $F$ such that any illegal derivation in $\Pi$ corresponds to a valid solution to the $\rwPHP(\PLS)$ instance.
	
	The nodes in $\Pi$ are partitioned into $n+1$ layers, numbered from layer $0$ to layer $n$. Each layer $t\in[n+1]$ has either two rows or one row: When layer $t$ has two rows, we denote the nodes in the first row by $\set{D^t_{(y,a)}}$ and those in the second row by $\set{E^t_i}$; when layer $t$ has only one row, we denote the nodes by $\set{E^t_i}$. (Therefore, $\set{E^t_i}$ always denote the \emph{last} row of layer $t$.) %
 After all these $n+1$ layers of clauses, we put the axioms of $F$ at the very end. It is easy to translate a resolution refutation in this layout into one in the format of \autoref{def: Resolution refutation} by decision trees of block-depth $1$.

  \paragraph{The construction.} The layer $0$ has one node $E_0^0:=\bot$. For each $t$ from $1,\dots,n$:
	\begin{enumerate}
		\item Let $E^{t-1}_0,\dots,E^{t-1}_{k-1}$ be the nodes on the last row of layer $t-1$; we will always guarantee that $k \le M$.
		\item \textbf{Case 1: $2k\le M$.} In this case, layer $t$ will only have one row of nodes, defined as follows. For every node $E^{t-1}_i$ on layer $t-1$, we generate $2$ nodes $E^t_{2i}$ and $E^t_{2i+1}$ on layer $t$, where the clauses written are $E^t_{2i}=E^{t-1}_i\vee x_t$ and $E^t_{2i+1}=E^{t-1}_i\vee \overline{x}_t$. We define $E^{t-1}_i$ to be \underline{resolved} from $E^t_{2i}$ and $E^t_{2i+1}$.
		\item \textbf{Case 2: $2k > M$.}\label{item: case 2}
		\begin{enumerate}
			\item First, prepare $2M$ nodes $D^t_{(0, 0)}, D^t_{(1, 0)}, \dots,D^t_{(2M-1, 0)}$. It would be instructive to think of $\{D^t_{(y, a)}: a \in [L]\}$ for each fixed $y$ as a chain and we are now preparing the heads of these $2M$ chains. In what follows, we denote $C_i = D^t_{(i, 0)}$ for ease of notation.
   
            For each $i\in[k]$, let $C_{2i}=E^{t-1}_i\vee x_t$, $C_{2i+1}=E^{t-1}_i\vee \overline{x}_t$, and define $E^{t-1}_i$ to be \underline{resolved} from $C_{2i}$ and $C_{2i+1}$. We make sure that there are \emph{exactly} $2M$ clauses on the first row by making several copies of $C_{2k-1}$: for each $i \in \set{2k, \dots, 2M-1}$, let $C_{i}=C_{2k-1}$.
			\item Generate $M$ nodes on the second row of layer $t$: for every $i\in[M]$, let $E_i^t:=C_{f(i)}$. (Intuitively, if $f:[M]\to[2M]$ \emph{were} a surjection, then every node in $\{C_i\}_{i\in [2M]}$ \emph{would} appear in $\{E_i^t\}_{i\in [M]}$.)
			\item Now, for each $y\in[2M]$, we ``link'' the node $C_y = D^t_{(y, 0)}$ to their corresponding $E_{f^{-1}(y)}^t$ on the second row, using the $\Iter$ instance $I_y$. Recall that for each $y\in[2M]$, $I_y$ consists of a function $S_y: [L] \to [L]$, and solutions of $I_y$ are those $a\in[L]$ such that\label{item: case 2 c}
			\begin{center}
				either $S_y(a)<a$ or ($S_y(a)>a$ and $S_y(S_y(a))=S_y(a)$).
			\end{center}
            As a special case, if $a=0$ and $S_y(0)=0$, then $0$ also counts as a solution.
			
            Every clause on the chain $\{D_{(y, a)}^t: a \in [L]\}$ will be equal to $C_y = D_{(y, 0)}^t$ except those on a ``junk'' node $a$ such that $S_y(a)=a$ (see \hyperref[item: case 2 c ii inactive node]{Case 2 (c) ii.}); these clauses will be weakenings of each other, and the instance $I_y$ dictates the structure of the weakening relationship. For every $a\in [L]$, the node $D_{(y, a)}^t$ is defined as follows:     %
			\begin{enumerate}
                \item If $a$ is a solution of $I_y$, then $g_y(a)$ is a purported pre-image of $y$. The clause written on $D_{(y,a)}^t$ is equal to $C_y$ and we define it to be a \underline{weakening} of $E_{g_y(a)}^t$. (Note that $E_{g_y(a)}^t = C_{f(g_y(a))}$ by definition, meaning that if the weakening from $E_{g_y(a)}^t$ to $D_{(y, a)}^t$ is an illegal derivation, then $f(g_y(a)) \ne y$.)\label{item: case 2 c i}
                \item If $a$ is not a solution and $S_y(a)=a$, then the clause written on $D_{(y,a)}^t$ is defined to be the \underline{weakening} of an arbitrary axiom in $F$ (say the first axiom). The node $D_{(y, a)}^t$ is considered a ``junk'' node and will never be used later. \label{item: case 2 c ii inactive node}
				\item Otherwise, we have $S_y(a)>a$. The clause written on $D_{(y, a)}^t$ is equal to $C_y$ and we define it to be a \underline{weakening} of $D_{(y,S_y(a))}^t$.
			\end{enumerate}
		\end{enumerate}
	\end{enumerate}
	
	After constructing all these nodes above, we put the axioms of $F$ at the very end. Each clause $E_i^n$ in the last row of layer $n$ will be a weakening of some axiom in $F$. In particular, note that each $E_i^n$ has exactly $n$ literals (this can be easily seen from induction) and therefore is satisfied by exactly one assignment $\alpha_i$. Recall that $F$ is an unsatisfiable CNF formula, so for each $E_i^n$, there exists an axiom $A$ in $F$ such that $\alpha_i$ falsifies $A$; hence we can define $E_i^n$ to be a \underline{weakening} of $A$.
	
	This finishes the construction. 

	The above construction gives a resolution refutation $\Pi$ for $F$ that has size $s_F := O(nLM + |F|)$. The only place in $\Pi$ where illegal derivations might occur is in \hyperref[item: case 2 c i]{Case 2 (c) i.}~when we define $D^t_{(y, a)}$ to be a \underline{weakening} of $E^t_{g_y(a)}$. If this is an illegal derivation, then $f(g_y(a)) \ne y$, which means that we have found a valid solution for the $\rwPHP(\PLS)$ instance. Therefore, the above construction is a correct reduction from $\rwPHP(\PLS)$ to $\Refuter(s(F\vdash_{\Res}\bot) > s_F)$, as long as $s_F > C\cdot (nLM + |F|)$ for some large universal constant $C$.
	
	Finally, we analyze the query complexity of this reduction. It suffices to show that every node $D^t_{(y, a)}$ and $E^t_i$ can be computed in block-depth $O(n)$ from the input $\rwPHP(\PLS)$ instance. Note that to compute one node, we need to calculate both its origin (i.e., resolved or weakening from which node) and the clause written on it. We use induction on $t$ to show that every clause in layer $t$ can be computed in block-depth $c\cdot (t+1)$ for some universal constant $c$. Fix a layer $t$ and we argue as follows.
    \begin{itemize}[align=left]
		\item {\bf (Base case)} If layer $t$ contains only one row, then we can read off the clause $E^t_i$ from the binary representation of $i$; the node $E^t_i$ is always \underline{resolved} from $E^{t+1}_{2i}$ and $E^{t+1}_{2i+1}$ (if layer $t+1$ also contains only one row) or $C^{t+1}_{2i}$ and $C^{t+1}_{2i+1}$ (otherwise).
		\item {\bf (Induction step)} If layer $t$ contains two rows, then we argue as follows.\begin{enumerate}
			\item For $i < 2k$, depending on the parity of $i$, we have that the clause written on $D^t_{(i, 0)}$ is either $E^{t-1}_{\lfloor i/2\rfloor}\lor x_t$ or $E^{t-1}_{\lfloor i/2\rfloor}\lor \overline{x}_t$. For $i \ge 2k$, the clause written on $D^t_{(i, 0)}$ is always equal to $D^t_{(2k-1, 0)}$. For every $i \in [2M]$ and $a \in [L]$, the clause written on $D^t_{(i, a)}$ is either equal to the clause written on $D^t_{(i, 0)}$, or equal to some axiom of $F$, and this can be decided in block-depth $2$ (see \hyperref[item: case 2 c ii inactive node]{Case 2 (c) ii.}). Since it takes block-depth $ct$ to compute $E^{t-1}_{\lfloor i/2\rfloor}$, it takes block-depth $ct+2$ to compute the clause written on $D^t_{(i, a)}$.
			\item Every $D^t_{(y, a)}$ (for $y \in [2M]$ and $a\in [L]$) belongs to one of the following three cases:
            \begin{itemize}
                \item if $a$ is a solution of $I_y$, then $D^t_{(y, a)}$ is a \underline{weakening} of $E^t_{g_y(a)}$;
                \item if $a\ne 0$ and $S_y(a) = a$, then $D^t_{(y, a)}$ is a \underline{weakening} of some axiom in $F$ and is a ``junk'' node;
                \item otherwise, $D^t_{(y, a)}$ is a \underline{weakening} of $D^t_{(y, S_y(a))}$.
            \end{itemize}
            Therefore, we can use $O(1)$ additional block-depth to determine all information regarding $D^t_{(y, a)}$.

            \item Let $i\in [M]$, then $E^t_i = D^t_{(f(i), 0)}$, and $E^t_i$ is either \underline{resolved} from $C^{t+1}_{2i}$ and $C^{t+1}_{2i+1}$ (when $t < n$) or a \underline{weakening} of some axiom (when $t = n$). This can be computed in constant additional block depth.
		\end{enumerate}
	\end{itemize}

    It follows that $\Pi$ can be computed from our input $\rwPHP(\PLS)$ instance in block-depth $O(n)$.
\end{proof}

\begin{corollary}\label{cor: main result}
	$\Refuter(s(\PHP_{(n+1)\to n}\vdash_{\Res}\bot) \leq 1.01^n)$ is complete for $\rwPHP(\PLS)$.
\end{corollary}
\begin{proof}[Proof Sketch]
    By combining \autoref{thm: refuting Res lb for PHP is in rwPHP(PLS)} and \autoref{thm: rwPHP(PLS)-hardness of refuters}. Note that the reductions have $\poly(n)$ block-depth and each block contains $\poly(n)$ bits, therefore they are polynomial-time (many-one) reductions.
\end{proof}

The above hardness result in $\TFNP$ can be interpreted as a reversal result in bounded reverse mathematics as well. To state this reversal result, we define the following two families of $\forall\Sigma_1^b(\alpha)$-sentences. For $\PV(\alpha)$ function symbols $F, I, G$, let $\rwPHP(\PLS)(F, I, G)$ denote the natural $\forall\Sigma_1^b(\alpha)$-sentence expressing the existence of a solution for the $\rwPHP(\PLS)$-instance defined by $(F, I, G)$:\begin{itemize}
	\item For every auxiliary input $z$ and every $t, L$, there exists $y \in [2t]$ and $ans \in [L]$ such that $ans$ is a $\PLS$ solution for the $\Iter$ instance $I_{z, y}: [L] \to [L]$ and that ($G_{z, y}(ans) > t$ or $F_z(G_{z, y}(ans)) \ne y$).
\end{itemize}
\def\LBPHP{\mathrm{LB}^{\sf Res}_{\sf PHP}}
Similarly, let $\LBPHP(\PV(\alpha))$ denote the family of $\forall\Sigma_1^b(\alpha)$-sentences consisting of
\[\forall n\in\Log\,\forall z\,\exists i\in [1.01^n]~\mistake_\PHP(n, M, z, i)\]
for every $\PV(\alpha)$ function symbol $M(i, z)$ (here $z$ is a parameter). %

\begin{theorem}
	For every $\PV(\alpha)$ function symbols $F, I, G$,
    \[\PV(\alpha) + \LBPHP(\PV(\alpha))\vdash \rwPHP(\PLS)(F, I, G).\]
\end{theorem}
\begin{proof}[Proof Sketch]
    Argue in $\PV(\alpha)$. Let $\Pi$ be the purported resolution proof for PHP as constructed in the proof of \autoref{thm: rwPHP(PLS)-hardness of refuters} from $(F, I, G)$, then $\Pi$ can be expressed as a $\PV(\alpha)$ function symbol (that depends on $F, I$, and $G$). From $\LBPHP(\PV(\alpha))$, we know that there exists an illegal derivation in $\Pi$. This illegal derivation can only occur in \hyperref[item: case 2 c i]{Case 2 (c) i.}, and hence it points to a weakening from some $D^t_{(y, ans)}$ to some $E^t_{g_y(ans)}$. This means the existence of a solution $(y, ans)$ of the $\rwPHP(\PLS)$-instance $(F, I, G)$.
\end{proof}

We remark that like \autoref{thm: rwPHP(PLS)-hardness of refuters}, the proof of the above theorem does not depend on the hard tautology being $\PHP$.

We finish this section by the following nice-looking characterization of $\forall\Sigma_1^b$-consequences (i.e., provably total $\NP$ search problems) of $\T^1_2 + \dwPHP(\PV)$:
\begin{corollary}\label{cor: main reversal result}
	\begin{enumerate}
		\item $\Refuter(s(\PHP_{(n+1)\to n}\vdash_{\Res}\bot) \leq 1.01^n)$ is complete for the class of $\NP$ search problems provably total in $\T^1_2 + \dwPHP(\PV)$.
		\item A $\forall\Sigma_1^b(\alpha)$-sentence is provable in the theory $\T^1_2(\alpha) + \dwPHP(\PV(\alpha))$ if and only if it is provable in the theory $\PV(\alpha) + \LBPHP(\PV(\alpha))$.
	\end{enumerate}
\end{corollary}

\section{Refuters for Other Formulas}
\label{sec: more upper bounds}

This section presents additional upper bounds for the refuter problems associated with resolution lower bounds. We start with a \emph{universal} $\PLS$ upper bound for width refuters, showing that any resolution width lower bound \emph{that is true} can be refuted in non-uniform $\PLS$. Then, we provide further examples of resolution lower bounds proven by ``random restriction + width lower bounds'' and show that the refuter problems for these lower bounds are in $\rwPHP(\PLS)$. In particular, we present the following three classic resolution lower bounds and show that the refuter problems for all of them are in $\rwPHP(\PLS)$: 
\begin{enumerate}
    \item[(a)] size-width tradeoffs from $\XOR$-lifting \cite{DantchevR03, Krajicek11-Ramsey} (\autoref{sec: rwPHP(PLS) for XOR-lifting});
    \item[(b)] exponential size lower bounds for the Tseitin formulas \cite{Urquhart87,Schoning97} (\autoref{sec: more rwPHP(PLS) upper bounds}); and
    \item[(c)] exponential size lower bounds for random $k$-CNFs \cite{ChvatalS88,beame1996simplified} (\autoref{sec: random k CNF}).
\end{enumerate}
We believe that the case of random $k$-CNFs is especially compelling: the \emph{vast majority} of resolution lower bounds have refuters in $\rwPHP(\PLS)$!

\subsection{Universal Refuters for \emph{Every} Narrow Resolution Proof}
\label{subsection: width black-box}

This subsection shows a very general result: For \emph{every} (possibly non-uniform) family of unsatisfiable CNFs $\calF = \{F_n\}$ and \emph{every} sequence of integers $\{w_n\}$, if for every $n\in\N$, $w_n$ is indeed a resolution width lower bound for $F_n$, then the refuter problem corresponding to this width lower bound is in $\PLS$ under non-uniform decision tree reductions. 

We note that such a membership result is \emph{inherently} non-uniform since it is crucial to consider algorithms with unlimited \emph{computational} power. For example, in general, it is not obvious how to decide if $w_n$ is a valid resolution width lower bound for $F_n$ (although it is certainly computable with unlimited computational power). In fact, even checking if $F_n$ is unsatisfiable is itself $\NP$-complete. On the other hand, even though these two tasks are computationally hard, they only require querying at most $\poly(n)$ bits of the given resolution proof. Thus, we can still consider these refuter problems in $\TFNP^\dt$ and study its query complexity in the non-uniform setting. 

\begin{theorem}\label{lemma: width refuter upper bound}
	Let $\calF$ be any (possibly non-uniform) family of unsatisfiable CNFs with polynomially many clauses. Let $w_0=w(\calF\vdash_{\Res}\bot)$. Then there exists a (non-uniform) decision-tree reduction of block-depth $2$ from $\Refuter(w(\calF\vdash_{\Res} \bot)<w_0)$ to $\Iter$.
\end{theorem}
\begin{proof}
	Consider any instance of $\Refuter(w(\calF\vdash_{\Res} \bot)<w_0)$. Recall from \autoref{def: Resolution refutation} that the instance is a purported resolution refutation $\Pi$ that consists of clauses $C_{-k}, \dots, C_{-1}, C_0, \dots, C_{L-1}$ where $C_{-k}, \dots, C_{-1}$ are the axioms of $\calF$ and $C_{L-1} = \bot$. Also, recall that we syntactically ensure the width of $\Pi$ is $<w_0$ by only allocating $w_0-1$ literals for each clause. The key point in the reduction is that, for any clause $C_i$ that is resolved from $C_{j_1}$ and $C_{j_2}$, if $\width(F\vdash_{\Res} C_i)\geq w_0$, then either $\width(F\vdash_{\Res} C_{j_1})\geq w_0$, or $\width(F\vdash_{\Res} C_{j_2})\geq w_0$.
	
	The length of the reduced reversed $\Iter$ instance is exactly $L$. Next, we define the successor pointers: for every $i\in[L]$, let $C_i$ be the $i$-th clause and $C_{j_1}$ and $C_{j_2}$ with $j_1<j_2<i$ be the two clauses from which $C_i$ is resolved, then
	\begin{equation*}
	S(i) := \begin{cases}
		i & \text{if $\width(F\vdash_{\Res} C_i)<w_0$};\\
		j_1 & \text{if $\width(F\vdash_{\Res} C_{j_1})\geq w_0$};\\
		j_2 & \text{otherwise}.
	\end{cases}
	\end{equation*}
	
	Clearly, this is a query-efficient reduction with block-depth $2$. It is not time-efficient because it needs to compute whether $\width(F\vdash_{\Res} C)<w_0$ for some clauses $C$. 
	
	To show correctness, we consider any possible solution of the constructed reversed $\Iter$. For any $i$ such that $S(i)>i$, we have either $j_1 > i$ or $j_2 > i$, which means that $C_i$ is an invalid derivation. Now consider any $i$ such that $S(i) < i$ and $S(S(i)) = S(i)$. Since $S(i) < i$, we have that $\width(F\vdash_{\Res} C_i)\geq w_0$. Since $S(S(i)) = S(i)$, we have both $\width(F\vdash_{\Res} C_{j_1}) < w_0$ and $\width(F\vdash_{\Res} C_{j_2}) < w_0$. Thus, the resolution step from $C_{j_1}$ and $C_{j_2}$ to $C_i$ must be an invalid derivation. This finishes the proof.
\end{proof}

Note that \autoref{lemma: width refuter lower bound} already shows a universal $\PLS$-hardness, which even holds for uniform reduction. Combining the
the $\PLS$-membership  (\autoref{lemma: width refuter upper bound}) above, we have the following corollary. 

\begin{corollary}
\label{thm: black-box model PLS-completeness for width}
    Let $\calF$ be any (possibly non-uniform) family of unsatisfiable CNFs with polynomially many clauses. Let $w_0:=w(\calF\vdash_{\Res}\bot)$. Then $\Refuter(w(\calF\vdash_{\Res} \bot)<w_0)$ is $\PLS$-complete under (non-uniform) decision tree reductions.
\end{corollary}

\subsection{Refuters for \texorpdfstring{$\XOR$}{XOR}-Lifted Lower Bounds}\label{sec: rwPHP(PLS) for XOR-lifting}

We show that for a large family of resolution lower bounds proved by \emph{lifting theorems}, their corresponding refuter problems are in $\rwPHP(\PLS)$.

Given an unsatisfiable CNF $F$ which is hard for a ``weak'' proof system, a \emph{lifting theorem} produces another unsatisfiable CNF $F'$ (typically by composing $F$ with some \emph{gadgets}) that is hard for a ``stronger'' proof system. Lifting is a very influential technique for proving lower bounds in proof complexity, see e.g.~\cite{huynh2012virtue,goos2014communication,de2016limited,GGKS,de2020lifting}. This subsection examines one of the simplest lifting theorems for proving lower bounds for resolution, which originated from the technique of ``relativization'' \cite{DantchevR03, Krajicek11-Ramsey} (see also \cite[Section 13.2]{krajicek_proof_complexity}).

Let $F(z_1, z_2, \dots, z_n)$ be an unsatisfiable CNF. Roughly speaking, the CNF $F\circ\XOR$ is obtained by replacing each variable $z_i$ with $x_i\oplus y_i$, where $x_i$ and $y_i$ are new variables corresponding to $z_i$. More formally, the formula $F\circ\XOR$ takes $2n$ Boolean variables $x_1, x_2, \dots, x_n$ and $y_1, y_2, \dots, y_n$ as inputs. Denoting $z_i^b = z_i$ if $b = 1$ and $\overline{z}_i$ if $b = 0$; each width-$d$ clause
\[z_{i_1}^{b_1} \lor z_{i_2}^{b_2} \lor \dots \lor z_{i_d}^{b_d}\]
becomes a set of $2^d$ width-$2d$ clauses
\[\mleft\{\mleft(x_{i_1}^{r_1\oplus 1} \lor y_{i_1}^{b_1\oplus r_1}\mright) \lor \mleft(x_{i_2}^{r_2\oplus 1} \lor y_{i_2}^{b_2\oplus r_2}\mright) \lor \dots \lor \mleft(x_{i_d}^{r_d\oplus 1} \lor y_{i_d}^{b_d\oplus r_d}\mright)\mright\}_{r_1, r_2, \dots, r_d \in \{0, 1\}}.\]

A classical lifting theorem states that if $F$ requires large resolution \emph{width}, then $F\circ\XOR$ requires large resolution \emph{size}. Here, the ``weak'' proof system is \emph{narrow} resolution and the ``strong'' proof system is \emph{short} resolution. More formally:

\begin{theorem}\label{thm: XOR lifting for Resolution}
    Let $F$ be an unsatisfiable CNF that requires resolution width $\ge w$, then $F\circ\XOR$ requires resolution size $\ge 2^{w/3}$.
\end{theorem}

The classical proof of this theorem goes through a random restriction argument. Let $\Pi$ be a purported resolution proof of $F\circ\XOR$ of length $L < 2^{w/3}$. Consider a random restriction $\rho$ as follows: For each index $i$, with probability $1/2$, we set $\rho_{x_i}=0/1$ uniformly at random and $\rho_{y_i}=*$; otherwise, we set $\rho_{y_i}=0/1$ uniformly at random and $\rho_{x_i}=*$. By the construction above, $\Pi|_\rho$ is a resolution proof of $F$ up to substituting some variables by their negations, for any $\rho$ in the support. Moreover, for any clause $C \in \Pi$ of width at least $t$, $C$ is killed by a random restriction $\rho$ (i.e., $C|_\rho\equiv 1$) w.p.~at least $1-2^{-\Omega(t)}$. By a union bound over all $L<2^{w/3}$ clauses in $\Pi$, it follows that there is a random restriction $\rho$ killing every clause of width $>w$ in $\Pi$. Therefore, $\Pi|_\rho$ is a resolution refutation for $F$, contradicting the width lower bound for $F$.

To obtain a reduction to $\rwPHP(\PLS)$, it would be helpful to rephrase the above proof as a \emph{compression} argument:

\begin{proof}[Proof of \autoref{thm: XOR lifting for Resolution}]
    Let $\calR$ be the space of the random restrictions in the above proof. Each $\rho \in \calR$ can be described in $2n$ bits:\begin{itemize}
        \item For each index $i$, if $\rho_{x_i} = *$, then we write down $0y_i$ (the first bit being $0$ indicates that $x_i$ is set to $*$, and the second bit encodes $y_i$); otherwise we write down $1x_i$.
    \end{itemize}
    We call the above encoding the \emph{standard encoding} of a restriction; this encoding is a bijection between $\calR$ and $\{0, 1\}^{2n}$, showing that $|\calR| = 4^n$.
    
    In contrast, if $C$ is a clause of width $w$ and $\rho\in\calR$ is a restriction such that $C|_\rho \ne 1$, then given $C$, such a $\rho$ can be described in $(\log_2 3)w + 2(n-w) < 2n$ bits. This is because for each literal in $C$ (say $x_i$, $y_i$, $\overline{x}_i$, or $\overline{y}_i$), if this literal is not simplified to $1$, then there are only $3$ possible choices for $(\rho_{x_i}, \rho_{y_i})$; for example, if this literal is $x_i$, then $(\rho_{x_i}, \rho_{y_i})$ might be one of $(0, *)$, $(*, 0)$, or $(*, 1)$, but never $(1, *)$. We call this the \emph{short} encoding of $\rho$ w.r.t.~$C$; note that this encoding only works when $C|_\rho \ne 1$. %
    
    Now, let $\Pi = (C_0, C_1, \dots, C_{L-1})$ be a resolution refutation of $F\circ\XOR$ with $L < 2^{w/3}$ clauses. Let $f:[L]\times [3^w4^{n-w}] \to [4^n]$ be the function that on input $(i, \rho')$, where $i\in[L]$ and $\rho'$ is the short encoding of a restriction w.r.t.~$C_i$, outputs the standard encoding of $\rho'$ in $\{0, 1\}^{2n}$. Since
    \[L\times 3^w4^{n-w} \le 2^{w/3}\cdot 4^n (3/4)^w < 0.99\cdot 4^n\text{ (whenever $w\ge 1$)},\]
    it follows from the \emph{dual weak pigeonhole principle} that there exists a $\rho \in \{0, 1\}^{2n}$ outside the range of $f$. This restriction $\rho$ simplifies $\Pi$ into a width-$w$ resolution proof of $F$.
    
    In conclusion, if there is a resolution refutation of $F\circ\XOR$ with $<2^{w/3}$ clauses, then \emph{by the dual weak pigeonhole principle}, there is a resolution refutation of $F$ with width $<w$, contradicting the assumed hardness of $F$.
\end{proof}

Now we are ready to show the following result: for every unsatisfiable CNF of the form $F\circ\XOR$ whose resolution size lower bound can be derived from \autoref{thm: XOR lifting for Resolution}, the refuter problem for this resolution size lower bound is in $\rwPHP(\calP)$, where $\calP$ corresponds to the refuter problem for the width lower bound for $F$. Since the refuter problem corresponding to \emph{every} resolution width lower bound admits a non-uniform reduction to $\PLS$ (\autoref{thm: black-box model PLS-completeness for width}), the refuter problems corresponding to size lower bounds for $F\circ\XOR$ non-uniformly reduce to $\rwPHP(\PLS)$ as well. Even if we restrict ourselves to uniform reductions, the refuter problems for many interesting width lower bounds reduce to $\PLS$ (such as \autoref{thm: white-box PLS-completeness for width}), thus the refuter problems for size lower bounds for the corresponding lifted CNFs also reduce to $\rwPHP(\PLS)$.
\begin{theorem}\label{thm: XOR-lifting is in rwPHP(PLS)}
    Let $\{F_n\}$ be a family of unsatisfiable CNFs, $w(n)$ be a width lower bound for $F_n$, and $\calP$ denote the problem $\Refuter(w(F_n) > w(n))$. Then there is a decision tree reduction from $\Refuter(s(F_n \circ\XOR) < 2^{w(n) / 3})$ to $\rwPHP(\calP)$ with block-depth $1$.
\end{theorem}
\begin{proof}
    Let $\Pi$ be the input instance of $\Refuter(s(F_n\circ\XOR) < 2^{w(n) / 3})$. That is, $\Pi = (C_0, C_1, \dots, C_{L-1})$ is a purported resolution refutation of $F_n\circ \XOR$ with $L < 2^{w(n)/3}$ clauses, and we want to find an invalid derivation in $\Pi$.

    Let $f:[0.99N] \to [N]$ be the function defined in the proof of \autoref{thm: XOR lifting for Resolution}, where $N := 4^n$. That is, given a pair $(i, \rho')$ where $i\le L$ and $\rho'$ is the short encoding of a restriction w.r.t.~$C_i$, $f(i, \rho')$ is the standard encoding of this restriction. The range of $f$ consists of (the standard encodings of) all \emph{bad} restrictions, i.e., those that do \emph{not} simplify $\Pi$ to a width-$w$ resolution refutation. 

    Given any restriction $\rho \in \{0, 1\}^{2n}$, let $\Pi|_\rho$ denote the restriction of $\Pi$ under $\rho$ where we force every clause to have width at most $w$; $\Pi|_\rho$ is a purported width-$w$ resolution refutation for $F_n$. In particular, for each $(\rho, i)$, the $i$-th clause of $\Pi|_\rho$ is equal to the restriction of $C_i$ under $\rho$, truncated at width $w$. Note that if $F_n$ indeed requires resolution width $>w$, then $\Pi|_\rho$ must be an \emph{invalid} resolution refutation of $F_n$. Suppose that the $i$-th clause in $\Pi|_\rho$ is derived illegally, then it could be for the following two reasons:
    \begin{itemize}
        \item Either the derivation of $C_i$ in $\Pi$ is already illegal;
        \item or the width of $C_i|_\rho$ is actually $>w$ and the $i$-th clause in $\Pi|_\rho$ is illegal because it was truncated.
    \end{itemize}
    Let $g'_{\rho, i}$ denote the short encoding of $\rho$ w.r.t.~$C_i$, and define $g_{\rho, i} := (i, g'_{\rho, i})$. In the second case, we have a clause $C_i$ and a restriction $\rho$ such that $C_i|_\rho \ne 1$ (in fact, the width of $C_i|_\rho$ is large), thus the short encoding makes sense and, indeed, $f(g_{\rho, i}) = \rho$. In any case, if the short encoding does not make sense (i.e., $C_i|_\rho = 1$), we can set $g_{\rho, i}$ arbitrarily.
    
    Now we have all the ingredients needed in our reduction from the problem of finding an invalid derivation in $\Pi$ to $\rwPHP(\calP)$:\begin{enumerate}
        \item a purported ``surjection'' $f:[0.99N] \to [N]$;
        \item a $\calP$ instance $\Pi|_\rho$ for every $\rho\in [N]$;
        \item for every $\rho \in [N]$ and every solution $i$ of $\Pi|_\rho$ (as a $\calP$ instance), a number $g_{\rho, i}$ pointing to a purported pre-image in $f^{-1}(\rho)$.
    \end{enumerate}
    Every entry $f(i, \rho')$, $\Pi|_\rho(i)$, and $g_{\rho, i}$ only depend on $C_i$ and $\rho$, thus are computable by a decision tree of block-depth $1$.

    A solution of the above $\rwPHP(\calP)$ instance consists of a restriction $\rho$ and a solution $i$ of $\Pi|_\rho$ such that $f(g_{\rho, i}) \ne \rho$. In this case, the derivation of $C_i$ in $\Pi$ must be invalid. That is, given a solution of the $\rwPHP(\calP)$ instance, we can find an invalid derivation of $\Pi$ by a decision tree of block-depth $1$.
\end{proof}

\subsection{Refuters for Tseitin Formulas}\label{sec: more rwPHP(PLS) upper bounds}\label{sec: tseitin}

\paragraph{Tseitin formulas.} Let $G = (V, E)$ be a undirected connected graph, where each vertex $v\in V$ is associated with a value $\tau(v) \in \{0, 1\}$, and each edge $e\in E$ is associated with a Boolean variable $x_e$. The goal is to assign values to each $x_e$ so that for each vertex $v\in V$, the XOR of edge labels incident to $v$ is equal to $\tau(v)$; that is,
\begin{equation}
    \bigoplus_{e\sim v}x_e = \tau(v),\label{eq: Tseitin at v}
\end{equation}
where $e\sim v$ denotes that the edge $e$ is incident to the vertex $v$.

We say $\tau$ is an \emph{odd-weighted} function if $\bigoplus_{v\in V}\tau(v) = 1$. It is not hard to see that the above task is impossible if and only if $\tau$ is odd-weighted (\cite[Lemma 4.1]{Urquhart87}).

\begin{definition}\label{def: Tseitin formula}
	The \emph{Tseitin formula} $\Tseitin(G, \tau)$ \cite{tseitin1983complexity} consists of \autoref{eq: Tseitin at v} for every vertex $v$. When $G$ is a $d$-regular graph (i.e., every vertex $v$ is incident to exactly $d$ edges), we can write \autoref{eq: Tseitin at v} as a $d$-CNF with $2^{d-1}$ clauses:
	\begin{equation}
		\bigwedge_{y_1 \oplus y_2 \oplus \dots \oplus y_d \ne \tau(v)} \mleft((x_{e_1}\ne y_1) \lor (x_{e_2}\ne y_2) \lor \dots \lor (x_{e_d} \ne y_d)\mright), \tag{\ref{eq: Tseitin at v}'}
	\end{equation}
	where $e_1, e_2, \dots, e_d$ are edges incident to $v$.
\end{definition}

For every odd-weighted function $\tau$, $\Tseitin(G, \tau)$ is unsatisfiable; when $G$ is an \emph{expander} graph, $\Tseitin(G, \tau)$ becomes hard for resolution.
\begin{definition}\label{def: expander graphs}
	Let $G = (V, E)$ be an undirected graph. For $S, T\subseteq V$, denote $E(S, T)$ as the set of edges in $E$ with one endpoint in $S$ and the other endpoint in $T$. The \emph{expansion} of $G$ is defined as:
	\[e(G) := \min\{|E(S, V\setminus S)|: S\subseteq V, |V|/3 \le |S| \le 2|V|/3\}.\]
\end{definition}

This gives rise to a family of popular hard tautologies in proof complexity. The first exponential resolution lower bound for Tseitin formulas was proved by Urquhart \cite{Urquhart87}; the proof was subsequently simplified by \cite{Schoning97,Ben-SassonW01}. We restate the theorem from~\cite{Ben-SassonW01} below.

\begin{theorem}[{\cite[Theorem 4.4]{Ben-SassonW01}}]\label{thm: Tseitin width lower bound}
    For every undirected connected graph $G$ and odd-weighted function $\tau:V\to \{0, 1\}$, any resolution refutation of $\Tseitin(G, \tau)$ contains a clause $C$ with $w(C) \ge e(G)$.
\end{theorem}

In this paper, we only consider Tseitin formulas on graphs with constant degree $d = O(1)$.

\paragraph{Width Refuters.} Similar to the Pigeonhole Principle, we first study the width refuter for Tseitin formulas.

\begin{definition}\label{def: refutation problem for Tseitin width lower bounds}
    Let $\Refuter(w(\Tseitin\vdash_\Res\bot) < e(G))$ denote the following problem. The input consists of an undirected connected graph $G = (V, E)$ on $n$ vertices with degree $d=O(1)$, an odd-weighted assignment $\tau:V \to \{0, 1\}$, a parameter $e \le |E|$, and a purported resolution refutation $\Pi$ of $\Tseitin(G, \tau)$ with width less than $e$. A valid solution is either of the following:\begin{itemize}
        \item an index $i$ such that the $i$-th node in $\Pi$ is an invalid derivation, or
        \item a vertex set $S\subseteq V$ such that $|V|/3 \le |S| \le 2|V|/3$ and $|E(S, V\setminus S)| < e$.
    \end{itemize}
    (Note: in this $\TFNP^\dt$ problem, we think of $\poly(n)$-time algorithms as ``efficient'', hence an efficient procedure can read the whole graph $G$, verify that $\tau$ is indeed odd-weighted, or count the number of edges between $S$ and $V\setminus S$. When we calculate block-depth, the inputs $(G, \tau, e)$ are treated as a single block.)
\end{definition}

\begin{mdframed}[hidealllines=true,backgroundcolor=gray!10,skipabove=0.1em,skipbelow=-0.4em,innertopmargin=0]
	\small
\begin{remark}\label{remark: formalization of Tseitin}
    This definition is different from most refuter problems considered in this paper, as it is not for a single family of tautology, and it does not even \emph{guarantee} that the tautology is hard! Instead, it asks to find either an invalid derivation in the purported proof or a \emph{certificate} of the tautology being easy (i.e., a sparse cut in the graph).

    \def\Expander{\mathrm{Expander}}
    \def\Correct{\mathrm{Correct}}
    We argue that this is a natural definition. Let $\pf_{\Tseitin}^\alpha(G, \tau, e)$ denote the $\Pi_1^b(\alpha)$-sentence ``$\alpha$ encodes a width-$e$ proof of $\Tseitin(G, \tau)$'' (note that $\alpha$ is treated as an oracle, i.e., a second-order object, while $G$, $\tau$, and $e$ are inputs, i.e., first-order objects). That is,
    \[\pf_{\Tseitin}^\alpha(G, \tau, e) := \forall i~\Correct^\alpha(G, \tau, e, i),\]
    where $\Correct^\alpha(G, \tau, e, i)$ expresses that the $i$-th step of $\alpha$, as a width-$(e-1)$ proof of $\Tseitin(G, \tau)$, is correct. Similarly, let $\Expander(G, e)$ denote the $\Pi_1^b$-sentence that $e(G) \ge e$. That is,
    \[\Expander(G, e) := \forall S\subseteq V~\mleft(|S| \in \mleft[(1/3)|V|, (2/3)|V|\mright] \implies |E[S, V\setminus S]| \ge e\mright).\]
    The proof in \cite{Ben-SassonW01} actually shows that $\Expander(G, e) \implies \lnot \pf_{\Tseitin}^\alpha(G, \tau, e)$, which after rearranging is equivalent to:
    \begin{equation}\label{eq: Tseitin width lower bound}
        \exists i~\lnot \Correct^\alpha(G, \tau, e, i) \lor \exists S~\mleft(|S| \in \mleft[(1/3)|V|, (2/3)|V|\mright] \land |E[S, V\setminus S]| < e\mright).
    \end{equation}
    It is easy to see that \autoref{def: refutation problem for Tseitin width lower bounds} is exactly the $\TFNP^\dt$ problem corresponding to \autoref{eq: Tseitin width lower bound}.
\end{remark}
\end{mdframed}

\begin{theorem}\label{thm: Tseitin width lower bound refuter}
    $\Refuter(w(\Tseitin\vdash_\Res\bot) < e(G))$ is $\PLS$-complete. 
\end{theorem}
\begin{proof}
We will show that there is a (uniform) decision tree reduction of block-width $3$ from this problem to $\Iter$.

Let $G=(V,E)$ be an undirected graph with purported expansion parameter $e$. Let $\tau:V\to \set{0,1}$ be an odd-weighted function. Let $\Pi$ be a purported resolution refutation that consists of clauses $C_{-k},\cdots,C_{-1},C_0,\cdots,C_{L-1}$. where $C_{-k},\cdots,C_{-1}$ are axioms of the unsatisfiable CNF associated with $G$ and $\tau$. Note that we can syntactically require that each $C_i$ has width at most $e-1$. Our goal is to find either an invalid derivation in $\Pi$ or a witness that the expansion of $G$ is, in fact, less than $e$. In particular, the witness is a vertex set $S\subseteq V$ such that $|V|/3 \le |S| \le 2|V|/3$ and $|E(S, V\setminus S)| < e$.

Similar to before, we first introduce a complexity measure for a clause $C$. Let $v\in V$ be a vertex. We say an assignment $\alpha$ is $v$-critical if $\alpha$ only falsifies the constraint associated with $v$ and satisfies all other constraints of the given unsatisfiable CNF. The complexity measure, denoted by $\cri(C)$, is defined as follows.
    $$\cri(C)\coloneqq \mleft|\set{v\in V : \exists v\text{-critical assignment } \alpha \text{  such that } C(\alpha)=0}\mright|. $$ 

    Note that $\cri$ has four important properties:
    \begin{itemize}
        \item $\cri(\bot)=n$;
        \item $\cri(C_i)=1$ for all $-k\leq i\leq -1$, namely, $\cri(C)=1$ for all axioms $C$;
        \item $\cri$ is subadditive with respect to resolution derivation, namely, if $C$ is resolved from $A$ and $B$, then $\cri(C)\leq \cri(A)+\cri(B)$;
        \item if $C$ is obtained from a weakening of $A$, then $\cri(C)\leq \cri(A)$. 
    \end{itemize}
    We first show that $\cri(\cdot)$ can be computed in polynomial time. Then we show that any clause $C_i$ such that $n/3\leq \cri(C_i)\leq 2n/3$ will give us a solution. The $\PLS$-membership follows from that the standard $1/3$-$2/3$ trick can be implemented via a reduction to reversed $\Iter$.

    \begin{lemma}
        For any clause $C$, $\cri(C)$ can be computed in $\poly(n)$ time.
    \end{lemma}
    \begin{claimproof}
    Fix any clause $C$. We will enumerate $v\in V$ and check the existence of $v$-critical assignments. 

    Note that the aimed assignment $\alpha$ needs to satisfy that $C(\alpha)=0$, so all literals in $C$ are fixed. For $\alpha$ being a $v$-critical assignment, the constraint associated with $v$ needs to be falsified. We enumerate an axiom in the constraint associated with $v$. Since $d$ is a constant, there are only $2^{d-1} = O(1)$ axioms that we need to enumerate.
    
    Fix such an axiom, and set all literals in this axiom to be 0 as well (if setting them to be 0 is not consistent with $C(\alpha)=0$, then skip this axiom and try the next one). Now we have fixed some variables and left other variables free. Let $\rho\in\set{0,1,*}^m$ be this partial assignment, where $m=|E|$. Note that $C(\rho)=0$ and the constraint associated with $v$ has also been falsified. So we only need to check if there is a complement $\alpha$ of $\rho$ such that all other constraints can be satisfied by $\alpha$. This reduces to checking whether a system of linear equations over $\mathbb{F}_2$ has a solution, which can be done in polynomial time.
\end{claimproof}

Then we show that finding a clause $C_i$ such that $\cri(C_i)\in[n/3,2n/3]$ can be reduced to $\Iter$.

\noindent \textbf{Reduction to $\Iter$:} 	
The instance of a reversed $\Iter$ is defined by the following function $S:[L] \to [L]$. For every $i\in[L]$:
\begin{itemize}
    \item if $\cri(C_{i})<\frac{2n}{3}$, then $S(i) = i$;
    \item otherwise, if $C_i$ is a weakening from $C_j$, then let $S(i)=j$;
    \item Finally, let $C_{i}$ be resolved from $C_{j}$ and $C_{k}$: If $\cri(C_{j}) \ge \cri(C_{k})$, then $S(i) = j$; otherwise $S(i) = k$.
\end{itemize}

It is easy to see that this reduction can be implemented in block-depth $3$: for example, if $C_i$ is resolved from $C_j$ and $C_k$, then one only needs to read the $i$-th, $j$-th, and $k$-th node in the resolution refutation.

Note that when we find any solution $i$ of this reversed $\Iter$ instance, it satisfies $S(i)<i$ and $S(S(i))=i$. This means $\cri(C_i)\geq 2n/3$ but $\cri(C_{S(i)})<2n/3$. Thus we have $\cri(C_{S(i)})\in[n/3,2n/3]$.

\vspace{0.2cm}
\noindent \textbf{Correctness of the Reduction:} 	
Fix $C$ such that $n/3\leq \cri(C)\leq 2n/3$. Let $$E'=\set{(u,v)\in E\mid u\in \cri(C), v\in V\setminus \cri(C)}.$$

We show that $C$ contains every variable that appears in $E'$. If not, let $e=(u,v)\in E'$ be a missing variable and suppose without loss of generality that $u\in\cri(C)$ and $v\not\in\cri(C)$. Since $u\in\cri(C)$, by definition we know there exists a $u$-critical assignment $\alpha_u$ such that $C(\alpha_u)=0$. Let $\alpha_u'$ be the same assignment but flipping $x_{(u,v)}$. Then by definition, we obtain a new assignment $\alpha_u'$ that is $v$-critical. However, recall that $v\not\in\cri(C)$, which leads to a contradiction.

Thus, suppose that $C$ is not obtained by an invalid derivation, then since $\width(C)<e$, we know that $|E'|<e$, which means that $\cri(C)$ is a witness that the expansion of $G$ is in fact less than $e$.

This finishes the proof.
\end{proof}

\paragraph{Size Refuter.} After the $\PLS$-membership of width refuter, we are ready to study the size refuter.

We consider Tseitin formulas where the underlying graph $G=(V,E)$ is an expander. Recall from \autoref{def: expander graphs} that the expansion of $G$, denoted as $e(G)$, is the minimum number of edges between $S$ and $V\setminus S$ over every subset $S\subseteq V$ such that $|V|/3 \le |S| \le 2|V|/3$. It is proved in \cite{Schoning97,Ben-SassonW01} that for every constant-degree expander $G$ with $e(G) \ge n$ and every odd-weighted function $\tau: V\to \{0, 1\}$, the tautology $\Tseitin(G, \tau)$ requires size-$2^{\Omega(n)}$ resolution proof.

Now, analogous to \autoref{def: refutation problem for Tseitin width lower bounds}, we define the refuter problem for the size lower bounds, where the graph $G$ is also given as an input, and a certificate for $G$ not being an expander is also a valid output:

\begin{definition}\label{def: refuter problem for Tseitin size lower bounds}
	Let $\Refuter(s(\Tseitin \vdash_\Res\bot) < 1.01^{n/d})$ denote the following problem. The input consists of an undirected connected $d$-regular graph $G = (V, E)$ on $n$ vertices, an odd-weighted assignment $\tau: V\to \{0, 1\}$, and a purported resolution refutation $\Pi$ for $\Tseitin(G, \tau)$ that contains at most $1.01^{n/d}$ clauses. A valid solution is either of the following:\begin{itemize}
		\item an index $i$ such that the $i$-th node in $\Pi$ is an invalid derivation, or
		\item a vertex set $S\subseteq V$ such that $|V|/3 \le |S| \le 2|V|/3$ and $|E(S, V\setminus S)| < n$.
	\end{itemize}
	Again, when we calculate the block-depth of reductions, we treat $(G, \tau)$ as one input block.
\end{definition}

\begin{theorem}\label{thm: Res size lower bound for Tseitin}
	Let $G = (V, E)$ be a $d$-regular undirected connected graph and $\tau: V \to \{0, 1\}$ be an odd-weighted function. Then, if $e(G)\ge n$, then $\Tseitin(G, \tau)$ requires resolution size $\ge 1.01^{n/d}$.
	
	Moreover, there is a uniform decision tree reduction from $\Refuter(s(\Tseitin\vdash_\Res\bot) < 1.01^{n/d})$ to $\rwPHP(\PLS)$ with block-depth $3$.
\end{theorem}
\begin{proof}
	We follow the proof in \cite{Schoning97}, which (also) uses a random restriction argument and a width lower bound. Our exposition about the random restrictions will be careful and slow (since this is relevant to our reduction to $\rwPHP(\PLS)$), but we will be sketchy about other parts.

	Consider a random restriction as follows. Let $t := n/10$, pick $t$ edges $E' = \{e_1, e_2, \dots, e_t\}$ uniformly at random, and for each edge $e_i$ assign a uniformly random bit to $x_{e_i}$. For an edge $e = (x, y)$, each time we assign $x_e \gets 0$, we do nothing with the function $\tau$; each time we assign $x_e\gets 1$, we flip both $\tau(x)$ and $\tau(y)$. After picking these $t$ edges, we reduced the formula $\Tseitin(G, \tau)$ to the formula $\Tseitin(G', \tau')$, where $G'$ is the graph $G$ with edges in $E'$ removed, and $\tau'$ is the assignment on vertices we obtained at the end. It is easy to see that $e(G') \ge e(G) - t$, hence by \autoref{thm: Tseitin width lower bound}, any resolution refutation for $\Tseitin(G', \tau')$ requires width $\ge e(G) - t$.
	
	It would be helpful to rigorously define the space of random restrictions. Fix an ordering $\prec$ (e.g., the lexicographic one) over the $nd = 2|E|$ literals. A restriction is described by a sequence $(i_0, i_1, \dots, i_{t-1})$ as follows. We first pick the $i_0$-th literal $\ell_0$ according to $\prec$ and set $\ell_0 := 1$. Now we are left with $nd-2$ literals (as both $\ell_0$ and $\bar{\ell_0}$ are set) and we pick the $i_1$-th literal $\ell_1$ among them, according to $\prec$. After setting $\ell_1 := 1$, we are left with $nd-4$ literals and we pick the $i_2$-th one, and so on. Each sequence corresponds to a restriction that sets the values of $t$ edges (but note that each restriction corresponds to $t!$ such sequences). The space of random restrictions is denoted as 
	\[\calR := [nd] \times [nd-2] \times [nd-4] \times \dots \times [nd-2t+2].\]
	
	Let $w := e(G) - t$ and fix a clause $C$ of width $\ge w$. If we know that a restriction $\rho$ does not kill $C$, then there is a more efficient way to describe $\rho$ by a sequence $(j_0, j_1, \dots, j_{t-1})$, as follows. We first pick the $j_0$-th literal $\ell_0$ \emph{among those $nd-w$ ones not in $C$}, according to $\prec$, and set $\ell_0 := 1$. After this round, there are at most $nd-w-1$ remaining literals not in $C$: If $\bar{\ell_0} \in C$ then there are exactly $nd-w-1$ such literals (i.e., excluding $\ell_0$), otherwise there are $nd-w-2$ remaining literals (i.e., excluding $\ell_0$ and $\bar{\ell}_0$). Anyway, we use $nd-w-1$ as an upper bound on the number of literals that we can choose after the first round. In the next round, we choose the $j_1$-th literal $\ell_1$ not in $C$ according to $\prec$, set $\ell_1 := 1$, and now there remains at most $nd-w-2$ literals. In the next round, we choose the $j_2$-th such literal, and so on. The space of ``bad'' restrictions that do not kill $C$ is
	\[\Bad := [nd-w] \times [nd-w-1] \times \dots \times [nd-w-t+1].\]
	
	Given any $C$ and $b \in \Bad$, we can compute $\seq(C, b)\in\calR$ as the ``$b$-th bad restriction corresponding to $C$''.\footnote{Note that some $b\in\Bad$ might not correspond to a valid restriction. We can set $\seq(C, b)$ to be an arbitrary value.} Given any clause $C$ of width $\ge w$ and any restriction $\rho \in \calR$ that does not kill $C$, we can compute an encoding $\bad(C, \rho) \in \Bad$ such that $\seq(C, \bad(C, \rho)) = \rho$. The following calculation corresponds to a ``union bound'' over $L := 1.01^{n/d}$ clauses in the purported resolution proof:
	\begin{align*}
		L\cdot \frac{|\Bad|}{|\calR|} =&\, L \cdot \prod_{i\in [t]}\frac{nd-e(G)+t-i}{nd-2i}\\
		\le&\, L\cdot \mleft(1-\frac{e(G)-3t}{nd-2t}\mright)^t\\
		\le&\, 1.01^{n/d}\cdot \mleft(1-\frac{7}{10d-2}\mright)^{n/10} \le 1/2.
	\end{align*}

	The lower bound argument proceeds as follows. Let $\Pi = (C_0, C_1, \dots, C_{L-1})$ be a purported size-$L$ resolution proof for $\Tseitin(G, \tau)$. By the above union bound, there is a restriction $\rho \in \calR$ that kills every clause in $\Pi$ with width $\ge w$. This restriction shrinks $\Pi$ into $\Pi|_\rho$ which is a width-$w$ resolution proof for $\Tseitin(G', \tau')$, contradicting the width lower bound. Therefore, every resolution proof for $\Tseitin(G, \tau)$ requires more than $1.01^{n/d}$ many clauses.

	Finally, we describe the reduction from $\Refuter(s(\Tseitin\vdash_\Res\bot) < 1.01^{n/d})$ to $\rwPHP(\PLS)$:\begin{itemize}
		\item [($f$)] The function $f: [L]\times \Bad \to \calR$ is defined as $f(i, b) := \seq(C_i, b)$.
		\item [($I_\rho$)] For every $\rho\in \Seq$, we obtain a purported width-$w$ resolution proof $\Pi|_\rho$ for $\Tseitin(G', \tau')$. Every node in $\Pi|_\rho$ can be computed in block-depth $1$ from $\Pi$. Using \Autoref{thm: Tseitin width lower bound refuter}, we reduce the problem of finding an invalid derivation in $\Pi_\rho$ to an $\Iter$ instance $I_\rho$, where each node in $I_\rho$ is computed in block-depth $3$ from $\Pi|_\rho$.
		\item [($g$)] For every $\rho \in \Seq$ and every valid solution $o$ of $I_\rho$, we can compute an index $i\in [L]$ from $o$ such that the $i$-th step in $\Pi|_\rho$ is an illegal derivation. We let $g_{\rho, o} := (i, \Bad(C_i, \rho))$. 
	\end{itemize}

	Suppose that $(\rho, o)$ is any solution to the $\rwPHP(\PLS)$ instance defined above. Let $i\in [L]$ be computed from $o$ as above, then we claim that the $i$-th step of $\Pi$ must be an illegal derivation. Indeed, since $o$ is a solution of $I_\rho$, the $i$-th step of $\Pi|_\rho$ must be illegal. On the other hand, if the $i$-th step of $\Pi$ is not illegal, then $C_i|_\rho$ is a clause of width $\ge w$, and thus
	\[f(g_{\rho, o}) = \seq(C_i, \Bad(C_i, \rho)) = \rho,\]
	contradicting that $(\rho, o)$ is a valid solution to the reduced $\rwPHP(\PLS)$ instance.
\end{proof}

\subsection{Refuters for Random \texorpdfstring{$k$}{k}-CNFs}
\label{sec: random k CNF}

Finally, we show that resolution lower bounds for random $k$-CNFs can be refuted in $\rwPHP(\PLS)$. More precisely, as in \cite{ChvatalS88}, we consider the distribution $\calF(k, n, m)$ over $k$-CNFs with $n$ variables and $m$ clauses where each clause is i.i.d.~chosen from all $\binom{n}{k}2^k$ ordinary clauses of size $k$ over the $n$ variables. (A clause is \emph{ordinary} if there is no variable $x_i$ such that both $x_i$ and $\bar{x}_i$ occur in this clause.) Let $c \ge 1, \eps > 0$ be constants, and $\{F_n\}_{n\in \N}$ be a family of $k$-CNFs, where each $F_n$ is a $k$-CNF over $n$ variables and $cn$ clauses. In the search problem
\[\Refuter(s(F_n) < (1+\eps)^n),\]
we are given query access to a purported resolution refutation $\Pi$ for $F_n$ that contains at most $(1+\eps)^n$ clauses, and our goal is to locate an invalid derivation in $\Pi$.

\begin{theorem}\label{thm: refuters for random k-CNF lower bounds}
	For every large enough positive integer $k$ and $c \ge 0.7\cdot 2^k$, there is a constant $\eps > 0$ such that the following holds. Let $\{F_n\}_{n\in \N}$ be a sequence of random $k$-CNFs where each $F_n$ is independently chosen according to the distribution $\calF(k, n, cn)$. With probability $1$, there is a non-uniform decision tree reduction of block-depth $2$ from the problem $\Refuter(s(F_n) < (1+\eps)^n)$ to $\rwPHP(\PLS)$ that works for all large enough $n$.
\end{theorem}

The unsatisfiability of $\{F_n\}_{n\in \N}$ and the resolution lower bounds for $\{F_n\}_{n\in \N}$ are already shown in the seminal work of Chv{\'{a}}tal and Szemer{\'{e}}di \cite{ChvatalS88}. We prove \autoref{thm: refuters for random k-CNF lower bounds} by formalizing their resolution lower bound proofs as decision tree reductions to $\rwPHP(\PLS)$.

The reason that our reduction in \autoref{thm: refuters for random k-CNF lower bounds} is non-uniform is very similar to that in \autoref{subsection: width black-box}. First, it appears infeasible to decide if $(1+\eps)^n$ is indeed a valid resolution size lower bound for the input formula $F_n$. Second, the proofs in \cite{ChvatalS88} involve some objects that appear to be infeasible to compute given $F_n$; however, these objects do not depend on the purported size-$(1+\eps)^n$ resolution refutation, thus can be hardwired in a non-uniform decision tree. It might be possible to obtain a ``uniform version'' of \autoref{thm: refuters for random k-CNF lower bounds} like what we did for Tseitin formulas (\autoref{def: refutation problem for Tseitin width lower bounds}, \autoref{remark: formalization of Tseitin}, \autoref{def: refuter problem for Tseitin size lower bounds}), by completely formalizing \cite{ChvatalS88} in bounded arithmetic. We choose not to do so because we believe that a non-uniform upper bound of $\rwPHP(\PLS)$ already supports our claim that $\rwPHP(\PLS)$ captures the complexity of proving \emph{most} resolution lower bounds; dealing with extra details in \cite{ChvatalS88} would only be distracting.

We assume familiarity with the (quite involved) proof in \cite{ChvatalS88}. In particular, we need the following definitions and theorems:
\begin{itemize}
	\item Fix a $k$-CNF $F$ over $n$ variables and $cn$ clauses. Let $X = \{x_1, x_2, \dots, x_n\}$ denote the set of variables of $F$. The ``structure'' of $F$ can be described by a $k$-uniform (multi-)hypergraph $H$ over the vertex set $X$, where each clause $F_i$ of $F$ corresponds to the hyperedge
	\[E_i := \{x_j \in X: F_i\text{ contains }x_j\text{ or }\bar{x}_j\}.\]
	\item Let $E'$ be a subset of hyperedges in $H$, the \emph{boundary} of $E'$ is the set of all vertices that belong to exactly one hyperedge in $E'$. We say that $H$ has \underline{property $P(a)$} if, for every $m \le an$, every family of $m$ edges has boundary size at least $m/2$.
	\item Let $\calS$ be a family of subsets of $X$ (note that $\calS$ might be a multiset). A \emph{system of distinct representatives} (SDR) of $\calS$ is a mapping from each $S \in \calS$ to an element in $S$ such that different subsets in $\calS$ are mapped to different elements. Alternatively, consider the bipartite graph $(\calS, X)$ such that an edge between $S \in \calS$ and $x_i \in X$ is drawn if and only if $x_i \in S$, then an SDR of $\calS$ is an $\calS$-perfect matching of this bipartite graph (that is, every vertex in $\calS$ is matched).
	\item Let $S$ be a subset of vertices in $H$ of size $s := \lfloor bn\rfloor$. We say that $S$ is \emph{good} if there is a subset $D$ of $S$ with $|S\setminus D|\le (a/32)|S|$ such that every family of at most $an$ edges has an SDR that is disjoint from $D$. We denote this subset as $D(S)$. We say that $H$ has \underline{property $Q(a, b)$} if a random size-$s$ subset $S\subseteq X$ is good with probability at least $1/2$.
	\item \cite[Lemma 3]{ChvatalS88} showed that any hypergraph satisfying certain ``sparsity'' conditions will have properties $P(a)$ and $Q(a, b)$. As a corollary (\cite[Lemma 4]{ChvatalS88}), for every large enough integers $k$ and $c\ge 0.7\cdot 2^k$, there are $a, b > 0$ with $b\le a/8$ such that a random $k$-uniform hypergraph with $n$ vertices and $cn$ hyperedges has properties $P(a)$ and $Q(a, b)$ with probability $\ge 1-n^{-2}$ for large enough $n$.\footnote{If this probability is at least $1-n^{-2}$, then we can argue that with probability $1$ over an infinite family of random $k$-CNFs, our reduction to $\rwPHP(\PLS)$ is correct on all but finitely many input lengths; see the end of the proof of \autoref{thm: refuters for random k-CNF lower bounds}. Although \cite{ChvatalS88} only claimed a probability of $1-o(1)$, their proof actually shows a probability of $1-n^{-\Omega(k)}$ where the big $\Omega$ hides some absolute constant. This is at least $1-n^{-2}$ when $k$ is large enough; we suspect that our results can be extended to all $k\ge 3$ via a more careful argument.} %
\end{itemize}

Now we outline the strategy of \cite{ChvatalS88}. Let $F = F_n$ be a $k$-CNF over $n$ variables and $cn$ clauses whose underlying hypergraph $H$ satisfies $P(a)$ and $Q(a, b)$. We first choose a ``special pair'' $(S, \rho)$ where $S$ is a random subset of $s := \lfloor bn \rfloor$ vertices and $\rho \in \{0, 1\}^{D(S)}$ is a uniformly random restriction on variables in $D(S)$. Then we use a \emph{random restriction} argument to reduce the size lower bound to a \emph{width lower bound}:
\def\Vars{\mathrm{Vars}}
\begin{itemize}[align=left]
	\item [{\bf Random restriction:}] Let $C$ be any clause in the purported resolution refutation for $F$ such that the width of $C$ is at least $an/8$. With probability $1-2^{-\Omega(n)}$ over the choice of $S$, we have $|\Vars(C) \cap S| \ge as/16$, where $\Vars(C)$ denotes the set of variables contained in $C$. Since $|S\setminus D(S)| \le as/32$, it follows that $|\Vars(C) \cap D(S)| \ge as/32$, hence the probability over $\rho$ that $C$ is not killed by $\rho$ is at most $2^{-as/32}$. A union bound over all $C\in \Pi$ implies that with high probability over $(S, \rho)$, every clause of width $\ge an/8$ in $\Pi$ is killed by $\rho$.
	\item [{\bf Width lower bound:}] Now we are left with a purported resolution refutation $\Pi|_\rho$ of width less than $an/8$ for the statement $F|_\rho$. For a clause $C \in \Pi$, let $\mu(C)$ denote the minimum number of clauses from $F$ that logically implies $C$ under $\rho$. (That is, for any assignment extending $\rho$, if all these clauses are satisfied, then $C$ is also satisfied.) Every subset of $\le an/2$ clauses $F'\subseteq F$ can be satisfied by some assignment on $X\setminus D(S)$ (indeed, we can simply choose an SDR for $F'$ that is disjoint from $D(S)$, and fix this SDR), hence $\mu(\bot) > an/2$. On the other hand, every clause in $F$ has $\mu$ value at most $1$. Let $C' \in \Pi|_\rho$ be the first clause in $\Pi|_\rho$ such that $\mu(C') > an/2$ (recall that $\bot$ is the last clause in $\Pi|_\rho$). One can use a classical argument to show that $an/2 < \mu(C') \le an$, i.e., the smallest subset of clauses $F'\subseteq F$ that logically implies $C'$ has size between $an/2$ and $an$. Since $H$ satisfies $P(a)$ and $|F'| \le an$, the boundary of $F'$ contains at least $|F'|/2 \ge an/4$ variables. It can be shown that $C'$ contains every variable in the boundary of $F'$ but not in $S$, hence $w(C') \ge an/8$, a contradiction.
\end{itemize}

Now we are ready to prove \autoref{thm: refuters for random k-CNF lower bounds}.
\begin{proof}[Proof of \autoref{thm: refuters for random k-CNF lower bounds}]
	Let $F = F_n$ be the random $k$-CNF and $H = H_n$ be the underlying hypergraph for $F$. Let $a, b > 0$ be constants that arise from \cite[Lemma 4]{ChvatalS88}, we assume that $H$ has properties $P(a)$ and $Q(a, b)$ (this assumption will be justified at the end of the proof). Our reduction needs the following non-uniform advice $\{S_i\}, \{D_i\}, \{\calR_{i, \rho}\}$ (of course, they only depend on $F$ and is independent of the purported resolution refutation):
	\begin{itemize}
		\item A list of subsets $S\subseteq X$ with size $s := \lfloor bn\rfloor$ that are good. Since $H$ has property $Q(a, b)$, there are at least $N_{\sf good} := \binom{n}{s}/2$ such subsets and we only need to encode the first $N_{\sf good}$ ones. For each $i\in [N_{\sf good}]$, denote the $i$-th good subset as $S_i$, we also need the subset $D_i \subseteq S_i$ of size $\ge (1-a/32)s$ such that every family of at most $an$ edges has an SDR disjoint from $D_i$.
		\item For each index $i$ and each restriction $\rho \in \{0, 1\}^{D_i}$, we compute the subformula $F|_\rho$. The above width lower bound argument (along with properties $P(a)$ and $Q(a, b)$) implies that $F|_\rho$ requires resolution width $>an/8$. Invoking \autoref{lemma: width refuter upper bound}, we obtain a non-uniform decision tree reduction from $\Refuter(w(F|_\rho) \le an/8)$ to $\PLS$ with block-depth $2$, which we denote as $\calR_{i, \rho}$.
	\end{itemize}

	Let $\Pi$ be a purported resolution refutation for $F$ consisting of at most $L := (1+\eps)^n$ clauses. Now we describe our reduction from $\Refuter(s(F_n) \le (1+\eps)^n)$ to $\rwPHP(\PLS)$: 
	\begin{itemize}
		\item [($f$)] The function $f$ takes as inputs $i \in [L]$, $type \in \{0, 1\}$, and $w \in [\binom{n}{s} \cdot 2^{(1-a/32)s}\cdot 2^{-c'n}]$, where $c' > 0$ is a small enough constant depending on $a$ and $b$. Essentially, it treats $(i, type, w)$ as the compression of a bad ``special pair'' $(S, \rho)$ (where $S$ is a good size-$s$ subset and $\rho$ is an assignment over $D(S)$) and decompresses it. We start by checking that $w(C_i) \ge an/8$; if this is not the case then $f$ outputs $\bot$. Next:\begin{itemize}
			\item If $type = 0$, then this means $|\Vars(C_i)\cap S| < as/16$. Recall that if $|S| = s$ is chosen uniformly at random, then the probability that $|\Vars(C_i) \cap S| < as/16 \le 0.5\cdot |C_i|s/n$ should be at most $2^{-c'n}$ for some small enough constant $c' > 0$. Hence, $(S, \rho)$ can be compressed into $(\log\binom{n}{s} - c'n) + |D(S)|$ bits. We treat $w$ as this compression and recover $(S, \rho)$ from $w$.
			\item If $type = 1$, then $|\Vars(C_i)\cap S| \ge as/16$ but $C_i$ is not killed under $\rho$. In this case, the values of $\rho$ over $\Vars(C_i) \cap D(S)$ can be inferred from $C_i$. Since $|\Vars(C_i) \cap D(S)| \ge as/16-as/32 = as/32$, this provides us a way to compress $(S, \rho)$ into $\log\binom{n}{s} + (|D(S)| - as/32) \le \log\binom{n}{s} + |D(S)| - c'n$ bits. Again, we treat $w$ as this compression and recover $(S, \rho)$ from $w$.
		\end{itemize}
		Now that we obtained $(S, \rho)$, we can find an index $j\in [N_{\sf good}]$ such that $S = S_j$ (using non-uniformity). If such $j$ does not exist, then $f$ outputs $\bot$; otherwise $f$ outputs $(j, \rho)$.

		Hence we have $f: [L]\times \{0, 1\} \times [\binom{n}{s} \cdot 2^{(1-a/32)s}\cdot 2^{-c'n}] \to [N_{\sf good}] \times \{0, 1\}^{(1-a/32)s}$. (If $f$ outputs $\bot$ then we can assume that it outputs a default value, say $(0, 0^{(1-a/32)s})$, instead.) Recall that $L = (1+\eps)^n$ and $N_{\sf good} = \binom{n}{s}/2$, which means if $\eps > 0$ is small enough then
		\begin{equation}\label{eq: dwPHP used in random kCNF}
			\frac{2L\cdot \binom{n}{s}\cdot 2^{(1-a/32)s}\cdot 2^{-c'n}}{N_{\sf good}\cdot 2^{(1-a/32)s}} \le 2^{-\Omega(n)} \ll 1,
		\end{equation}
		hence $f$ is indeed shrinking. Given an input $(i, type, w)$, its $f$ value can be computed by a non-uniform decision tree of block-depth $1$.
		\item [($I_{j, \rho}$)] Given $j\in [N_{\sf good}]$ and $\rho \in \{0, 1\}^{(1-a/32)s}$, we compute a $\PLS$ instance $I_{j, \rho}$ as follows. Abusing notation, we also use $\rho$ to denote the restriction that equals to $\rho$ on $D_j$ and does not restrict any variable outside $D_j$. Let $\Pi|_\rho$ denote the restriction of $\Pi$ over $\rho$, then each clause of $\Pi|_\rho$ can be computed in block-depth $1$ from $\Pi$. Then we apply the reduction $\calR_{i, \rho}$ on $\Pi|_\rho$ to obtain the $\PLS$ instance $I_{j, \rho}$.
		\item [($g$)] Let $j \in [N_{\sf good}]$ and $\rho \in \{0, 1\}^{(1-a/32)s}$. Given a valid solution $o$ of $I_{j, \rho}$, we can compute an index $i\in [L]$ from $o$ such that the $i$-th step in $\Pi|_\rho$ is an illegal derivation. As in the definition of $f$, we can (assume $w(C_i) \ge an/8$ and) compress $(j, \rho)$ as $(i, type, w)$; then we set $g_{(j, \rho), o} = (i, type, w)$. If $f(i, type, w) \ne (j, \rho)$, then it must be the case that the $i$-th step in $\Pi$ is already incorrect (instead of the case that $w(C_i)$ is too large).
	\end{itemize}

	The above reduction is correct as long as $H$ has properties $P(a)$ and $Q(a, b)$, and its block-depth is $2$.
	
	It remains to show that our reduction is correct with probability $1$. In fact, for each $N\ge 1$, the probability that for every $n\ge N$, $H_n$ has properties $P(a)$ and $Q(a, b)$ is at least
	\[\prod_{n\ge N}(1-n^{-2}) = \frac{N-1}{N}.\]
	It follows that with probability $1$ over the family $\{F_n\}_{n\in \N}$, all but finitely many $H_n$ has properties $P(a)$ and $Q(a, b)$. In this case, our reduction will be correct on all but finitely many input lengths.
\end{proof}

\subsection{Open Problems: What We \texorpdfstring{\emph{Failed}}{Failed} to Formalize}\label{sec: what we failed to formalize}

One interesting problem left open by this work is whether the general size-width trade-offs in \cite{Ben-SassonW01} can be proved in $\rwPHP(\PLS)$. Ben-Sasson and Wigderson showed that for any unsatisfiable $k$-CNF $F$, if $F$ requires resolution width $w$ to refute, then $F$ also requires resolution size $2^{\Omega(w-k)^2/n}$ to refute. This naturally leads to the following conjecture:

\begin{conjecture}[Informal]\label{informal conjecture on BSW01}
	Let $F$ be an unsatisfiable $k$-CNF with resolution width $> w_F$ and let $s_F := 2^{\Omega(w_F-k)^2/n}$. Let $\calP$ denote the problem $\Refuter(w(F\vdash_\Res\bot) \le w_F)$, then there is an efficient decision-tree reduction from $\Refuter(s(F\vdash_\Res\bot) \le s_F)$ to $\rwPHP(\calP)$. In particular, there is always an efficient non-uniform decision tree reduction from $\Refuter(s(F\vdash_\Res\bot) \le s_F)$ to $\rwPHP(\PLS)$.
\end{conjecture}

Roughly speaking, one obstacle against proving \autoref{informal conjecture on BSW01} is that the averaging argument used in the proof of \cite[Theorem 3.5]{Ben-SassonW01} seems to rely on ``$\APC_2$-style'' \cite{Jerabek-APC2} approximate counting: one needs to estimate the number of ``fat'' clauses up to an $(1+\eps)$-multiplicative factor. Therefore, we have been unable to formalize the proof of \cite[Theorem 3.5]{Ben-SassonW01} in $\T^1_2 + \dwPHP(\PV)$ where only ``$\APC_1$-style'' \cite{Jerabek-APC1} approximate counting is available.

We also leave open the complexity of proving resolution lower bounds by combining monotone circuit lower bounds \cite{Razborov_monotone,AlonB87, Haken95} with feasible interpolation \cite{razborov1995unprovability, Krajicek97, Pudlak97}. To formalize Razborov's approximation method \cite{Razborov_monotone}, it seems that we need to iteratively define exponentially many set families (one for each node in the resolution proof) and apply the sunflower lemma \cite{erdos1960intersection, ALWZ-sunflower} to each of them. It is unclear to us how to formalize such arguments in $\T^1_2 + \dwPHP(\PV)$. (See also \cite{GGKS} who used lifting techniques to prove monotone circuit lower bounds and resolution lower bounds.)

We showed in \autoref{thm: black-box model PLS-completeness for width} that the refuter problem for every true resolution width lower bound is $\PLS$-complete under non-uniform reductions. It would be very interesting to see whether the size lower bound analog holds or not. We propose the following conjecture (which is stronger than the non-uniform version of \autoref{informal conjecture on BSW01}):
\begin{conjecture}[Informal]\label{informal conjecture: every Resolution size lb is in rwPHP(PLS)}
    Let $F$ be an unsatisfiable CNF that requires resolution size $\ge s_F$ to refute. Then the problem $\Refuter(s(F\vdash_\Res\bot) < s_F)$ is $\rwPHP(\PLS)$-complete under non-uniform decision tree reductions.
\end{conjecture}
(Note that the \emph{average-case} version of \autoref{informal conjecture: every Resolution size lb is in rwPHP(PLS)}, where $F$ is a random $k$-CNF and $s_F = 2^{\Omega(n)}$, is already proved in \autoref{sec: more rwPHP(PLS) upper bounds}, by formalizing the resolution size lower bounds of \cite{ChvatalS88}.)

We end this subsection by mentioning a subtle technical issue in our proofs. There are two natural properties in the completeness of resolution (i.e., resolution can prove every true statement within size $2^n$): the proof does not require weakening, and it avoids producing duplicate clauses. However, in our current $\PLS$-hardness of refuting resolution width lower bounds and $\rwPHP(\PLS)$-hardness of refuting resolution size lower bounds, the resolution proofs produced in our reduction rely on both weakening rules and duplicated clauses. This raises an open question: What is the complexity of the corresponding refuter problems if the proofs are restricted from using either weakening rules or duplicate clauses?

\setlist[description]{leftmargin=\parindent,labelindent=\parindent}
    
    \newcommand{\probP}{\mathcal{P}}
    \newcommand{\probQ}{\mathcal{Q}}
    \newcommand{\probR}{\mathcal{R}}
    \newcommand{\probS}{\mathcal{S}}
    \newcommand{\pAss}{\rho}
    \newcommand{\ue}{=^U}
    \newcommand{\uleq}{\leq^U}
    
    \newcommand{\FWLB}{\mathcal{F}_{\mathsf{wLB}}^w}
    \newcommand{\FSLB}{\mathcal{F}_{\mathsf{sLB}}^L}
    \def\criEPHP{\mathrm{cri_EPHP}}
    \def\criTs{\mathrm{cri_Tseitin}}
    
    \newcommand{\Rftdm}[4]{\textnormal{\textsc{Refuter}}_{#3,#4}(#1 \rightarrow #2)}

\section{Applications}\label{sec: applications}

    \subsection{Proof Complexity of Proof Complexity Lower Bounds}\label{sec: proof complexity of proof complexity lower bounds}
    In this subsection, we translate our $\TFNP$ upper bounds for the refuter problems into \emph{proof complexity upper bounds} for \emph{proof complexity lower bounds}, showing that resolution lower bounds can actually be proved in weak proof systems! In particular, we show that low-width resolution (itself) can prove lower bounds on resolution width (\autoref{thm: low width res prove res LB}), while low-width \emph{random resolution} (as defined in \cite{BussKT14, PudlakT19}) can prove resolution size lower bounds (\autoref{thm: low width rRes prove size lb}).\footnote{More precisely, we use \autoref{lemma: width refuter upper bound} to show that low-width resolution can prove \emph{every} resolution width lower bound that is true, and use \autoref{thm: refuters for random k-CNF lower bounds} to show that low-width random resolution can prove \emph{most} resolution size lower bounds.} 
    This stands in stark contrast to the results proven in \cite{AtseriasM19, Garlik19, RezendeGNPR021} that resolution cannot prove size lower bounds against itself.

    \paragraph{Formalization of proof complexity lower bounds as CNFs.} Suppose a family of formulas $\calF = \{F_n\}$ does not have a width $w_F$ resolution refutation (i.e., $w(\calF \vdash_\Res\bot) > w_F$). Then we can transform the refuter problem $\Refuter(w(\calF\vdash_{\Res} \bot) \le w_F)$ into a family of unsatisfiable CNFs $\FWLB$ using via false clause search problem (\autoref{equ: search F}). That is, an unsatisfiable CNF $F^w_{\sf wLB}$ in the family $\FWLB$ is defined as follows:\begin{itemize}
        \item The input of $F^w_{\sf wLB}$ is a purported length-$L$ resolution refutation for $F_n$ represented as a list of nodes $C_0, C_1, \dots, C_{L-1}$ and each node $C_i$ can be encoded in $O(w_F\log n)$ bits.
        \item Each potential solution $sol$ of the refuter problem can be verified by a decision tree of block-depth $3$, hence they can be turned into a CNF $C_{sol}$ of width $O(w_F\log n)$. $F^w_{\sf wLB}$ is simply the conjunction of these CNFs.
    \end{itemize}
    We can similarly transform a resolution \emph{size} lower bound $s(\calF \vdash_\Res \bot) > L$ into a family of unsatisfiable CNFs $\FSLB$ via the refuter problem $\Refuter(s(\calF\vdash_{\Res} \bot) \le L)$. The only difference is that each node consists of an (unbounded-width) clause and thus is encoded in $O(n + \log L)$ bits.
    
    It is easily seen that $\FWLB$ are CNFs of width $O(w_F\log n)$ and $\FSLB$ are CNFs of width $O(n+\log L)$. (When $L = 2^{n^{\Omega(1)}}$, these width parameters are $\polylog(L)$ and can be thought of as ``efficient''.)

\begin{mdframed}[hidealllines=true,backgroundcolor=gray!10,skipabove=0.3em,skipbelow=-0.4em,innertopmargin=0]
	\small
\begin{remark}[Comparison with previous formalizations]\label{remark: formalization of resolution lower bounds}
    Similar formalizations of resolution lower bound statements have also appeared in \cite{AtseriasM19, Garlik19, RezendeGNPR021}. The biggest difference between these formalizations is that in \cite{RezendeGNPR021}, the predecessors of each node are represented in binary and as $O(\log N)$ bits; while in \cite{AtseriasM19, Garlik19}, the predecessors are represented in unary and we have tables $L[i, j]$ and $R[i, j]$ denoting whether node $j$ is a predecessor of node $i$. Note that in the unary representation, it requires an axiom of width $L$ to express that every node $u$ has at least one predecessor $L[u]$ and at least one predecessor $R[u]$. Thus it is impossible to prove resolution width lower bounds in resolution width $O(w\log N) \ll L$. Therefore, we choose to use the binary formalization as in \cite{RezendeGNPR021}.

    The formalization in \cite{RezendeGNPR021} allows \emph{disabled} nodes in the resolution proof. Our proof complexity upper bounds hold regardless of whether such nodes are allowed in the formalization.
\end{remark}
\end{mdframed}

    \paragraph{Low-width resolution can prove resolution width lower bounds.} First, we show that:

    \begin{theorem}\label{thm: low width res prove res LB}
        For every family of unsatisfiable CNFs $\calF$, if $w(\calF \vdash_\Res\bot) > w_F$, then $w(\FWLB \vdash_\Res\bot) \leq O(w_F\log N)$.
    \end{theorem}

    \autoref{thm: low width res prove res LB} follows from the proof of \autoref{lemma: width refuter upper bound} and \autoref{thm: pls characterize resolution}: since the refuter problem corresponding to resolution width lower bounds can be solved in $\PLS$ and the totality of $\PLS$ can be proved in low resolution width, it follows that resolution width lower bounds themselves can be proved in low resolution width. For the sake of intuition, we also present an equivalent but more direct proof using \emph{Prover-Delayer games}~\cite{Pudlak00ProofsAsGames}. The necessary backgrounds on Prover-Delayer games are presented in \autoref{app:res_pls:pls2res}.
    
    \begin{proof}
         It suffices to construct a Prover strategy with memory size $O(w_F\log N)$ in the {Prover-Delayer game} for $\FWLB$. The Prover starts by querying the last node in the purported resolution proof, which should contain the empty clause $\bot$. The Prover maintains the invariant that she is always at some (not disabled) clause $C_i$ such that $w(\calF \vdash C_i) > w_F$, i.e., it requires resolution width $> w_F$ to derive $C_i$ from the axioms. Each time the Prover is at some clause $C_i$:\begin{itemize}
             \item Suppose $C_i$ is \emph{resolved} from the clauses $C_j, C_k$. Then the Prover queries $C_j$ and $C_k$; if $j\ge i$, $k\ge i$, or the derivation from $(C_j, C_k)$ to $C_i$ is invalid, then she wins the game. Otherwise, since the widths of $C_j$ and $C_k$ are at most $w_F$ (recall that this is guaranteed syntactically by only allocating $w_F$ variables to each clause), one of $C_j, C_k$ must require $> w_F$ width to derive. Suppose it is $C_j$; that is, $w(\calF\vdash C_j) > w_F$. Then the Prover forgets $C_i$ and $C_k$ and only remembers $C_j$.
             \item Suppose $C_i$ is a \emph{weakening} of a clause $C_j$. The Prover queries $C_j$; if $j \ge i$ or the weakening from $C_j$ to $C_i$ is invalid, then she wins the game. Otherwise, it must be the case that $w(\calF \vdash C_j) > w_F$. Then the Prover forgets $C_i$ and only remembers $C_j$.
         \end{itemize}
         Since the index $i$ is always decreasing, the Prover is guaranteed to win the game. The Prover only needs to memorize $O(1)$ resolution nodes, i.e., $O(w_F\log N)$ bits.
    \end{proof}

    \paragraph{Low-width random resolution can prove resolution size lower bounds.} We first define the random resolution system (denoted as $\rRes$):

    \begin{definition}[\cite{BussKT14, PudlakT19}]
    An $\varepsilon$-\emph{random} resolution refutation of an unsatisfiable formula $F$ is a distribution ${\calD}$ supported on pairs $(\Pi,B)$, such that 
    \begin{enumerate}
        \item each $B$ is a CNF formula over the variables of $F$,
        \item $\Pi$ is a resolution refutation of $F \wedge B$, and
        \item for any assignment $x \in\{0,1\}^n$, $\Pr_{(\Pi,B)\sim {\calD}}[B(x)=1] \geq 1-\varepsilon$.\label{item: most B satisfies worst-case x}
    \end{enumerate}
    The \emph{size} $s(F \vdash_{\rRes}\bot)$, and \emph{width} $w(F \vdash_{\rRes}\bot)$ of a random resolution refutation $\calD$ for $F$ are the maximum size and width of a proof $\Pi$ in the support of $\calD$, respectively.
    \end{definition}

    We remark that random resolution is not a standard (i.e., Cook--Reckhow) proof system since the distribution $\calD$ might potentially require exponentially many bits to describe and it is also unclear how to verify \autoref{item: most B satisfies worst-case x} above. (In fact, random resolution cannot be simulated by a Cook--Reckhow proof system unless $\P = \NP$ \cite[Proposition 3.3]{PudlakT19}.) On the other hand, strong lower bounds on both width and size are known for random resolution \cite{PudlakT19}, suggesting that it may be classified as a ``weak'' proof system.

    \begin{theorem}\label{thm: low width rRes prove size lb}
        For every $k\ge 3$ and $c\ge 0.7\cdot 2^k$, there exists some $\eps > 0$ such that the following holds. Let $F$ be a random $k$-CNF formula chosen from the distribution $\calF(k, n, cn)$, $L := (1+\eps)^n$, and $\FSLB(F)$ be the CNF formula encoding the lower bound that $F$ requires size-$L$ resolution refutation. With probability tending to $1$ (when $n\to \infty$) over $F$, $\FSLB(F)$ admits a $\poly(n)$-width $\gamma$-\emph{random} resolution refutation with $\gamma := 2^{-\Omega(n)}$.
    \end{theorem}

    Similarly, \autoref{thm: low width rRes prove size lb} is a corollary of \autoref{thm: refuters for random k-CNF lower bounds}: if a search problem reduces to $\rwPHP(\PLS)$, then it also \emph{randomly} reduces to $\PLS$, and such a random reduction can be translated into a random resolution refutation. Nevertheless, for the sake of intuition, we present an (equivalent) proof that directly constructs the random resolution refutation $(\Pi, B)$.

    \def\Ngood{N_{\sf good}}
    \begin{proof}[Proof Sketch]
        We assume familiarity with the proofs in \autoref{sec: random k CNF}. We use the parameters $a, b$ from \cite[Lemma 4]{ChvatalS88}, and denote $s := \lfloor bn\rfloor$. We assume that the properties $P(a)$ and $Q(a, b)$ holds for $F$; by \cite[Lemma 4]{ChvatalS88}, this is true with high probability over $F\gets\calF(k, n, cn)$.
        
        Recall that the variables in $\FSLB(F)$ encode a length-$L$ resolution refutation $C_0, \dots, C_{L-1}$ of $F$, where $L := (1+\eps)^n$. Let $S_0, S_1, \dots, S_{\Ngood-1}$ denote the first $\Ngood := \binom{n}{s}/2$ good size-$s$ subsets. For each $j\in [\Ngood]$, also let $D_j$ denote any subset of $S_j$ of size $\ge (1-a/32)s$ such that every family of at most $an$ edges has an SDR disjoint from $D_j$. To sample a pair $(\Pi, B)$:\begin{enumerate}
            \item We first pick a random $j\in [\Ngood]$ and then pick a string $\rho \gets \{0, 1\}^{D_j}$. We also treat $\rho$ as a restriction that fixes every variable in $D_j$ and leaves everything else unchanged.
            \item For each $i\in [L]$, let $B_i$ be the decision tree verifying that either $w(C_i) < an/8$ or $\rho$ kills $C_i$. Note that $B_i$ only depends on the clause $C_i$, which can be encoded in $\poly(n)$ bits. Let $B \coloneqq \bigwedge_{i\in [L]}B_i$, then $B$ is a $\poly(n)$-width CNF.
            
            Moreover, following the same calculation as \autoref{eq: dwPHP used in random kCNF},  we can show that for any assignment $x$ to the variables in $\FSLB(F)$ (i.e., $x$ encodes a purported resolution refutation of $F$), the probability over $B$ (i.e., over $j$ and $\rho$) that $B(x) = 1$ is at least $1-2^{-\Omega(n)}$.

            \item It remains to argue that there always exists a $\poly(n)$-width resolution refutation $\Pi$ for $\FSLB(F) \land B$. We can use a similar Prover's strategy as described in \autoref{thm: low width res prove res LB}. Recall that for any clause $C$, $\mu(C)$ denotes the minimum number of clauses from $F$ that logically implies $C$ under $\rho$, and that $\mu(\bot) > an/2$. The Prover starts from $C_{L-1} = \bot$ and maintains the invariant that she is always at some clause $C_i$ where $\mu(C_i) > an/2$. In addition, when the Prover is at some clause $C_i$, she also ensures that $B_i$ is satisfied. At some stage, she will encounter some $C_i$ that is resolved from $C_j, C_k$, such that $\mu(C_i) > an/2$ and $\mu(C_j), \mu(C_k) < an/2$. But due to the width lower bound in \autoref{sec: random k CNF}, this will imply that either the $i$-th derivation is invalid, or that $B_i$ is violated.
            
            It is easy to check that this Prover strategy only requires $\poly(n)$ memory.\qedhere
        \end{enumerate}
    \end{proof}

\subsection{Complexity of Black-Box \texorpdfstring{$\TFNP$}{TFNP} Separations}\label{sec:app:tfnp}

    In this subsection, we introduce a new type of refuter problems --- \emph{$\TFNP^\dt$ refuter} --- which corresponds to the ``complexity'' of proving black-box $\TFNP^\dt$ separations. We present the definition and several basic properties of them in \autoref{sec:app:tfnp:def}. In \autoref{sec:app:tfnp:pls}, we relate the $\TFNP^\dt$ refuter to the resolution width refuter (\autoref{lem: res_ref to pls_ref}). Combining this with our results on resolution width refuter for $\EPHP$ and $\Tseitin$, we characterize the ``complexity'' of separating $\PPA$ and $\PPP$ from $\PLS$ in the black-box setting by the class $\PLS$ itself.

    \paragraph{Notations.} For two $\TFNP^{\dt}$ problems $\probP, \probQ$, we write $\probP \leq_m \probQ$ if there is a \emph{many-one} reduction from $\probP$ to $\probQ$; if the reduction is also uniform, we write $\probP \leq_m^U \probQ$. %

\subsubsection{Black-Box \texorpdfstring{$\TFNP$}{TFNP} Refuters and its Properties}\label{sec:app:tfnp:def}

We start by providing a formal definition of the $\TFNP^\dt$ refuter problems. Roughly speaking, in the problem $\Rftdm{\probP}{\probQ}{d}{M}$, we are given a shallow decision tree that claims to reduce $\probP$ to $\probQ$, and our goal is to find a witness that this shallow decision tree is incorrect.

    \begin{mdframed}[hidealllines=true,backgroundcolor=gray!10]
    \begin{center}
        \textbf{Problem $\Rftdm{\probP}{\probQ}{d}{M}$}
    \end{center}

    \underline{Parameters}: Two $\TFNP^\dt$ problems $\probP = \{P_N\}, \probQ = \{Q_N\}$ and two functions $d \coloneqq d(N), M \coloneqq M(N)$ such that there is no depth-$d(N)$ decision tree reduction from $P_N$ to $Q_{M(N)}$ for any $N$.
    
    \underline{Input}: A purported depth-$d$ decision tree reduction $(f_i, g_o)_{i \in M, o \in O_Q}$ from $P_N = \{0,1\}^N \times O_P$ to $Q_M = \{0,1\}^M \times O_Q$.
    
    \underline{Output}: A pair $(\pAss, o^*)$, where
            \begin{itemize}
                \item $\pAss \in \{0,1,*\}^N$ is a partial assignment encoded by specifying the locations and the values of all non-$*$ bits;
                \item $o^* \in O_Q$ is a solution of the problem $Q_M$.
            \end{itemize}
            
            The pair $(\pAss, o^*)$ satisfy that for any input $x \in \{0,1\}^N $ consistent with $\rho$,
            \begin{enumerate}
                \item $(f(x), o^*) \in \probQ$ and $(x, g_{o^*}(x)) \notin \probP$;
                \item only bits specified in $\pAss$ are ever queried when calculating $g_{o^*}(x)$ and verifying $(f(x), {o^*}) \in \probQ$ and $(x, g_{o^*}(x)) \notin \probP$.
            \end{enumerate}
\end{mdframed}

    In this section, we only consider refuting \emph{low-depth} many-one decision tree reduction. Thus, we always assume the functions $d(N)$ and $\log M(N)$ are poly-logarithmic in $N$ when we write ``for any $d, M$''. We also assume $M(N) \geq N$, so we will not consider reductions that are too weak. Note that it is necessary to have $o^*$ as part of the solution for this problem to be in $\TFNP^\dt$; otherwise, it might take too many queries to the input (reduction) to find $o^*$, which is used to refute the reduction later. %

    We call a $\TFNP^\dt$ problem $\mathcal{R}$ \emph{syntactical} if all the decision trees $(T_o)$ for verifying the solution can be replaced by a single polynomial-time oracle Turing machine.\footnote{A \emph{syntactical} $\TFNP^\dt$ problem is essentially a type-$2$ $\TFNP$ ($\TFNP^2$) problem, see~\cite{BeameCEIP98}.} Since we mostly care about the black-box separations between \emph{syntactical} \TFNP\ subclasses, we assume all the $\TFNP^\dt$ problems in this section are syntactical.

    We now present several basic properties regarding the $\TFNP^\dt$ separation refuter.
    First, a weaker reduction, which has a lower depth or smaller instance size, is easier to refute. The proof trivially follows from the definition (where item 2 needs the problem $\probQ$ to be paddable).

    \begin{lemma}\label{lem:tfnp_ref:trivial}
        \begin{enumerate}
            \item If $d_1 \leq d_2$, then
                $\Rftdm{\probP}{\probQ}{d_1}{M} \uleq_m \Rftdm{\probP}{\probQ}{d_2}{M}$.
            \item If $M_1 \leq M_2$, then
                $\Rftdm{\probP}{\probQ}{d}{M_1} \uleq_m \Rftdm{\probP}{\probQ}{d}{M_2}$.
        \end{enumerate}
    \end{lemma}

    Even in the easiest parameter settings, i.e., $d=0, M(N) = N$, the refuter problem $\Rftdm{\probP}{\probQ}{0}{N}$ is least as hard as $\probQ$ itself, because a valid solution of $\probQ$ is always required to witness a mistake given by the input reduction.

    \begin{lemma}\label{lem:tfnp_ref:depth0 reduction}
        $\probQ \uleq_m \Rftdm{\probP}{\probQ}{0}{N}$.
    \end{lemma}
    
    \begin{proof}
        Let $y \in \{0,1\}^N$ be an instance of problem $Q_N \in \{0,1\}^N \times O_Q$ and let $o_P$ be an arbitrary fixed solution of problem $P_N$.
        We construct a trivial depth-$0$ reduction  $(f_i, g_o)_{i \in N, o \in O_Q}$, where 
        \[ f_i(x) = y_i, \forall i \in N; g_o(x) = o_P, \forall o \in O_Q. \]
    
        Consider such reduction $(f_i, g_o)$ as an instance of $\Rftdm{\probP}{\probQ}{0}{N}$, and let $(\pAss, o^*)$ be any solution of it. By definition, $o^*$ is a valid solution of instance $y$. Moreover, our reduction is uniform, though $(f_i, g_o)$ is not.
    \end{proof}

    Finally, we present two useful lemmas, which state that it is easier to refute a reduction when the \emph{difficulty gap} between these two problems becomes larger.

    \begin{lemma}\label{lem:tfnp_ref:gap1}
        If $\probP \uleq_m \probS$, and $d_1(N), \log M_1(N) = \polylog(N)$, then 
        \[\Rftdm{\probS}{\probQ}{d_1}{M_1} \uleq_m \Rftdm{\probP}{\probQ}{d_2}{M_2},\]
        for some $d_2(N), \log M_2(N) = \polylog(N)$.
    \end{lemma}

    \begin{proof}
        Given a depth-$d_1$ reduction $(f_i, g_o)_{i \in M_1, o \in O_Q}$ from $S_N$ to $Q_{M_1}$, we compose it with any (uniform) low-depth reduction $(h_i, l_o)_{i \in N, o \in O_S}$ from $P_{N'}$ to $S_N$. Now we get a depth-$d'$ reduction $(f'_i, g'_o)$ from $P_{N'}$ to $Q_{M_1}$ with
        \[f'_i(x) = f_i(h(x)),\forall i \in [M_1]; \; \quad \; g'_o(x) = l_s(x), s \coloneqq g_o(h(x)), \forall o \in O_Q. \]
        Let $d_2(N') \coloneqq d', M_2(N') \coloneqq M_1$, and it is easy to verify that $d_2(N'), \log M_2(N') = \polylog(N')$.
    
        Consider a pair $(\pAss_P, o^*)$ that refutes $(f'_i, g'_o)$. Let $x$ be any input of $P_{N'}$ that is consistent with $\pAss_P$ and define $y \coloneqq h(x)$.
        We show how to construct a partial assignment $\pAss_S$ consistent with $y$ such that $(\pAss_S, o^*)$ refutes $(f_i, g_o)$. 
        We start with setting $\pAss_S$ to all $*$ strings, and then execute the process of

        \begin{center}
            \begin{minipage}{0.28\linewidth}
                \centering
                \textbf{S1:} calculating $g_{o^*}(y)$;
            \end{minipage}
            \begin{minipage}{0.34\linewidth}
                \centering
                \textbf{S2:} verifying $(f(y), {o^*}) \in \probQ$;
            \end{minipage}
            \begin{minipage}{0.32\linewidth}
                \centering
                \textbf{S3:} verifying $(y, g_{o^*}(y)) \notin \probS$.
            \end{minipage}
        \end{center}

        During the above process, if $y_i$ is queried and $y_i$ has not been specified by $\pAss_S$, we will execute $h_i(x)$ to calculate $y_i$ and then store its value in $\pAss_S$.
    
        For the correctness of our construction, recall that only bits specified in $\pAss_P$ are ever queried when 
        \begin{center}
            \begin{minipage}{0.28\linewidth}
                \centering
                \textbf{P1:} calculating $g'_{o^*}(x)$;
            \end{minipage}
            \begin{minipage}{0.34\linewidth}
                \centering
                \textbf{P2:} verifying $(f'(x), {o^*}) \in \probQ$;
            \end{minipage}
            \begin{minipage}{0.32\linewidth}
                \centering
                \textbf{P3:} verifying $(x, g'_{o^*}(x)) \notin \probP$.
            \end{minipage}
        \end{center}

        By our construction, process \textbf{S1, S2} are sub-procedures of \textbf{P1, P2}, and thus they will only query locations of $x$ that are already specified in $\pAss_P$. However, process \textbf{S3} might query some locations that are not specified in $\pAss_P$. In this case, it is safe to return arbitrary values for those queries. This is because the correctness of reduction $(h_i, l_o)$ guarantees that there must be $(y, g_{o^*}(y)) \notin \probS$. %

        Finally, note that our whole reduction, including the construction of $(f'_i, g'_o)$ and the execution of process \textbf{S1, S2, S3}, can be done in a uniform manner.
    \end{proof}
    
    With a similar argument, we can also formalize the other direction.
    
    \begin{lemma}\label{lem:tfnp_ref:gap2}
        If $\probS \uleq_m \probQ$ and let $d_1(N), \log M_1(N) = \polylog(N)$, then \[\Rftdm{\probP}{\probS}{d_1}{M_1} \uleq_m \Rftdm{\probP}{\probQ}{d_2}{M_2},\]
        for some $d_2(N), \log M_2(N) = \polylog(N)$.
    \end{lemma}

    We often consider all low-depth reductions between two $\TFNP^{\dt}$ classes with no valid low-depth reductions possible. So, it is convenient to introduce a new kind of $\TFNP^{\dt}$ subclasses for this type of problem.

    \begin{definition}
        For two $\TFNP^\dt$ classes $\mathsf{A}, \mathsf{B}$ ($\mathsf{A} \nsubseteq \mathsf{B}$) with $\probP, \probQ$ being any complete problems of $\mathsf{A}$ and $\mathsf{B}$ respectively, $\Rft{\mathsf{A}}{\mathsf{B}}$ is defined as the class of $\TFNP^{\dt}$ problems that are reducible to $\Rftdm{\probP}{\probQ}{d}{M}$ for some $d(N), \log M(N) = \polylog(N)$.
    \end{definition}
    
    This notation is well-defined because \autoref{lem:tfnp_ref:gap1} and \autoref{lem:tfnp_ref:gap2} guarantee that the choice of the complete problems does not matter. We also have the following corollary of \autoref{lem:tfnp_ref:trivial} and \autoref{lem:tfnp_ref:depth0 reduction}.

    \begin{corollary}\label{cor: tfnp_ref: general_lb}
        For any two $\TFNP^\dt$ classes $\mathsf{A}, \mathsf{B}$ such that $\mathsf{A} \nsubseteq \mathsf{B}$, $\mathsf{B} \subseteq \Rft{\mathsf{A}}{\mathsf{B}}.$
    \end{corollary}

\subsubsection{Refuter for Separating from \texorpdfstring{$\PLS$}{PLS}}\label{sec:app:tfnp:pls}

    Now we study the complexity of refuting separations between $\PLS$ and other classes in $\TFNP^\dt$, in particular the separations
    \[\PPA^\dt \nsubseteq \PLS^\dt \quad\text{and}\quad\PPP^\dt\nsubseteq \PLS^\dt.\]
    Our main tool is the equivalence between resolution and $\PLS$ via the \emph{false clause search} problem (cf.~\cite{RezendeGR22survey}): recall that $\SearchCNF(\calF) \in \PLS$ if and only if $\calF$ have a $\polylog(N)$-width resolution refutation (\autoref{thm: pls characterize resolution}).

    Studying this equivalence from a computational perspective, we related the $\TFNP^\dt$ refuter for $\PLS$ with the resolution width refuter.%

    \begin{restatable}{lemma}{LemmaResReftoPLSRef}\label{lem: res_ref to pls_ref}
        For any family of unsatisfiable CNF $\calF$ that has no $\polylog(N)$-width resolution refutation,
        \[\Rftdm{\SearchCNF(\calF)}{\Iter}{d}{M} \leq_m \Refuter(w(\calF\vdash_{\Res} \bot)<w_0)\] for some $w_0 = \polylog(N)$ that may depend on $d, M$.

        Furthermore, this reduction is uniform when $\calF$ is a uniform family of unsatisfiable CNFs.
    \end{restatable}

    The proof of \autoref{lem: res_ref to pls_ref} follows from the standard procedure that transforms a low-depth decision tree reduction to $\PLS$
    (i.e., a $\PLS$ formulation) to a low-width resolution proof, using the \emph{Prover-Delayer game}~\cite{Pudlak00ProofsAsGames}. This proof is rather straightforward, but many details have to be taken care of to make sure that the reduction is uniform.
    To be self-contained, we formally present the transformation from a $\PLS$ formulation to a resolution proof in \autoref{app:res_pls:pls2res};\footnote{This transformation is a well-known folklore among the \emph{Proof Complexity and} $\TFNP$ community. However, to the best of the authors' knowledge, it has not yet been formally written down in any previous literature.} we then prove  \autoref{lem: res_ref to pls_ref} in \autoref{app:res_pls:fullproof}.

    As an application, we combine \autoref{lem: res_ref to pls_ref} with our results on resolution width refuter for $\EPHP$ and $\Tseitin$ formulas. Note that $\SearchCNF(\EPHP)$ and $\SearchCNF(\Tseitin)$ are in $\PPP$ and $\PPA$ respectively. Therefore, we can reduce the $\TFNP^\dt$ refuter for $\PPP^\dt\nsubseteq \PLS^\dt$ and $\PPA^\dt\nsubseteq \PLS^\dt$ to the resolution width refuters for $\EPHP$ and $\Tseitin$ respectively.

    \begin{theorem}\label{thm: TFNP ref PPP PPA PLS}
        Let $\probP, \probQ$ be any complete problems for $\PPP$ and $\PPA$ respectively, then for any $d, M$, both
        $\Rftdm{\probP}{\Iter}{d}{M}$ and $\Rftdm{\probQ}{\Iter}{d}{M}$ are $\PLS$-complete via uniform reductions. 

        In particular, $\Rft{\PPP}{\PLS} = \Rft{\PPA}{\PLS} = \PLS$.
    \end{theorem}

    Equivalently, \autoref{thm: TFNP ref PPP PPA PLS} says that \emph{local search} arguments are both \emph{necessary} and \emph{sufficient} for separating $\PPP$ and $\PPA$ from $\PLS$ in the black-box setting.

    \begin{proof}%
        Note that \autoref{lem:tfnp_ref:depth0 reduction} already gives the $\PLS$-hardness result, we will focus on showing that for any $d, M$, $\Rftdm{\probP}{\Iter}{d}{M}$ and $\Rftdm{\probQ}{\Iter}{d}{M}$ are in $\PLS$ via uniform reductions.

        We start with $\PPP$ versus $\PLS$. Note that $\SearchCNF(\EPHP)$ is in $\PPP$, thus, 
        by \autoref{lem:tfnp_ref:gap1}, we have \[ \Rftdm{\probP}{\Iter}{d}{M} \leq_m^U \Rftdm{\SearchCNF(\EPHP)}{\Iter}{d'}{M'},\] where $d', \log M'$ are also poly-logarithmic in $N$. Combining with \autoref{lem: res_ref to pls_ref}, there is
        \[ \Rftdm{\probP}{\Iter}{d}{M} \leq_m^U \Refuter(w(\EPHP\vdash_{\Res} \bot)<w_0)\] for some $w_0 = \polylog(N)$ that may depend on $d, M$.

        Recall that \autoref{thm: white-box PLS-completeness for width} shows that $\Refuter(w(\EPHP\vdash_{\Res} \bot)< w_0)$ is in $\PLS$ via a uniform reduction when $w_0 = n/3$. The same reduction to $\Iter$ would still work when $w_0 = \polylog(N)$, and the only issue is to make sure that the $\cri(C)$ function could be calculated ``efficiently'' in this different parameter regime. Note that when $w_0 = \poly(n)$ (and $N = 2^{\Omega(n)})$), a $\poly(n)$ time procedure (\autoref{lem: cri EPHP efficient}) would be considered as time efficient; however, when $w_0 = \polylog(n)$ and $N$ being quasi-polynomial in $n$, only a $\polylog(n)$ running time is acceptable. Since $|C| \leq w_0 = \polylog(n)$, it suffices to prove that the following claim.
        \begin{claim}
            $\cri(C)$ can be calculated in $\polylog(n)$ time when $|C| = \polylog(n)$. 
        \end{claim}
        \begin{claimproof}
            We modify the algorithm described in the proof of \autoref{lem: cri EPHP efficient}. 
            First, notice that we do not have to enumerate all possible $\ell \in [n+1]$, because only $\polylog(n)$ pigeons are \emph{involved} in the clause $C$, where we say a pigeon $\ell$ is \emph{involved} in $C$ if a literal related to $\ell$ appears in $C$. Any pigeons that are not involved in $C$ would be equivalent, thus, we only need to consider any one of them.
            
            For a fixed $\ell$, deciding whether an $\ell$-critical assignment exists for $C$ is reduced to the following graph problem: Given a complete bipartite graph with $n$ pigeons on the left and $n$ holes on the right, $\polylog(n)$ sets of edges are then deleted, determine whether a perfect matching still exists in the end. Each deleted set can be described by a triple $(i, j_1, j_2)$, representing the set $\{(i,j): j_1 \leq j \leq j_2\}$. 
        
            It is not difficult to design an $\polylog(n)$ time algorithm for this problem by exploiting the sparsity:
            \begin{enumerate}
                \item We first ignore all pigeons with full degree $n$, because they could always be matched in the end.
                \item Suppose we have $t_1 = \polylog(n)$ pigeons left after the first step. We then ignore all pigeons with the degree at least $t+1$ for the same reason.
                \item We have $t_2 = \polylog(n)$ pigeons left now, and there are at most $t_1 \cdot t_2 = \polylog(n)$ edges connected to those pigeons. So, we can run the standard maximum matching algorithm on the subgraph of the remaining pigeons. The original graph has a perfect matching if and only if all $t_2$ pigeons could be matched.\qedhere
            \end{enumerate}
        \end{claimproof}

        A similar argument works for $\PPA$. We use the fact that $\SearchCNF(\Tseitin(G, \tau))$ is in $\PPA$ when the graph $G$ has a constant degree. We will fix a family of strongly explicit expander graph $G$ and an odd-weighted function $\tau$, rather than giving them as input as we did in \autoref{sec: tseitin}. For example, we can take $G$ as a 2D-grid with a boundary being wrapping around, and $\tau(v) = 1$ only if $v$ is some designated vertex (say $(1, 1)$). Then, we claim that the $\cri(C)$ function (defined differently for the $\Tseitin$ formula in \autoref{thm: Tseitin width lower bound refuter}) can also be calculated in $\polylog(n)$ time when $|C| = \polylog(n)$ by exploiting the sparsity of $C$. We omit the proof this claim here.
        Finally, using the same proof of \autoref{thm: Tseitin width lower bound refuter}, we show that $\Refuter(w(\Tseitin\vdash_\Res\bot) < w_0$ is in $\PLS$ via a uniform reduction when $w_0 = \polylog(N)$, which concludes the proof.
    \end{proof}

\section*{Acknowledgments}
Jiawei thanks Igor C.~Oliveira for introducing him to refuters and their connections with $\TFNP$ and thanks Robert Robere and Noah Fleming for knowledge of proof complexity. 

Yuhao thanks Toniann Pitassi for the knowledge and guidance in the field of proof complexity and thanks Robert Robere for introducing him to beneficial intuition about proof systems and $\TFNP^{\dt}$ classes. 

Hanlin thanks Svyatoslav Gryaznov and Iddo Tzameret for helpful discussions regarding \cite{RezendeGNPR021} and proof complexity in general, and Rahul Santhanam and Ján Pich for beneficial conversations. %

We thank Lijie Chen, Jiatu Li, and Igor C.~Oliveira for sending us a preliminary version of \cite{ChenLiOliveira24}. We thank Michal \Garlik for helpful discussions on \cite{Garlik19}. We thank Ján Pich and anonymous referees for their helpful suggestions that improve the presentation of this paper. %

{\small \bibliography{main}}

\newcommand{\etalchar}[1]{$^{#1}$}
\begin{thebibliography}{dRMN{\etalchar{+}}20}

\bibitem[AB87]{AlonB87}
Noga Alon and Ravi~B. Boppana.
\newblock The monotone circuit complexity of {B}oolean functions.
\newblock {\em Comb.}, 7(1):1--22, 1987.
\newblock \href {https://doi.org/10.1007/BF02579196}
  {\path{doi:10.1007/BF02579196}}.

\bibitem[ABM23]{AtseriasBM23}
Albert Atserias, Sam Buss, and Moritz M{\"{u}}ller.
\newblock On the consistency of circuit lower bounds for non-deterministic
  time.
\newblock In {\em {STOC}}, pages 1257--1270. {ACM}, 2023.
\newblock \href {https://doi.org/10.1145/3564246.3585253}
  {\path{doi:10.1145/3564246.3585253}}.

\bibitem[AKPS24]{ArtecheKPS24}
Noel Arteche, Erfan Khaniki, J{\'{a}}n Pich, and Rahul Santhanam.
\newblock From proof complexity to circuit complexity via interactive
  protocols.
\newblock In {\em {ICALP}}, volume 297 of {\em LIPIcs}, pages 12:1--12:20.
  Schloss Dagstuhl - Leibniz-Zentrum f{\"{u}}r Informatik, 2024.
\newblock \href {https://doi.org/10.4230/LIPICS.ICALP.2024.12}
  {\path{doi:10.4230/LIPICS.ICALP.2024.12}}.

\bibitem[ALWZ21]{ALWZ-sunflower}
Ryan Alweiss, Shachar Lovett, Kewen Wu, and Jiapeng Zhang.
\newblock {Improved bounds for the sunflower lemma}.
\newblock {\em Annals of Mathematics}, 194(3):795--815, 2021.
\newblock \href {https://doi.org/10.4007/annals.2021.194.3.5}
  {\path{doi:10.4007/annals.2021.194.3.5}}.

\bibitem[AM20]{AtseriasM19}
Albert Atserias and Moritz M{\"{u}}ller.
\newblock Automating resolution is {$\NP$}-hard.
\newblock {\em J. {ACM}}, 67(5):31:1--31:17, 2020.
\newblock \href {https://doi.org/10.1145/3409472} {\path{doi:10.1145/3409472}}.

\bibitem[AT14]{AtseriasT14}
Albert Atserias and Neil Thapen.
\newblock The ordering principle in a fragment of approximate counting.
\newblock {\em {ACM} Trans. Comput. Log.}, 15(4):29:1--29:11, 2014.
\newblock \href {https://doi.org/10.1145/2629555} {\path{doi:10.1145/2629555}}.

\bibitem[Ats03]{Atserias03}
Albert Atserias.
\newblock Improved bounds on the weak pigeonhole principle and infinitely many
  primes from weaker axioms.
\newblock {\em Theor. Comput. Sci.}, 295:27--39, 2003.
\newblock \href {https://doi.org/10.1016/S0304-3975(02)00394-8}
  {\path{doi:10.1016/S0304-3975(02)00394-8}}.

\bibitem[BB17]{BeckmannB17}
Arnold Beckmann and Sam Buss.
\newblock The {$\NP$} search problems of {F}rege and {E}xtended {F}rege proofs.
\newblock {\em {ACM} Trans. Comput. Log.}, 18(2):11:1--11:19, 2017.
\newblock \href {https://doi.org/10.1145/3060145} {\path{doi:10.1145/3060145}}.

\bibitem[BCE{\etalchar{+}}98]{BeameCEIP98}
Paul Beame, Stephen~A. Cook, Jeff Edmonds, Russell Impagliazzo, and Toniann
  Pitassi.
\newblock The relative complexity of {$\NP$} search problems.
\newblock {\em J. Comput. Syst. Sci.}, 57(1):3--19, 1998.
\newblock \href {https://doi.org/10.1006/JCSS.1998.1575}
  {\path{doi:10.1006/JCSS.1998.1575}}.

\bibitem[Bel20]{Bell20}
Zo{\"{e}} Bell.
\newblock Automating regular or ordered resolution is {$\NP$}-hard.
\newblock {\em Electron. Colloquium Comput. Complex.}, {TR20-105}, 2020.
\newblock URL: \url{https://eccc.weizmann.ac.il/report/2020/105}.

\bibitem[BFI23]{BFI23}
Sam Buss, Noah Fleming, and Russell Impagliazzo.
\newblock {$\TFNP$} characterizations of proof systems and monotone circuits.
\newblock In {\em {ITCS}}, volume 251 of {\em LIPIcs}, pages 30:1--30:40.
  Schloss Dagstuhl - Leibniz-Zentrum f{\"{u}}r Informatik, 2023.
\newblock \href {https://doi.org/10.4230/LIPIcs.ITCS.2023.30}
  {\path{doi:10.4230/LIPIcs.ITCS.2023.30}}.

\bibitem[BIK{\etalchar{+}}94]{BeameIKPP94}
Paul Beame, Russell Impagliazzo, Jan Kraj\'{\i}\v{c}ek, Toniann Pitassi, and
  Pavel Pudl{\'{a}}k.
\newblock Lower bound on {Hilbert's Nullstellensatz} and propositional proofs.
\newblock In {\em {FOCS}}, pages 794--806. {IEEE} Computer Society, 1994.
\newblock \href {https://doi.org/10.1109/SFCS.1994.365714}
  {\path{doi:10.1109/SFCS.1994.365714}}.

\bibitem[BIK{\etalchar{+}}97]{BussIPRS97}
Samuel~R. Buss, Russell Impagliazzo, Jan Kraj\'{\i}\v{c}ek, Pavel Pudl{\'{a}}k,
  Alexander~A. Razborov, and Jir{\'{\i}} Sgall.
\newblock Proof complexity in algebraic systems and bounded depth {F}rege
  systems with modular counting.
\newblock {\em Comput. Complex.}, 6(3):256--298, 1997.
\newblock \href {https://doi.org/10.1007/BF01294258}
  {\path{doi:10.1007/BF01294258}}.

\bibitem[BJ12]{BussJ12}
Samuel~R. Buss and Alan~S. Johnson.
\newblock Propositional proofs and reductions between {$\NP$} search problems.
\newblock {\em Ann. Pure Appl. Log.}, 163(9):1163--1182, 2012.
\newblock \href {https://doi.org/10.1016/J.APAL.2012.01.015}
  {\path{doi:10.1016/J.APAL.2012.01.015}}.

\bibitem[BK94]{BussKrajicek94}
Samuel~R. Buss and Jan Krajíček.
\newblock An application of {B}oolean complexity to separation problems in
  bounded arithmetic.
\newblock {\em Proceedings of the London Mathematical Society}, s3-69(1):1--21,
  1994.
\newblock \href {https://doi.org/10.1112/plms/s3-69.1.1}
  {\path{doi:10.1112/plms/s3-69.1.1}}.

\bibitem[BKO20]{BydzovskyKO20}
Jan Bydzovsky, Jan Kraj{\'{\i}}cek, and Igor~C. Oliveira.
\newblock Consistency of circuit lower bounds with bounded theories.
\newblock {\em Log. Methods Comput. Sci.}, 16(2), 2020.
\newblock \href {https://doi.org/10.23638/LMCS-16(2:12)2020}
  {\path{doi:10.23638/LMCS-16(2:12)2020}}.

\bibitem[BKT14]{BussKT14}
Samuel~R. Buss, Leszek~Aleksander Kołodziejczyk, and Neil Thapen.
\newblock Fragments of approximate counting.
\newblock {\em J. Symb. Log.}, 79(2):496--525, 2014.
\newblock \href {https://doi.org/10.1017/JSL.2013.37}
  {\path{doi:10.1017/JSL.2013.37}}.

\bibitem[BKZ15]{buss2015collapsing}
Samuel~R. Buss, Leszek~Aleksander Ko{\l}odziejczyk, and Konrad Zdanowski.
\newblock Collapsing modular counting in bounded arithmetic and constant depth
  propositional proofs.
\newblock {\em Trans. Amer. Math. Soc.}, 367(11):7517--7563, 2015.
\newblock \href {https://doi.org/10.1090/S0002-9947-2015-06233-3}
  {\path{doi:10.1090/S0002-9947-2015-06233-3}}.

\bibitem[BM04]{B-OM04}
Josh Buresh{-}Oppenheim and Tsuyoshi Morioka.
\newblock Relativized {$\NP$} search problems and propositional proof systems.
\newblock In {\em {CCC}}, pages 54--67. {IEEE} Computer Society, 2004.
\newblock \href {https://doi.org/10.1109/CCC.2004.1313795}
  {\path{doi:10.1109/CCC.2004.1313795}}.

\bibitem[BM20]{BydzovskyM20}
Jan Bydzovsky and Moritz M{\"{u}}ller.
\newblock Polynomial time ultrapowers and the consistency of circuit lower
  bounds.
\newblock {\em Arch. Math. Log.}, 59(1-2):127--147, 2020.
\newblock \href {https://doi.org/10.1007/S00153-019-00681-Y}
  {\path{doi:10.1007/S00153-019-00681-Y}}.

\bibitem[BP96]{beame1996simplified}
Paul Beame and Toniann Pitassi.
\newblock Simplified and improved resolution lower bounds.
\newblock In {\em {FOCS}}, pages 274--282. IEEE, 1996.
\newblock \href {https://doi.org/10.1109/SFCS.1996.548486}
  {\path{doi:10.1109/SFCS.1996.548486}}.

\bibitem[Bro11]{Brouwer}
L.~E.~J. Brouwer.
\newblock Über abbildung von mannigfaltigkeiten.
\newblock {\em Mathematische Annalen}, 71:97--115, 1911.
\newblock In German.
\newblock \href {https://doi.org/10.1007/BF01456931}
  {\path{doi:10.1007/BF01456931}}.

\bibitem[Bus85]{Buss85}
Samuel~R. Buss.
\newblock {\em Bounded arithmetic}.
\newblock Princeton University, 1985.

\bibitem[BW01]{Ben-SassonW01}
Eli Ben{-}Sasson and Avi Wigderson.
\newblock Short proofs are narrow - resolution made simple.
\newblock {\em J. {ACM}}, 48(2):149--169, 2001.
\newblock \href {https://doi.org/10.1145/375827.375835}
  {\path{doi:10.1145/375827.375835}}.

\bibitem[CDT09]{chen2009settling}
Xi~Chen, Xiaotie Deng, and Shang-Hua Teng.
\newblock Settling the complexity of computing two-player {N}ash equilibria.
\newblock {\em J. {ACM}}, 56(3):1--57, 2009.
\newblock \href {https://doi.org/10.1145/1516512.1516516}
  {\path{doi:10.1145/1516512.1516516}}.

\bibitem[CEI96]{clegg1996using}
Matthew Clegg, Jeffery Edmonds, and Russell Impagliazzo.
\newblock Using the {G}roebner basis algorithm to find proofs of
  unsatisfiability.
\newblock In {\em {STOC}}, pages 174--183, 1996.
\newblock \href {https://doi.org/10.1145/237814.237860}
  {\path{doi:10.1145/237814.237860}}.

\bibitem[CJSW24]{CJSW21}
Lijie Chen, Ce~Jin, Rahul Santhanam, and Ryan Williams.
\newblock Constructive separations and their consequences.
\newblock {\em {TheoretiCS}}, {volume 3}, February 2024.
\newblock \href {https://doi.org/10.46298/theoretics.24.3}
  {\path{doi:10.46298/theoretics.24.3}}.

\bibitem[CK07]{CookK07}
Stephen~A. Cook and Jan Kraj\'{\i}\v{c}ek.
\newblock Consequences of the provability of {$\NP\subseteq \P/_\poly$}.
\newblock {\em J. Symb. Log.}, 72(4):1353--1371, 2007.
\newblock \href {https://doi.org/10.2178/JSL/1203350791}
  {\path{doi:10.2178/JSL/1203350791}}.

\bibitem[CKKO21]{CarmosinoKKO21}
Marco Carmosino, Valentine Kabanets, Antonina Kolokolova, and Igor~C. Oliveira.
\newblock {LEARN}-uniform circuit lower bounds and provability in bounded
  arithmetic.
\newblock In {\em {FOCS}}, pages 770--780. {IEEE}, 2021.
\newblock \href {https://doi.org/10.1109/FOCS52979.2021.00080}
  {\path{doi:10.1109/FOCS52979.2021.00080}}.

\bibitem[CLO24]{ChenLiOliveira24}
Lijie Chen, Jiatu Li, and Igor~C. Oliveira.
\newblock Reverse mathematics of complexity lower bounds.
\newblock In {\em {FOCS}}, pages 505--527. {IEEE}, 2024.
\newblock \href {https://doi.org/10.1109/FOCS61266.2024.00040}
  {\path{doi:10.1109/FOCS61266.2024.00040}}.

\bibitem[CLO25]{CLO24b}
Lijie Chen, Jiatu Li, and Igor~Carboni Oliveira.
\newblock On the unprovability of circuit size bounds in intuitionistic
  {$\S^1_2$}.
\newblock {\em Logical Methods in Computer Science}, Volume 21, Issue 3, Sep
  2025.
\newblock \href {https://doi.org/10.46298/lmcs-21(3:26)2025}
  {\path{doi:10.46298/lmcs-21(3:26)2025}}.

\bibitem[CN10]{Cook-Nguyen}
Stephen~A. Cook and Phuong Nguyen.
\newblock {\em Logical Foundations of Proof Complexity}, volume~11.
\newblock Cambridge University Press, 2010.
\newblock \href {https://doi.org/10.1017/CBO9780511676277}
  {\path{doi:10.1017/CBO9780511676277}}.

\bibitem[Cob64]{cobham1964intrinsic}
Alan Cobham.
\newblock The intrinsic computational difficulty of functions.
\newblock In {\em Proc. Logic, Methodology, and the Philosophy of Science},
  pages 24--30, 1964.

\bibitem[Coo75]{Cook75}
Stephen~A. Cook.
\newblock Feasibly constructive proofs and the propositional calculus
  (preliminary version).
\newblock In {\em {STOC}}, pages 83--97. {ACM}, 1975.
\newblock \href {https://doi.org/10.1145/800116.803756}
  {\path{doi:10.1145/800116.803756}}.

\bibitem[Coo07]{Cook07}
Stephen~A. Cook.
\newblock Bounded reverse mathematics, 2007.
\newblock Plenary lecture for CiE 2007.

\bibitem[CP90]{CookP90}
Stephen~A. Cook and Toniann Pitassi.
\newblock A feasibly constructive lower bound for resolution proofs.
\newblock {\em Inf. Process. Lett.}, 34(2):81--85, 1990.
\newblock \href {https://doi.org/10.1016/0020-0190(90)90141-J}
  {\path{doi:10.1016/0020-0190(90)90141-J}}.

\bibitem[CS88]{ChvatalS88}
Vasek Chv{\'{a}}tal and Endre Szemer{\'{e}}di.
\newblock Many hard examples for resolution.
\newblock {\em J. {ACM}}, 35(4):759--768, 1988.
\newblock \href {https://doi.org/10.1145/48014.48016}
  {\path{doi:10.1145/48014.48016}}.

\bibitem[CTW23]{CTW23}
Lijie Chen, Roei Tell, and Ryan Williams.
\newblock Derandomization vs refutation: {A} unified framework for
  characterizing derandomization.
\newblock In {\em {FOCS}}, pages 1008--1047. {IEEE}, 2023.
\newblock \href {https://doi.org/10.1109/FOCS57990.2023.00062}
  {\path{doi:10.1109/FOCS57990.2023.00062}}.

\bibitem[DGP09]{daskalakis2009complexity}
Constantinos Daskalakis, Paul~W. Goldberg, and Christos~H. Papadimitriou.
\newblock The complexity of computing a {N}ash equilibrium.
\newblock {\em {SIAM} J. Comput.}, 39(1):195--259, 2009.
\newblock \href {https://doi.org/10.1137/070699652}
  {\path{doi:10.1137/070699652}}.

\bibitem[DLL62]{DavisLL62}
Martin Davis, George Logemann, and Donald~W. Loveland.
\newblock A machine program for theorem-proving.
\newblock {\em Commun. {ACM}}, 5(7):394--397, 1962.
\newblock \href {https://doi.org/10.1145/368273.368557}
  {\path{doi:10.1145/368273.368557}}.

\bibitem[DP60]{DavisP60}
Martin Davis and Hilary Putnam.
\newblock A computing procedure for quantification theory.
\newblock {\em J. {ACM}}, 7(3):201--215, 1960.
\newblock \href {https://doi.org/10.1145/321033.321034}
  {\path{doi:10.1145/321033.321034}}.

\bibitem[DR03]{DantchevR03}
Stefan~S. Dantchev and S{\o}ren Riis.
\newblock On relativisation and complexity gap for resolution-based proof
  systems.
\newblock In {\em {CSL}}, volume 2803 of {\em Lecture Notes in Computer
  Science}, pages 142--154. Springer, 2003.
\newblock \href {https://doi.org/10.1007/978-3-540-45220-1\_14}
  {\path{doi:10.1007/978-3-540-45220-1\_14}}.

\bibitem[dRGN{\etalchar{+}}21]{RezendeGNPR021}
Susanna~F. de~Rezende, Mika G{\"{o}}{\"{o}}s, Jakob Nordstr{\"{o}}m, Toniann
  Pitassi, Robert Robere, and Dmitry Sokolov.
\newblock Automating algebraic proof systems is {$\NP$}-hard.
\newblock In {\em {STOC}}, pages 209--222. {ACM}, 2021.
\newblock \href {https://doi.org/10.1145/3406325.3451080}
  {\path{doi:10.1145/3406325.3451080}}.

\bibitem[dRGR22]{RezendeGR22survey}
Susanna~F. de~Rezende, Mika G{\"{o}}{\"{o}}s, and Robert Robere.
\newblock Proofs, circuits, and communication.
\newblock {\em {SIGACT} News}, 53(1):59--82, 2022.
\newblock \href {https://doi.org/10.1145/3532737.3532746}
  {\path{doi:10.1145/3532737.3532746}}.

\bibitem[dRMN{\etalchar{+}}20]{de2020lifting}
Susanna~F. de~Rezende, Or~Meir, Jakob Nordstr{\"o}m, Toniann Pitassi, Robert
  Robere, and Marc Vinyals.
\newblock Lifting with simple gadgets and applications to circuit and proof
  complexity.
\newblock In {\em {FOCS}}, pages 24--30. {IEEE}, 2020.
\newblock \href {https://doi.org/10.1109/FOCS46700.2020.00011}
  {\path{doi:10.1109/FOCS46700.2020.00011}}.

\bibitem[dRNV16]{de2016limited}
Susanna~F. de~Rezende, Jakob Nordstr{\"o}m, and Marc Vinyals.
\newblock How limited interaction hinders real communication (and what it means
  for proof and circuit complexity).
\newblock In {\em {FOCS}}, pages 295--304. {IEEE}, 2016.
\newblock \href {https://doi.org/10.1109/FOCS.2016.40}
  {\path{doi:10.1109/FOCS.2016.40}}.

\bibitem[Ebt23]{ebtehaj2023variants}
Mohammad~Hossein Ebtehaj.
\newblock Variants of pseudo-deterministic algorithms and duality in {$\TFNP$}.
\newblock Master's thesis, University of Waterloo, 2023.
\newblock URL: \url{http://hdl.handle.net/10012/19721}.

\bibitem[ER60]{erdos1960intersection}
Paul Erd{\"o}s and Richard Rado.
\newblock Intersection theorems for systems of sets.
\newblock {\em Journal of the London Mathematical Society}, 1(1):85--90, 1960.
\newblock \href {https://doi.org/10.1112/jlms/s1-35.1.85}
  {\path{doi:10.1112/jlms/s1-35.1.85}}.

\bibitem[FGPR24]{PPP-Turing}
Noah Fleming, Stefan Grosser, Toniann Pitassi, and Robert Robere.
\newblock Black-box {$\PPP$} is not {T}uring-closed.
\newblock In {\em {STOC}}, pages 1405--1414. {ACM}, 2024.
\newblock \href {https://doi.org/10.1145/3618260.3649769}
  {\path{doi:10.1145/3618260.3649769}}.

\bibitem[Gar19]{Garlik19}
Michal Garl{\'{\i}}k.
\newblock Resolution lower bounds for refutation statements.
\newblock In {\em {MFCS}}, volume 138 of {\em LIPIcs}, pages 37:1--37:13.
  Schloss Dagstuhl - Leibniz-Zentrum f{\"{u}}r Informatik, 2019.
\newblock \href {https://doi.org/10.4230/LIPICS.MFCS.2019.37}
  {\path{doi:10.4230/LIPICS.MFCS.2019.37}}.

\bibitem[Gar24]{Garlik24}
Michal Garl{\'{\i}}k.
\newblock Failure of feasible disjunction property for {$k$-DNF} resolution and
  {$\NP$}-hardness of automating it.
\newblock In {\em {CCC}}, volume 300 of {\em LIPIcs}, pages 33:1--33:23.
  Schloss Dagstuhl - Leibniz-Zentrum f{\"{u}}r Informatik, 2024.
\newblock \href {https://doi.org/10.4230/LIPICS.CCC.2024.33}
  {\path{doi:10.4230/LIPICS.CCC.2024.33}}.

\bibitem[GGKS20]{GGKS}
Ankit Garg, Mika G{\"{o}}{\"{o}}s, Pritish Kamath, and Dmitry Sokolov.
\newblock Monotone circuit lower bounds from resolution.
\newblock {\em Theory Comput.}, 16:1--30, 2020.
\newblock \href {https://doi.org/10.4086/TOC.2020.V016A013}
  {\path{doi:10.4086/TOC.2020.V016A013}}.

\bibitem[GHJ{\etalchar{+}}22]{GHMPRT22separation}
Mika G{\"{o}}{\"{o}}s, Alexandros Hollender, Siddhartha Jain, Gilbert Maystre,
  William Pires, Robert Robere, and Ran Tao.
\newblock Separations in proof complexity and {$\TFNP$}.
\newblock In {\em {FOCS}}, pages 1150--1161. {IEEE}, 2022.
\newblock \href {https://doi.org/10.1109/FOCS54457.2022.00111}
  {\path{doi:10.1109/FOCS54457.2022.00111}}.

\bibitem[GKMP20]{GoosKMP20}
Mika G{\"{o}}{\"{o}}s, Sajin Koroth, Ian Mertz, and Toniann Pitassi.
\newblock Automating cutting planes is {$\NP$}-hard.
\newblock In {\em {STOC}}, pages 68--77. {ACM}, 2020.
\newblock \href {https://doi.org/10.1145/3357713.3384248}
  {\path{doi:10.1145/3357713.3384248}}.

\bibitem[GP18a]{GoldbergP18}
Paul~W. Goldberg and Christos~H. Papadimitriou.
\newblock Towards a unified complexity theory of total functions.
\newblock {\em J. Comput. Syst. Sci.}, 94:167--192, 2018.
\newblock \href {https://doi.org/10.1016/J.JCSS.2017.12.003}
  {\path{doi:10.1016/J.JCSS.2017.12.003}}.

\bibitem[GP18b]{goos2014communication}
Mika G{\"{o}}{\"{o}}s and Toniann Pitassi.
\newblock Communication lower bounds via critical block sensitivity.
\newblock {\em {SIAM} J. Comput.}, 47(5):1778--1806, 2018.
\newblock \href {https://doi.org/10.1137/16M1082007}
  {\path{doi:10.1137/16M1082007}}.

\bibitem[GST07]{GutfreundST07}
Dan Gutfreund, Ronen Shaltiel, and Amnon Ta{-}Shma.
\newblock If {$\NP$} languages are hard on the worst-case, then it is easy to
  find their hard instances.
\newblock {\em Comput. Complex.}, 16(4):412--441, 2007.
\newblock \href {https://doi.org/10.1007/S00037-007-0235-8}
  {\path{doi:10.1007/S00037-007-0235-8}}.

\bibitem[Hak85]{Haken85}
Armin Haken.
\newblock The intractability of resolution.
\newblock {\em Theor. Comput. Sci.}, 39:297--308, 1985.
\newblock \href {https://doi.org/10.1016/0304-3975(85)90144-6}
  {\path{doi:10.1016/0304-3975(85)90144-6}}.

\bibitem[Hak95]{Haken95}
Armin Haken.
\newblock Counting bottlenecks to show monotone {$\P\ne\NP$}.
\newblock In {\em {FOCS}}, pages 36--40. {IEEE} Computer Society, 1995.
\newblock \href {https://doi.org/10.1109/SFCS.1995.492460}
  {\path{doi:10.1109/SFCS.1995.492460}}.

\bibitem[HN12]{huynh2012virtue}
Trinh Huynh and Jakob Nordstr{\"{o}}m.
\newblock On the virtue of succinct proofs: Amplifying communication complexity
  hardness to time-space trade-offs in proof complexity.
\newblock In {\em {STOC}}, pages 233--248. {ACM}, 2012.
\newblock \href {https://doi.org/10.1145/2213977.2214000}
  {\path{doi:10.1145/2213977.2214000}}.

\bibitem[ILW23]{IlangoLW23}
Rahul Ilango, Jiatu Li, and R.~Ryan Williams.
\newblock Indistinguishability obfuscation, range avoidance, and bounded
  arithmetic.
\newblock In {\em {STOC}}, pages 1076--1089. {ACM}, 2023.
\newblock \href {https://doi.org/10.1145/3564246.3585187}
  {\path{doi:10.1145/3564246.3585187}}.

\bibitem[IR22]{ItsyksonR22}
Dmitry Itsykson and Artur Riazanov.
\newblock Automating {OBDD} proofs is {$\NP$}-hard.
\newblock In {\em {MFCS}}, volume 241 of {\em LIPIcs}, pages 59:1--59:15.
  Schloss Dagstuhl - Leibniz-Zentrum f{\"{u}}r Informatik, 2022.
\newblock \href {https://doi.org/10.4230/LIPICS.MFCS.2022.59}
  {\path{doi:10.4230/LIPICS.MFCS.2022.59}}.

\bibitem[Je{\v{r}}04]{Jerabek04}
Emil Je{\v{r}}{\'{a}}bek.
\newblock Dual weak pigeonhole principle, {Boolean} complexity, and
  derandomization.
\newblock {\em Ann. Pure Appl. Log.}, 129(1-3):1--37, 2004.
\newblock \href {https://doi.org/10.1016/j.apal.2003.12.003}
  {\path{doi:10.1016/j.apal.2003.12.003}}.

\bibitem[Je{\v{r}}07a]{Jerabek-APC1}
Emil Je{\v{r}}{\'{a}}bek.
\newblock Approximate counting in bounded arithmetic.
\newblock {\em J. Symb. Log.}, 72(3):959--993, 2007.
\newblock \href {https://doi.org/10.2178/JSL/1191333850}
  {\path{doi:10.2178/JSL/1191333850}}.

\bibitem[Je{\v{r}}07b]{Jerabek-independence}
Emil Je{\v{r}}{\'{a}}bek.
\newblock On independence of variants of the weak pigeonhole principle.
\newblock {\em J. Log. Comput.}, 17(3):587--604, 2007.
\newblock \href {https://doi.org/10.1093/LOGCOM/EXM017}
  {\path{doi:10.1093/LOGCOM/EXM017}}.

\bibitem[Je{\v{r}}09]{Jerabek-APC2}
Emil Je{\v{r}}{\'{a}}bek.
\newblock Approximate counting by hashing in bounded arithmetic.
\newblock {\em J. Symb. Log.}, 74(3):829--860, 2009.
\newblock \href {https://doi.org/10.2178/JSL/1245158087}
  {\path{doi:10.2178/JSL/1245158087}}.

\bibitem[JPY88]{DBLP:journals/jcss/JohnsonPY88}
David~S. Johnson, Christos~H. Papadimitriou, and Mihalis Yannakakis.
\newblock How easy is local search?
\newblock {\em J. Comput. Syst. Sci.}, 37(1):79--100, 1988.
\newblock \href {https://doi.org/10.1016/0022-0000(88)90046-3}
  {\path{doi:10.1016/0022-0000(88)90046-3}}.

\bibitem[Kam19]{Pritish_Kamath_PhD}
Pritish Kamath.
\newblock {\em Some hardness escalation results in computational complexity
  theory}.
\newblock PhD thesis, Massachusetts Institute of Technology, 2019.
\newblock URL: \url{https://hdl.handle.net/1721.1/128290}.

\bibitem[KKMP21]{KKMP21}
Robert Kleinberg, Oliver Korten, Daniel Mitropolsky, and Christos~H.
  Papadimitriou.
\newblock Total functions in the polynomial hierarchy.
\newblock In {\em ITCS}, volume 185 of {\em LIPIcs}, pages 44:1--44:18, 2021.
\newblock \href {https://doi.org/10.4230/LIPIcs.ITCS.2021.44}
  {\path{doi:10.4230/LIPIcs.ITCS.2021.44}}.

\bibitem[KNT11]{KolodziejczykNT11}
Leszek~Aleksander Ko{\l}odziejczyk, Phuong Nguyen, and Neil Thapen.
\newblock The provably total {$\NP$} search problems of weak second order
  bounded arithmetic.
\newblock {\em Ann. Pure Appl. Log.}, 162(6):419--446, 2011.
\newblock \href {https://doi.org/10.1016/J.APAL.2010.12.002}
  {\path{doi:10.1016/J.APAL.2010.12.002}}.

\bibitem[KO17]{KrajicekO16}
Jan Kraj\'{\i}\v{c}ek and Igor~C. Oliveira.
\newblock Unprovability of circuit upper bounds in {C}ook's theory {$\PV$}.
\newblock {\em Log. Methods Comput. Sci.}, 13(1), 2017.
\newblock \href {https://doi.org/10.23638/LMCS-13(1:4)2017}
  {\path{doi:10.23638/LMCS-13(1:4)2017}}.

\bibitem[Kor21]{Korten21}
Oliver Korten.
\newblock The hardest explicit construction.
\newblock In {\em FOCS}, pages 433--444. {IEEE}, 2021.
\newblock \href {https://doi.org/10.1109/FOCS52979.2021.00051}
  {\path{doi:10.1109/FOCS52979.2021.00051}}.

\bibitem[Kor22]{Korten22}
Oliver Korten.
\newblock Derandomization from time-space tradeoffs.
\newblock In {\em {CCC}}, volume 234 of {\em LIPIcs}, pages 37:1--37:26.
  Schloss Dagstuhl - Leibniz-Zentrum f{\"{u}}r Informatik, 2022.
\newblock \href {https://doi.org/10.4230/LIPIcs.CCC.2022.37}
  {\path{doi:10.4230/LIPIcs.CCC.2022.37}}.

\bibitem[KP89]{KrajicekP89}
Jan Kraj{\'{\i}}{\v{c}}ek and Pavel Pudl{\'{a}}k.
\newblock Propositional provability and models of weak arithmetic.
\newblock In {\em {CSL}}, volume 440 of {\em Lecture Notes in Computer
  Science}, pages 193--210. Springer, 1989.
\newblock \href {https://doi.org/10.1007/3-540-52753-2\_40}
  {\path{doi:10.1007/3-540-52753-2\_40}}.

\bibitem[Kra95]{Krajicek_BA_PL_CT}
Jan Kraj\'{\i}\v{c}ek.
\newblock {\em Bounded arithmetic, propositional logic, and complexity theory},
  volume~60 of {\em Encyclopedia of mathematics and its applications}.
\newblock Cambridge University Press, 1995.
\newblock \href {https://doi.org/10.1017/CBO9780511529948}
  {\path{doi:10.1017/CBO9780511529948}}.

\bibitem[Kra97]{Krajicek97}
Jan Kraj{\'{\i}}\v{c}ek.
\newblock Interpolation theorems, lower bounds for proof systems, and
  independence results for bounded arithmetic.
\newblock {\em J. Symb. Log.}, 62(2):457--486, 1997.
\newblock \href {https://doi.org/10.2307/2275541} {\path{doi:10.2307/2275541}}.

\bibitem[Kra01]{Krajicek01b}
Jan Krajíček.
\newblock On the weak pigeonhole principle.
\newblock {\em Fundamenta Mathematicae}, 170(1-2):123--140, 2001.
\newblock URL: \url{http://eudml.org/doc/282141}.

\bibitem[Kra04]{Krajicek04a}
Jan Kraj\'{\i}\v{c}ek.
\newblock Dual weak pigeonhole principle, pseudo-surjective functions, and
  provability of circuit lower bounds.
\newblock {\em J. Symb. Log.}, 69(1):265--286, 2004.
\newblock \href {https://doi.org/10.2178/jsl/1080938841}
  {\path{doi:10.2178/jsl/1080938841}}.

\bibitem[Kra11a]{Krajicek11-Ramsey}
Jan Kraj\'{\i}\v{c}ek.
\newblock A note on propositional proof complexity of some {R}amsey-type
  statements.
\newblock {\em Arch. Math. Log.}, 50(1-2):245--255, 2011.
\newblock \href {https://doi.org/10.1007/S00153-010-0212-9}
  {\path{doi:10.1007/S00153-010-0212-9}}.

\bibitem[Kra11b]{Krajicek11}
Jan Kraj\'{\i}\v{c}ek.
\newblock On the proof complexity of the {Nisan-Wigderson} generator based on a
  hard {$\NP \cap\coNP$} function.
\newblock {\em J. Math. Log.}, 11(1), 2011.
\newblock \href {https://doi.org/10.1142/S0219061311000979}
  {\path{doi:10.1142/S0219061311000979}}.

\bibitem[Kra19]{krajicek_proof_complexity}
Jan Krajíček.
\newblock {\em Proof Complexity}.
\newblock Encyclopedia of Mathematics and its Applications. Cambridge
  University Press, 2019.
\newblock \href {https://doi.org/10.1017/9781108242066}
  {\path{doi:10.1017/9781108242066}}.

\bibitem[KST07]{KrajicekST07}
Jan Kraj\'{\i}\v{c}ek, Alan Skelley, and Neil Thapen.
\newblock {$\NP$} search problems in low fragments of bounded arithmetic.
\newblock {\em J. Symb. Log.}, 72(2):649--672, 2007.
\newblock \href {https://doi.org/10.2178/JSL/1185803628}
  {\path{doi:10.2178/JSL/1185803628}}.

\bibitem[KT22]{KolodziejczykT22}
Leszek~Aleksander Kołodziejczyk and Neil Thapen.
\newblock Approximate counting and {$\NP$} search problems.
\newblock {\em J. Math. Log.}, 22(3):2250012:1--2250012:31, 2022.
\newblock \href {https://doi.org/10.1142/S021906132250012X}
  {\path{doi:10.1142/S021906132250012X}}.

\bibitem[LC11]{Le-Cook}
Dai Tri~Man Le and Stephen~A. Cook.
\newblock Formalizing randomized matching algorithms.
\newblock {\em Log. Methods Comput. Sci.}, 8(3), 2011.
\newblock \href {https://doi.org/10.2168/LMCS-8(3:5)2012}
  {\path{doi:10.2168/LMCS-8(3:5)2012}}.

\bibitem[LO23]{LiO23}
Jiatu Li and Igor~C. Oliveira.
\newblock Unprovability of strong complexity lower bounds in bounded
  arithmetic.
\newblock In {\em {STOC}}, pages 1051--1057. {ACM}, 2023.
\newblock \href {https://doi.org/10.1145/3564246.3585144}
  {\path{doi:10.1145/3564246.3585144}}.

\bibitem[Maa84]{Maass84}
Wolfgang Maass.
\newblock Quadratic lower bounds for deterministic and nondeterministic
  one-tape {T}uring machines (extended abstract).
\newblock In {\em {STOC}}, pages 401--408. {ACM}, 1984.
\newblock \href {https://doi.org/10.1145/800057.808706}
  {\path{doi:10.1145/800057.808706}}.

\bibitem[Mor01]{morioka2001classification}
Tsuyoshi Morioka.
\newblock Classification of search problems and their definability in bounded
  arithmetic.
\newblock Master's thesis, University of Toronto, 2001.
\newblock URL: \url{https://hdl.handle.net/1807/16458}.

\bibitem[MP91]{MegiddoP91}
Nimrod Megiddo and Christos~H. Papadimitriou.
\newblock On total functions, existence theorems and computational complexity.
\newblock {\em Theor. Comput. Sci.}, 81(2):317--324, 1991.
\newblock \href {https://doi.org/10.1016/0304-3975(91)90200-L}
  {\path{doi:10.1016/0304-3975(91)90200-L}}.

\bibitem[MP96]{MacielP96}
Alexis Maciel and Toniann Pitassi.
\newblock Towards lower bounds for bounded-depth {F}rege proofs with modular
  connectives.
\newblock In {\em Proof Complexity and Feasible Arithmetics}, volume~39 of {\em
  {DIMACS} Series in Discrete Mathematics and Theoretical Computer Science},
  pages 195--227. {DIMACS/AMS}, 1996.
\newblock \href {https://doi.org/10.1090/DIMACS/039/12}
  {\path{doi:10.1090/DIMACS/039/12}}.

\bibitem[MP20]{MullerP20}
Moritz M{\"{u}}ller and J{\'{a}}n Pich.
\newblock Feasibly constructive proofs of succinct weak circuit lower bounds.
\newblock {\em Ann. Pure Appl. Log.}, 171(2), 2020.
\newblock \href {https://doi.org/10.1016/j.apal.2019.102735}
  {\path{doi:10.1016/j.apal.2019.102735}}.

\bibitem[MPW02]{MacielPW02}
Alexis Maciel, Toniann Pitassi, and Alan~R. Woods.
\newblock A new proof of the weak pigeonhole principle.
\newblock {\em J. Comput. Syst. Sci.}, 64(4):843--872, 2002.
\newblock \href {https://doi.org/10.1006/JCSS.2002.1830}
  {\path{doi:10.1006/JCSS.2002.1830}}.

\bibitem[M{\"{u}}l21]{Muller21}
Moritz M{\"{u}}ller.
\newblock Typical forcings, {$\NP$} search problems and an extension of a
  theorem of {R}iis.
\newblock {\em Ann. Pure Appl. Log.}, 172(4):102930, 2021.
\newblock \href {https://doi.org/10.1016/J.APAL.2020.102930}
  {\path{doi:10.1016/J.APAL.2020.102930}}.

\bibitem[Nas51]{nash1951non}
John Nash.
\newblock Non-cooperative games.
\newblock {\em Annals of mathematics}, 54(2):286--295, 1951.
\newblock \href {https://doi.org/10.2307/1969529} {\path{doi:10.2307/1969529}}.

\bibitem[Ngu08]{Nguyen-PhD}
Phuong Nguyen.
\newblock {\em Bounded reverse mathematics}.
\newblock PhD thesis, University of Toronto, 2008.

\bibitem[Pap94]{Papadimitriou94}
Christos~H. Papadimitriou.
\newblock On the complexity of the parity argument and other inefficient proofs
  of existence.
\newblock {\em J. Comput. Syst. Sci.}, 48(3):498--532, 1994.
\newblock \href {https://doi.org/10.1016/S0022-0000(05)80063-7}
  {\path{doi:10.1016/S0022-0000(05)80063-7}}.

\bibitem[Pap24]{Papamakarios24}
Theodoros Papamakarios.
\newblock Depth-{$d$} frege systems are not automatable unless {$\P = \NP$}.
\newblock In {\em {CCC}}, volume 300 of {\em LIPIcs}, pages 22:1--22:17.
  Schloss Dagstuhl - Leibniz-Zentrum f{\"{u}}r Informatik, 2024.
\newblock \href {https://doi.org/10.4230/LIPICS.CCC.2024.22}
  {\path{doi:10.4230/LIPICS.CCC.2024.22}}.

\bibitem[Pic15]{Pich15}
J{\'{a}}n Pich.
\newblock Circuit lower bounds in bounded arithmetics.
\newblock {\em Ann. Pure Appl. Log.}, 166(1):29--45, 2015.
\newblock \href {https://doi.org/10.1016/J.APAL.2014.08.004}
  {\path{doi:10.1016/J.APAL.2014.08.004}}.

\bibitem[PS19]{PichS19}
J{\'{a}}n Pich and Rahul Santhanam.
\newblock Why are proof complexity lower bounds hard?
\newblock In {\em {FOCS}}, pages 1305--1324. {IEEE} Computer Society, 2019.
\newblock \href {https://doi.org/10.1109/FOCS.2019.00080}
  {\path{doi:10.1109/FOCS.2019.00080}}.

\bibitem[PS21]{PichS21}
J{\'{a}}n Pich and Rahul Santhanam.
\newblock Strong co-nondeterministic lower bounds for {$\NP$} cannot be proved
  feasibly.
\newblock In {\em {STOC}}, pages 223--233. {ACM}, 2021.
\newblock \href {https://doi.org/10.1145/3406325.3451117}
  {\path{doi:10.1145/3406325.3451117}}.

\bibitem[PS26]{PichS23}
J{\'{a}}n Pich and Rahul Santhanam.
\newblock Towards {$\P\ne\NP$} from {E}xtended {F}rege lower bounds.
\newblock {\em J. ACM}, 2026.
\newblock \href {https://doi.org/10.1145/3801091} {\path{doi:10.1145/3801091}}.

\bibitem[PT12]{PudlakT12}
Pavel Pudl{\'{a}}k and Neil Thapen.
\newblock Alternating minima and maxima, {N}ash equilibria and {B}ounded
  {A}rithmetic.
\newblock {\em Ann. Pure Appl. Log.}, 163(5):604--614, 2012.
\newblock \href {https://doi.org/10.1016/J.APAL.2011.06.014}
  {\path{doi:10.1016/J.APAL.2011.06.014}}.

\bibitem[PT19]{PudlakT19}
Pavel Pudl{\'{a}}k and Neil Thapen.
\newblock Random resolution refutations.
\newblock {\em Comput. Complex.}, 28(2):185--239, 2019.
\newblock \href {https://doi.org/10.1007/S00037-019-00182-7}
  {\path{doi:10.1007/S00037-019-00182-7}}.

\bibitem[Pud97]{Pudlak97}
Pavel Pudl{\'{a}}k.
\newblock Lower bounds for resolution and cutting plane proofs and monotone
  computations.
\newblock {\em J. Symb. Log.}, 62(3):981--998, 1997.
\newblock \href {https://doi.org/10.2307/2275583} {\path{doi:10.2307/2275583}}.

\bibitem[Pud00]{Pudlak00ProofsAsGames}
Pavel Pudl{\'{a}}k.
\newblock Proofs as games.
\newblock {\em Am. Math. Mon.}, 107(6):541--550, 2000.
\newblock \href {https://doi.org/10.2307/2589349} {\path{doi:10.2307/2589349}}.

\bibitem[Pud20]{Pudlak20}
Pavel Pudl{\'{a}}k.
\newblock Reflection principles, propositional proof systems, and theories.
\newblock {\em arXiv}, abs/2007.14835, 2020.
\newblock \href {https://doi.org/10.48550/arXiv.2007.14835}
  {\path{doi:10.48550/arXiv.2007.14835}}.

\bibitem[PWW88]{ParisWW88}
Jeff~B. Paris, A.~J. Wilkie, and Alan~R. Woods.
\newblock Provability of the pigeonhole principle and the existence of
  infinitely many primes.
\newblock {\em J. Symb. Log.}, 53(4):1235--1244, 1988.
\newblock \href {https://doi.org/10.1017/S0022481200028061}
  {\path{doi:10.1017/S0022481200028061}}.

\bibitem[Raz85]{Razborov_monotone}
Alexander~A. Razborov.
\newblock Lower bounds on the monotone complexity of some {B}oolean function.
\newblock In {\em Soviet Math. Dokl.}, volume~31, pages 354--357, 1985.
\newblock URL: \url{https://www.mathnet.ru/eng/dan9192}.

\bibitem[Raz87]{Razborov87}
Alexander~A. Razborov.
\newblock Lower bounds on the size of bounded depth circuits over a complete
  basis with logical addition.
\newblock {\em Mathematical Notes of the Academy of Sciences of the USSR},
  41(4):333--338, 1987.

\bibitem[Raz95a]{Razborov_feasible_mathematics_II}
Alexander~A. Razborov.
\newblock Bounded arithmetic and lower bounds in {B}oolean complexity.
\newblock In {\em Feasible Mathematics II}, pages 344--386. Birkh{\"a}user
  Boston, 1995.
\newblock \href {https://doi.org/10.1007/978-1-4612-2566-9_12}
  {\path{doi:10.1007/978-1-4612-2566-9_12}}.

\bibitem[Raz95b]{razborov1995unprovability}
Alexander~A. Razborov.
\newblock Unprovability of lower bounds on circuit size in certain fragments of
  bounded arithmetic.
\newblock {\em Izvestiya: mathematics}, 59(1):205, 1995.
\newblock \href {https://doi.org/10.1070/IM1995v059n01ABEH000009}
  {\path{doi:10.1070/IM1995v059n01ABEH000009}}.

\bibitem[Raz98]{Razborov98_PC}
Alexander~A. Razborov.
\newblock Lower bounds for the polynomial calculus.
\newblock {\em Comput. Complex.}, 7(4):291--324, 1998.
\newblock \href {https://doi.org/10.1007/S000370050013}
  {\path{doi:10.1007/S000370050013}}.

\bibitem[Raz15]{Razborov15}
Alexander~A. Razborov.
\newblock Pseudorandom generators hard for {$k$-DNF} resolution and polynomial
  calculus resolution.
\newblock {\em Annals of Mathematics}, 181(2):415--472, 2015.
\newblock \href {https://doi.org/10.4007/annals.2015.181.2.1}
  {\path{doi:10.4007/annals.2015.181.2.1}}.

\bibitem[Sch97]{Schoning97}
Uwe Sch{\"{o}}ning.
\newblock Resolution proofs, exponential bounds, and {K}olmogorov complexity.
\newblock In {\em {MFCS}}, volume 1295 of {\em Lecture Notes in Computer
  Science}, pages 110--116. Springer, 1997.
\newblock \href {https://doi.org/10.1007/BFB0029954}
  {\path{doi:10.1007/BFB0029954}}.

\bibitem[Smo87]{Smolensky87}
Roman Smolensky.
\newblock Algebraic methods in the theory of lower bounds for boolean circuit
  complexity.
\newblock In {\em {STOC}}, pages 77--82. {ACM}, 1987.
\newblock \href {https://doi.org/10.1145/28395.28404}
  {\path{doi:10.1145/28395.28404}}.

\bibitem[ST11]{skelley2011provably}
Alan Skelley and Neil Thapen.
\newblock The provably total search problems of bounded arithmetic.
\newblock {\em Proceedings of the London Mathematical Society},
  103(1):106--138, 2011.
\newblock \href {https://doi.org/10.1112/plms/pdq044}
  {\path{doi:10.1112/plms/pdq044}}.

\bibitem[ST21]{SanthanamT21}
Rahul Santhanam and Iddo Tzameret.
\newblock Iterated lower bound formulas: a diagonalization-based approach to
  proof complexity.
\newblock In {\em {STOC}}, pages 234--247. {ACM}, 2021.
\newblock \href {https://doi.org/10.1145/3406325.3451010}
  {\path{doi:10.1145/3406325.3451010}}.

\bibitem[SW14]{SanthanamW14}
Rahul Santhanam and Ryan Williams.
\newblock On uniformity and circuit lower bounds.
\newblock {\em Comput. Complex.}, 23(2):177--205, 2014.
\newblock \href {https://doi.org/10.1007/S00037-014-0087-Y}
  {\path{doi:10.1007/S00037-014-0087-Y}}.

\bibitem[Tha02]{Thapen-PhD}
Neil Thapen.
\newblock {\em The weak pigeonhole principle in models of bounded arithmetic}.
\newblock PhD thesis, University of Oxford, 2002.

\bibitem[Tse83]{tseitin1983complexity}
Grigori~S Tseitin.
\newblock On the complexity of derivation in propositional calculus.
\newblock {\em Automation of reasoning: 2: Classical papers on computational
  logic 1967--1970}, pages 466--483, 1983.
\newblock \href {https://doi.org/10.1007/978-3-642-81955-1_28}
  {\path{doi:10.1007/978-3-642-81955-1_28}}.

\bibitem[Urq87]{Urquhart87}
Alasdair Urquhart.
\newblock Hard examples for resolution.
\newblock {\em J. {ACM}}, 34(1):209--219, 1987.
\newblock \href {https://doi.org/10.1145/7531.8928}
  {\path{doi:10.1145/7531.8928}}.

\bibitem[VS23]{von2023zero}
Bernhard Von~Stengel.
\newblock Zero-sum games and linear programming duality.
\newblock {\em Mathematics of Operations Research}, 2023.
\newblock \href {https://doi.org/10.1287/moor.2022.0149}
  {\path{doi:10.1287/moor.2022.0149}}.

\end{thebibliography}

\appendix

\section{Amplification for \texorpdfstring{$\rwPHP(\calP)$}{rwPHP(P)}}\label{sec: amplification for wPHP}

In this section, we prove \autoref{informal thm: amplification of rPHP}, showing that the relationship between $N$ and $M$ does not influence the complexity of $\rwPHP(\calP)$ (provided that $M$ and $N$ are not too close to each other), thus our choice of $N = 2M$ is indeed without loss of generality. This result requires $\calP$ to be closed under Turing reductions \cite{BussJ12}, as defined below:

\begin{definition}[Turing Reductions in $\TFNP^\dt$]
    Let $\calP, \calQ$ be problems in $\TFNP^\dt$. We say there is a time-$t$ (uniform) Turing reduction from $\calQ$ to $\calP$ if there is a time-$t$ oracle Turing machine $R^{x, \calP}$ that solves $\calQ$ in the following manner. Let $x\in\{0, 1\}^N$ be the input to $\calQ$. Besides work tapes and a query tape for access to $x$, $R$ has another query tapes for access to a $\calP$ oracle. Each query $q$ to $\calP$ is described as $(1^{t'}, L, M_q^x)$, where $L$ is the length of the query input, and $M_q^x$ is a time-$t'$ Turing machine with query access to $x$. The input to $\calP$ is defined as the $L$-bit string whose $i$-th bit is $M_q^x(i)$. The answer to this query would be any valid solution for this $L$-bit string as an input to $\calP$. Finally, for every $x\in \{0, 1\}^N$ and every valid computational history of $R$ (i.e., every query to $\calP$ is answered correctly), the output of $R$ should be a valid output of $x$ for $\calQ$.
\end{definition}

\begin{assumption}\label{assumption: P is closed under Turing reductions}
    $\calP$ is closed under Turing reductions. More precisely, for some function $\gamma(t)$ (think of $\gamma(t) \le \poly(t)$), if a $\TFNP^\dt$ problem $\calQ$ admits a time-$t$ Turing reduction to $\calP$, then $\calQ$ also admits a uniform depth-$\gamma(t)$ decision tree reduction to $\calP$.
\end{assumption}

Recall that $\rPHP(\calP)_{M\to N}$ denotes the $\rwPHP(\calP)$ problem where the purported ``surjection'' is $f:[M]\to [N]$.

\begin{fact}
	Let $M < N_1 \le N_2$, then there is a depth-$1$ decision tree reduction from $\rPHP(\calP)_{M\to N_2}$ to $\rPHP(\calP)_{M\to N_1}$.
\end{fact}

\begin{theorem}[Formal version of {\autoref{informal thm: amplification of rPHP}}]\label{thm: amplification for rwPHP}
	Let $N\ge 2M$ and $\eps > 0$ be parameters, and let $d := \gamma(O(\eps^{-1}))\cdot \gamma(O(\log \frac{N}{M}))$. If \autoref{assumption: P is closed under Turing reductions} holds for $\calP$, then there is a depth-$d$ decision tree reduction from $\rPHP(\calP)_{M\to (1+\eps)M}$ to $\rPHP(\calP)_{M\to N}$.
\end{theorem}

We prove \autoref{thm: amplification for rwPHP} in two steps: in \autoref{lemma: stretch reduction 1+eps to 2} we reduce $\rwPHP$ with stretch $(1+\eps)$ (i.e., $\rPHP(\calP)_{M\to (1+\eps)M}$) to $\rwPHP$ with stretch $2$, and in \autoref{lemma: stretch reduction 2 to large} we reduce $\rwPHP$ with stretch $2$ to $\rwPHP$ with arbitrarily large stretch. \autoref{thm: amplification for rwPHP} follows easily from \autoref{lemma: stretch reduction 1+eps to 2} and \ref{lemma: stretch reduction 2 to large}.

\def\Lhis{L_{\sf his}}

\begin{lemma}\label{lemma: stretch reduction 1+eps to 2}
	Let $M \ge 1$, $\eps > 0$ be parameters. Suppose that \autoref{assumption: P is closed under Turing reductions} holds for $\calP$. Then there is a depth-$\gamma(O(\eps^{-1}))$ decision tree reduction from $\rPHP(\calP)_{M\to \lfloor (1+\eps) M\rfloor}$ to $\rPHP(\calP)_{M\to 2M}$.
\end{lemma}
\begin{proof}
	Without loss of generality, assume that both $\eps\cdot M$ and $d := 1/\eps$ are integers. Let $(f, \{I_y\}, \{g_y\})$ be an instance of $\rPHP(\calP)_{M\to (1+\eps)M}$ and we want to reduce it to an instance $(f', \{I'_y\}, \{g'_y\})$ of $\rPHP(\calP)_{M\to 2M}$. Recall that:
	\begin{itemize}
		\item $f:[M]\to [(1+\eps)M]$ is the purported ``surjection''.
		\item For every $y\in [(1+\eps)M]$, $I_y$ is a $\calP$ instance where every possible answer $ans$ of $I_y$ is labelled with an integer $g_y(ans) \in [M]$.
		\item The goal is to find some $y\in[(1+\eps)M]$ and a solution $ans$ of $I_y$ such that $f(g_y(ans)) \ne y$.
	\end{itemize}
	
	For every $k\in [d]$ (recall $d = 1/\eps$), define $f_k:[M+k\eps M]\to [M+(k+1)\eps M]$ as the following function: on input $x\in[M+k\eps M]$, if $x < M$, then $f_k(x) := f(x)$; otherwise $f_k(x) := x + \eps M$. The function $f'$ in our reduction is simply $f' := f_{d-1} \circ f_{d-2} \circ \dots \circ f_0$. Intuitively, if (a weak theory thinks that) $f:[M]\to [(1+\eps)M]$ is a surjection, then (it also thinks that) $f':[M] \to [2M]$ is a surjection.
	
	Next we define the instances $\{I'_y\}$ and the functions $\{g'_y\}$. Roughly speaking, the input instance $\{I_y\}$ and $\{g_y\}$ defines a \emph{$\calP$-computable multi-function} (also denoted as) $g:[(1+\eps)M]\to [M]$, which is a purported inverse of $f$. By padding $g$, we obtain $\calP$-computable multi-functions $g_k:[M+(k+1)\eps M]\to [M+k\eps M]$ for each $k\in[d]$, and each $g_k$ is a purported inverse of $f_k$. We compose these multi-functions $g_k$ to obtain a single multi-function $g:[2M]\to [M]$ that can be computed by a Turing reduction to $\calP$. Details follow.
	
	Consider a Turing machine with oracle access to $I_y$ and $\calP$ that on input $y\in [2M]$, operates as follows. Let $y_d := y$. For each $k$ from $d-1$ downto $0$:\begin{itemize}
		\item if $y_{k+1} \ge (1+\eps)M$, then we define $y_k := y_{k+1} - \eps M$;
		\item otherwise we query $\calP$ to obtain a valid answer $ans_k$ for $I_{y_{k+1}}$ and let $y_k := g(ans_k)$.
	\end{itemize}
	Finally, the machine outputs $y_0$ (as a purported preimage of $y$ under $f'$).
	
	The computational history of this Turing machine defines a total search problem $\Lhis$ as follows. The input to $\Lhis$ consists of $(f, \{I_y\}, \{g_y\})$, as well as some $y\in [2M]$. The output consists of a sequence $(ans_0, ans_1, \dots, ans_{d-1})$. Denote $y_d = y$ and
	\[y_k = \begin{cases}
		y_{k+1} - \eps M & \text{if }y_{k+1} \ge (1+\eps)M,\\
		g_{y_{k+1}}(ans_k) & \text{otherwise}
	\end{cases}\]
	for each $k\in[d]$. We accept the output if for every $k$ such that $y_{k+1} < (1+\eps)M$, $ans_k$ is a valid solution for $I_{y_{k+1}}$; otherwise we reject the output.
	
	Clearly, the above Turing machine itself is a time-$O(d)$ Turing reduction from $\Lhis$ to $\calP$. Since $\calP$ is closed under Turing reductions, there is also a depth-$\gamma(O(d))$ mapping reduction from $\Lhis$ to $\calP$. That is, there is a depth-$\gamma(O(d))$ decision tree that on input $(f, \{I_y\}, \{g_y\})$ as well as $y\in [2M]$, outputs a $\calP$ instance (that we call $I'_y$), and a mapping that given any valid answer $ans$ of $I'_y$, finds a valid sequence $(ans_0, ans_1, \dots, ans_{d-1})$. We compute each $\{y_k\}_{k\in [d+1]}$ as above and define $g'_y(ans) := y_0$.
	
	This finishes the description of our reduction from $\rPHP(\calP)_{M\to (1+\eps)M}$ to $\rPHP(\calP)_{M\to 2M}$; it is easy to see that it has depth $\gamma(O(\eps^{-1}))$. Now, given a valid solution $(y', ans')$ for $(f', \{I'_y\}, \{g'_y\})$, we can compute a valid solution $(y, ans)$ for $(f, \{I_y\}, \{g_y\})$ as follows. First, since $ans'$ is a valid solution for $I'_{y'}$, we can unpack $ans'$ to obtain a sequence $(ans_0, ans_1, \dots, ans_{d-1})$. Then we define each $\{y_k\}_{k\in [d+1]}$ as above (starting with $y_d = y'$). Also, for every $k\in[d+1]$, define $f_{\ge k} := f_{d-1} \circ \dots \circ f_k$, then $f_{\ge k}$ is a purported surjection from $[M+k\eps M]$ to $[2M]$. (As special cases, $f_{\ge d}:[2M]\to [2M]$ is the identity function and $f_{\ge 0} = f'$.) Since $(y', ans')$ is a valid solution, we know that $f'(g'_{y'}(ans')) \ne y'$, which translates to $f_{\ge 0}(y_0) \ne y_d$. Since $f_{\ge d}(y_d) = y_d$, there is some integer $k\in [d]$ such that $f_{\ge k}(y_k) \ne y_d$ but $f_{\ge k+1}(y_{k+1}) = y_d$. We argue that $(y_{k+1}, ans_k)$ is a valid solution for $(f, \{I_y\}, \{g_y\})$:
	\begin{itemize}
		\item First, it must be the case that $y_{k+1} < (1+\eps)M$. If $y_{k+1} \ge (1+\eps)M$, then $y_k = y_{k+1} - \eps M \ge M$ and thus $f_k(y_k) = y_{k+1}$. It follows that
		\begin{equation}\label{eq: k vs k+1}
			y_d \ne f_{\ge k}(y_k) = f_{\ge k+1}(f_k(y_k)) = f_{\ge k+1}(y_{k+1}) = y_d,
		\end{equation}
		a contradiction.
		\item Since $(ans_0, ans_1, \dots, ans_{d-1})$ is a valid sequence, $ans_k$ is a valid solution for $I_{y_{k+1}}$.
		\item Finally, if $f(g(ans_k)) = f(y_k) = y_{k+1}$, then \eqref{eq: k vs k+1} holds, which is a contradiction. Therefore, it must be the case that $f(g_{y_{k+1}}(ans_k)) \ne y_{k+1}$ and thus $(y_{k+1}, ans_k)$ is a valid solution for $(f, \{I_y\}, \{g_y\})$.\qedhere
	\end{itemize}
\end{proof}

\begin{lemma}\label{lemma: stretch reduction 2 to large}
	Let $N\ge 2M$. Suppose that \autoref{assumption: P is closed under Turing reductions} holds for $\calP$. There is a depth-$\gamma(O(\log \frac{N}{M}))$ decision tree reduction from $\rPHP(\calP)_{M \to 2M}$ to $\rPHP(\calP)_{M\to N}$.
\end{lemma}
\begin{proof}
	The proof is similar to that of \autoref{lemma: stretch reduction 1+eps to 2}. Without loss of generality, we may assume that $d := \log\frac{N}{M}$ is an integer (i.e., $N/M$ is a power of $2$). Let $(f, \{I_y\}, \{g_y\})$ be an instance of $\rPHP(\calP)_{M\to 2M}$ and we want to reduce it to an instance $(f', \{I'_y\}, \{g'_y\})$ of $\rPHP(\calP)_{M\to N}$. Recall that:
	\begin{itemize}
		\item $f:[M] \to [2M]$ is the purported ``surjection''.
		\item For every $y\in[2M]$, $I_y$ is a $\calP$ instance where every possible answer $ans$ of $I_y$ is labelled with an integer $g(ans) \in [M]$.
		\item The goal is to find some $y\in[2M]$ and a solution $ans$ of $I_y$ such that $f(g_y(ans))\ne y$.
	\end{itemize}

	For every integer $k\in [d]$, we put $2^k$ copies of the instance $(f, \{I_y\}, \{g_y\})$ in parallel and obtain the instance $(f_k, \{(I_k)_y\}, \{g_{k, y}\}$ of $\rPHP(\calP)_{(2^k M) \to (2^{k+1} M)}$. More precisely:
	\begin{itemize}[align=left]
		\item [{\underline{(Definition of $f_k$)}}] Let $x\in [2^k M]$ be the input, and let $x = x_0\cdot M + x_1$ where $x_0\in [2^k]$ and $x_1\in [M]$. We define $f_k(x) := x_0\cdot 2M + f(x_1)$.
		\item [{\underline{(Definition of $(I_k)_y$ and $g_{k, y}$)}}] Let $y\in [2^{k+1}M]$ be the input, and let $y = y_0\cdot 2M + y_1$ where $y_0\in [2^k]$ and $y_1 \in [2M]$. We define $(I_k)_y := I_{y_1}$, and for every $ans$ that is a possible solution of $(I_k)_y = I_{y_1}$, define $g_{k, y}(ans) := y_0\cdot M + g_y(ans)$.
	\end{itemize}

	The mapping from $(f, \{I_y\}, \{g_y\})$ to $(f_k, \{(I_k)_y\}, \{g_{k, y}\})$ can be computed by a depth-$1$ decision tree. Given a valid solution $(y, ans)$ for $(f_k, \{(I_k)_y\}, \{g_{k, y}\})$, we write $y = y_0\cdot 2M + y_1$ where $y_0 \in [2^k]$ and $y_1\in [2M]$. Since $ans$ is a solution of $(I_k)_y = I_{y_1}$ and 
	\[y_0\cdot 2M + y_1 = y \ne f_k(g_{k, y}(ans)) = y_0\cdot 2M + f(g(ans)) \implies f(g_y(ans)) \ne y_1,\]
	it follows that $(y_1, ans)$ is also a valid solution for $(f, \{I_y\}, \{g_y\})$. Therefore, there is a depth-$1$ decision tree reduction from $\rPHP(\calP)_{M\to 2M}$ to $\rPHP(\calP)_{(2^kM)\to (2^{k+1}M)}$.
	
	Now, we compose the instances $(f_k, \{(I_k)_y\}, g_{k, y})$ for every $k\in [d]$ to obtain the instance $(f', \{I'_y\}, \{g'_y\})$. In particular, the ``surjection'' $f':[M] \to [2^dM]$ is defined as $f' := f_{d-1} \circ f_{d-2} \circ \dots \circ f_0$.
	
	To define $I'_y$ and $g'_y$, consider the Turing machine with oracle access to $\calP$ that, on input $y\in [2^dM]$, operates as follows. Let $y_d := y$. For each $k$ from $d-1$ to $0$, the machine queries $\calP$ to obtain a valid answer $ans_k$ for $(I_k)_{y_{k+1}}$, and then sets $y_k := g_{k, y_{k+1}}(ans_k)$. Finally, the machine outputs the number $y_0 \in [M]$.
	
	We define a total search problem $\Lhis$ based on the computational history of this machine. The input of $\Lhis$ consists of $M, N, (f, \{I_y\}, \{g_y\})$, as well as some $y\in[2^dM]$; note that given these inputs, one can define the $\rPHP(\calP)_{(2^kM)\to (2^{k+1}M)}$ instances $(f_k, \{(I_k)_y\}, \{g_{k, y}\})$ as before. The output consists of a sequence $(ans_0, ans_1, \dots, ans_{d-1})$. Denoting $y_d = y$ and $y_k = g_{k, y_{k+1}}(ans_k)$ for every $k\in [d]$, accept the output if for every $k\in[d]$, $ans_k$ is a valid solution for $(I_k)_{y_{k+1}}$; otherwise reject the output.
	
	Clearly, the above Turing machine itself is a time-$O(d)$ Turing reduction from $\Lhis$ to $\calP$. Since $\calP$ is closed under Turing reductions, there is also a depth-$\gamma(O(d))$ mapping reduction from $\Lhis$ to $\calP$. Therefore, there is a depth-$\gamma(O(d))$ decision tree that on input $(f, \{I_y\}, \{g_y\})$ as well as $y\in[2^dM]$, outputs a $\calP$ instance (that we call $I'_y$), and a mapping that given any valid answer $ans$ of $I'_y$, finds a valid sequence $(ans_0, ans_1, \dots, ans_{d-1})$. We define $g'_y(ans) := g_{0, y}(ans_0)$.

	This finishes the description of our reduction from $\rPHP(\calP)_{M\to 2M}$ to $\rPHP(\calP)_{M\to N}$; it is easy to see that it has depth $\gamma(O(\log \frac{N}{M}))$. Now, given a valid solution $(y', ans')$ for $(f', \{I'_y\}, \{g'_y\})$, we can compute a valid solution $(y, ans)$ for $(f, \{I_y\}, \{g_y\})$ as follows. First, since $ans'$ is a valid solution for $I'_{y'}$, we can unpack $ans'$ to obtain a sequence $ans_0, ans_1, \dots, ans_{d-1}$. Let $y_d = y'$ and $y_k = g_{k, y_{k+1}}(ans_k)$ for every $k$ from $d-1$ downto $0$, then $f'(y_0) \ne y_d$. For every $k\in \{0, 1, \dots, d\}$, let $f_{\ge k} := f_{d-1} \circ f_{d-2} \circ \dots \circ f_k$; notice that $f_{\ge k}$ is a purported surjection from $[2^kM]$ to $[2^dM]$. (Note that as special cases, $f_{\ge 0} = f'$ and $f_{\ge d}:[2^dM]\to[2^dM]$ is the identity function.) Since $f_{\ge 0}(y_0) \ne y_d$ but $f_{\ge d}(y_d) = y_d$, there is an integer $k\in [d]$ such that $f_{\ge k}(y_k) \ne y_d$ but $f_{\ge (k+1)}(y_{k+1}) = y_d$. We claim that $(y_{k+1}, ans_k)$ is a valid solution to the instance $(f_k, \{(I_k)_y\}, \{g_{k, y}\})$.
	\begin{itemize}
		\item Since $(ans_0, ans_1, \dots, ans_{d-1})$ is a valid solution of $\Lhis$ on input $(f, \{I_y\}, \{g_y\}, y')$, $ans_k$ is a valid solution for $(I_k)_{y_{k+1}}$.
		\item Suppose $f_k(g_{k, y_{k+1}}(ans_k)) = y_{k+1}$, then $f_{\ge k}(y_k) = f_{\ge (k+1)}(f_k(g_{k, y_{k+1}}(ans_k))) = f_{\ge (k+1)}(y_{k+1})$. However, the RHS is equal to $y_d$ while the LHS is not equal to $y_d$. Therefore it must be the case that $f_k(g_{k, y_{k+1}}(ans_k)) \ne y_{k+1}$.
	\end{itemize}

	It follows that given a valid solution for $(f', \{I'_y\}, \{g'_y\})$, one can always find some $k$ and a valid solution for $(f_k, \{(I_k)_y\}, \{g_{k, y}\})$. That is, there is a depth-$\gamma(O(\log \frac{N}{M}))$ reduction from solving $\rPHP(\calP)_{M\to N}$ to solving one of $\{\rPHP(\calP)_{(2^kM)\to (2^{k+1}M)}\}_{k\in [d]}$. Composing this with the aforementioned depth-$1$ reduction from $\rPHP(\calP)_{M\to 2M}$ to $\rPHP(\calP)_{(2^kM) \to (2^{k+1}M)}$ completes our reduction.
\end{proof}

\section{Comparing \texorpdfstring{$\Refuter$}{Refuter} with \texorpdfstring{$\WrongProof(\Res)$}{WrongProof(Res)}}
\label{sec: discuss wrong proof}

We discuss the similarities and differences between the refuter problems and the $\WrongProof$ problem. %
We first recall the formal definition of $\WrongProof(\Res)$ \cite{BeckmannB17, GoldbergP18}:

\begin{mdframed}[hidealllines=true,backgroundcolor=gray!10]
    \begin{center}
        \textbf{Problem $\WrongProof(\Res)$}
    \end{center}
    
    \underline{Input:} A CNF $F$ with $n$ variables and $k$ clauses; a purported resolution refutation $\Pi$ for $F$ represented as $C_{-k}, \dots, C_{-1}, C_0, C_1, \dots, C_{L-1}$, where $C_{-k}, \dots, C_{-1}$ are axioms of $F$, $C_{L-1} = \bot$, and $L = 2^{n^{\Omega(1)}}$; and a purported satisfying assignment $\alpha \in \{0, 1\}^n$.
    
    \underline{Output:} A number $i\in [L]$ such that $C_i$ is obtained by an invalid resolution derivation, or a number $-k \le j\le -1$ such that $\alpha$ does not satisfy $C_j$.
\end{mdframed}

At first glance, the $\Refuter$ problem looks similar to the $\WrongProof$ problem. First, both problems take as input a purported (but not correct) resolution proof. Second, both are looking for an invalid derivation as a solution. Moreover, when we consider the resolution proof system (and consider refuting \emph{width} lower bounds), both $\WrongProof$ and $\Refuter$ are $\PLS$-complete.

However, we think that they are fundamentally different. One primary difference is the reason of totality: When introduced to a (non-promise) $\TFNP$ problem, the initial inquiry ought to be: \emph{why is the problem total?} The totality of $\WrongProof(\Res)$ follows from the \emph{reflection principle} for resolution \cite{Pudlak20, BFI23}, i.e., it is impossible to derive $\bot$ from a satisfiable CNF. The same reasoning holds for every sound proof system, regardless of their power. However, the totality of $\Refuter$ is far from trivial: \emph{They rely on non-trivially proven width or size lower bounds}.

Furthermore, for comparison with $\Refuter$, we include a proof that $\WrongProof(\Res)$ is $\PLS$-complete (this is a folklore result, see e.g., \cite{BFI23}). The proof is seemingly similar to that of \autoref{lemma: width refuter lower bound} and \autoref{lemma: width refuter upper bound}, but there are crucial differences. For example, the reduction from $\WrongProof(\Res)$ to $\PLS$ is uniform, since the totality of $\WrongProof(\Res)$ relies on simpler reasoning. In contrast, the uniform $\PLS$-membership of $\Refuter(w(F\vdash_{\Res} \bot))$ crucially relies on nice properties of the family of CNFs (e.g., $\EPHP$), and it is possible that for some families, the refuter problem cannot be uniformly reduced to $\PLS$ at all. This demonstrates another difference between $\WrongProof$ and $\Refuter$.

\begin{lemma}
    $\WrongProof(\Res)$ is in $\PLS$.
\end{lemma}
\begin{proof}
    Let $(C_{-k}, \dots, C_{-1}, C_0, \dots, C_{L-1})$ be a purported resolution refutation of a CNF $F$, and $\alpha$ be a purported satisfying assignment of $F$. We will reduce this $\WrongProof(\Res)$ instance to an instance $S:\{-k, \dots, L-1\} \to \{-k, \dots, L-1\}$ of reversed $\Iter$.

    It would be convenient to think of a clause $C_i$ as ``active'' if $C_i(\alpha) = 0$. An invalid derivation in the resolution refutation corresponds to an edge from an active node to an inactive node. For every $i\in \{-k, \dots, L-1\}$, if $C_i(\alpha) = 1$ (i.e., $C_i$ is inactive), then we define $S(i) = i$. Otherwise, if $i < 0$ (i.e., $C_i$ is an axiom \emph{not} satisfied by $\alpha$), then we define $S(i) = 0$, making $i$ a solution since $S(i) > i$. Otherwise, suppose $C_i$ is derived from $C_j$ (i.e., $C_i$ is a \emph{weakening} of $C_j$, or $C_i$ is \emph{resolved} from $C_j$ and some other $C_k$), then we define $S(i) = j$. If $i$ is a solution for the reversed $\Iter$ instance $S$, then either $j < i$ or $j$ is inactive (which means $S(j) = j$), and in either case $i$ is a valid solution for $\WrongProof(\Res)$.
\end{proof}
\begin{mdframed}[hidealllines=true,backgroundcolor=gray!10,skipabove=0em,skipbelow=-0.4em,innertopmargin=0]
	\small
\begin{remark}
    The proof above is easy, but one can see that the crucial components are 1) the resolution proof system is sound, and 2) a resolution proof is a ``DAG''-like structure. This proof strategy can potentially be easily extended to other proof systems with similar properties.
\end{remark}
\end{mdframed}

\begin{restatable}{lemma}{lemmawrongproofisplshard}\label{lemma: wrongproof is pls hard}
    $\WrongProof(\Res)$ is $\PLS$-hard.
\end{restatable}
\begin{proof}
    We will reduce any reversed $\Iter$ instance to an instance of $\WrongProof(\Res)$. The construction below is very similar to the proof of \autoref{lemma: width refuter lower bound}. In fact, all clauses and derivations in the construction of the proof of \autoref{lemma: width refuter lower bound} are sound except for the solutions of the given reversed $\Iter$ instance.

    Let $F$ be any satisfiable CNF with $k$ clauses $C_{-k}, \dots, C_{-1}$ and $\alpha$ be any satisfying assignment of $F$. Without loss of generality assume there are two clauses $C_{-2}$ and $C_{-1}$ that we can apply a valid resolution step and call the resolved clause $D$. Let $S:[L] \to [L]$ be an instance of reversed $\Iter$ where $S(L) < L$. We construct a purported resolution refutation $\Pi = (C_{-k}, \dots, C_{-1}, C_0, \dots, C_{L-1})$ as follows: \begin{itemize}
        \item For every $i$ such that $S(i) = i$, we let $C_i := D$ to be \emph{resolved} from $C_{-2}$ and $C_{-1}$.
        \item For every $i$ that is a solution for $S$, let $C_i := \bot$ be a \emph{weakening} from an axiom (say $C_{-k}$). Note that this weakening step is invalid and $C_i$ becomes a solution for the $\WrongProof(\Res)$ instance.
        \item Finally, for every $i$ such that $S(i) < i$ and $S(S(i)) < S(i)$, let $C_i := \bot$ be a \emph{weakening} of $C_{S(i)}$. Note that $C_{S(i)}$ is also $\bot$, hence this is a valid derivation.
    \end{itemize}
    It is easy to see that the invalid derivations in $\Pi$ correspond exactly to the solutions of $S$.
\end{proof}

\section{Prover-Delayer Games, \texorpdfstring{$\PLS$}{PLS}, and the 
Proof of \autoref{lem: res_ref to pls_ref}}\label{app:res_pls}

In \autoref{app:res_pls:pls2res}, we provide a self-contained description of the transformation from a $\PLS$ formulation to a low-width resolution proof using \emph{Prover-Delayer} game, along with several properties of this transformation that are useful when proving \autoref{lem: res_ref to pls_ref}. We then prove \autoref{lem: res_ref to pls_ref} in \autoref{app:res_pls:fullproof}.

\subsection{From \texorpdfstring{$\PLS$}{PLS} to Resolution using Prover-Delayer Game}\label{app:res_pls:pls2res}

Introduced by Pudl{\'{a}}k~\cite{Pudlak00ProofsAsGames}, the \emph{Prover-Delayer game} provides an elegant characterization of resolution width. There are two players in the game, the \emph{Prover} (she) and the \emph{Delayer} (he). Fixing an unsatisfiable CNF formula $F$, and let $x = (x_1, \ldots, x_n)$ be the variables in $F$. At first, the Prover's memory is empty. Then, in each step, she can either
\begin{itemize}
    \item \emph{query} the Delayer for the value of a certain variable, and add that value to her memory;
    \item \emph{forget} the value of a certain variable stored in her memory; or
    \item \emph{output} a clause of $F$ that is falsified by the partial assignment stored in her memory, which means she wins the game.
\end{itemize} 

We assume the Delayer also has access to Prover's memory. If the Prover queries a variable that is currently in its memory, then the Delayer's answer must be consistent with the memory; otherwise, his answer could be arbitrary. Note that if the Prover queries a variable, forgets it, and queries it again, the Delayer is allowed to answer different values to these two queries of the same variable.

Of course, Prover can always win the game by querying all variables without forgetting any of them. However, for the connection with resolution width, her goal is to win the game with the minimum memory size, where the memory size is the maximum number of variables she remembered during the whole execution of the game. The Delayer is \emph{adversarial} to Prover's goal, i.e., wants her to spend as much memory as possible.

The following theorem shows that the minimum resolution width of an unsatisfiable CNF is characterized by the minimum Prover memory in the corresponding Prover-Delayer game.
\begin{theorem}[\cite{Pudlak00ProofsAsGames}]
    For any unsatisfiable CNF formula $F$, there exists a width-$w$ resolution refutation of $F$ if and only if there is a winning strategy for Prover using memory size $w$ in the Prover-Delayer game for $F$.
\end{theorem}

In this section, we prove the ``if'' direction in the previous lemma and highlight some nice properties of the obtained low-width resolution proof that will be helpful when proving \autoref{lem: res_ref to pls_ref}.

\paragraph{Making Prover's strategy uniform.} Note that both Prover's and Delayer's strategies could be quite non-uniform in the Prover-Delayer game model described above. Here, we twist the model a little bit by allowing the Prover to explicitly store several \emph{state registers} in its memory, besides a partial assignment of variables. These internal state registers are also counted in the memory size of her strategy. Later, in the proof of \autoref{lem: res_ref to pls_ref}, it is more convenient to describe a uniform Prover's strategy with state registers.%

\paragraph{From $\PLS$ to Prover-Delayer game.}

    A $\PLS$ formulation of a search problem $\SearchCNF(F_n)$ is a decision tree reduction $(f_i,g_o)_{i, o \in M}$ from $\SearchCNF(F_n)$ to $\Iter_M$. Let $x = (x_1, \ldots, x_n)$ be the variables in $F_n$, then we have $S(v) \coloneqq f_v(x)$, where $S$ is the successor function in the $\Iter_M$ instance reduced from $\SearchCNF(F_n)$. %

    \begin{lemma}[Folklore]\label{lem: pls to prover}
        Given a $\PLS$ formulation $(f,g)_{i, o \in M}$ of depth $d$ for $\SearchCNF(F_n)$, there exists a Prover's strategy of memory size $O(d + \log M)$ for $F_n$.
    \end{lemma}
    
    \begin{proof}
        We say the Prover queries a decision tree $T$ if she evaluates $T(x)$, and stores the queried variables in her memory in each step. Now we describe the Prover's strategy in what follows.
    
        \begin{description}
            \item[1.] The Prover starts from the node $0$ of the $\Iter_M$ instance and queries the decision tree $f_0$. If $f_0(x) = 0$, then $0$ is a valid solution for the $\Iter_M$ instance, hence she can query $g_0(x)$ to obtain a falsified clause in $F_n$. Otherwise, we say that the Prover is currently at node $v = f_0$ and previously visited node $0$.
            \item[2.] Assume the Prover is at node $v \in [M]$, and the previous node she visited is $u$. She queries the decision tree $f_v$ and obtains the next node $w = f_v(x)$.
            \item[3a.] If $w \leq v$, then she has found a solution of the $\Iter_M$ instance. In particular, if $w < v$ then the solution is $v$; if $w = v$ then the solution is $u$. Note that all the variables queried by $f_u, f_v, g_u$ are still in her memory. Therefore, the clause $F$ returned by $g_v(x)$ (if $w < v$) or $g_u(x)$ (if $w = v$) must be falsified by the variables in her memory.
            \item[3b.] If $w > v$, then the Prover forgets all the variables that are queried in $f_u$ and not queried in $f_v$. She then updates the current node as $w$ and the previous node as $v$, and loops back to \textbf{Step 2}.
        \end{description}
    
        The Prover's strategy will always end, since the index of $v$ increases in every step. The Prover needs to remember at most $3d$ variables at any time, and $O(d + \log M)$ bits to remember the current state to execute this strategy.
    \end{proof}

    By further examining the proof of \autoref{lem: pls to prover}, we obtain several properties that are useful for the proof of \autoref{lem: res_ref to pls_ref} later.

    \begin{observation}\label{obs: uniform prover}
        In \autoref{lem: pls to prover}, the Prover's strategy can be implemented in a uniform manner if the $\PLS$ formulation $(f,g)$ is given via oracle access.
    \end{observation}

    \begin{lemma}\label{lem: memory encoding}
        In \autoref{lem: pls to prover}, there exists an efficient binary encoding of Prover's memory, such that:
        \begin{enumerate}
            \item The encoding has bit-length $\poly(d,\log M)$.
            \item It is computationally efficient to transform Prover's memory into an encoding and vice versa. 
            \item The encoding of the Prover's memory is lexicographically increasing as the Prover's strategy proceeds.
            \item We say an encoding is \emph{invalid} if it is in the wrong format, or Prover's internal state registers are inconsistent with the partial assignment of variables w.r.t.~$(f,g)$. There is an efficient uniform algorithm for checking whether an encoding is invalid given oracle access to $(f,g)$. 
        \end{enumerate}
    \end{lemma}

    \begin{proof}
        The Prover's internal state registers should store the index of the current node, the previous node, and its current location in the decision tree it is querying. This part is lexicographically increasing since we always have $w > v$ in {\bf Step 3b}. Our encoding consists of these internal state registers followed by the partial assignment of variables.     

        It is trivial to construct such an encoding scheme satisfying conditions 1,2,3, and easy to check the correctness of its format. To check the validity of an encoding, we can query the (at most $3$) decision tree paths corresponding to the Prover's internal state registers, and check whether they are consistent with the partial assignment.
    \end{proof}

\paragraph{From Prover-Delayer games to resolution proofs}

    \begin{lemma}[\cite{Pudlak00ProofsAsGames}]\label{lem: prover to res proof}
        Given a Prover's strategy of memory cost $w$ for $F_n$, there exists a width-$w$ resolution proof refuting $F_n$.
    \end{lemma}

    \begin{proof}
        Without loss of generality, we assume the Prover will not query a variable that is already in her memory.
    
        We simulate the Prover's strategy. Initially, there is no variable stored in her memory, and it corresponds to the empty clause $\bot$ at the end of the resolution proof. We then generate the resolution proof recursively by maintaining the Prover's memory and the current node of the resolution proof.

        In each step, let $\rho$ be the current partial assignment of variables stored in the Prover's memory. Define $C(\rho)$ to be the only clause that is falsified by $\rho$, using only the variables that are set in $\rho$. For example, if $\rho$ is $\{x_1 = 1, x_2 = 0\}$, the $C(\rho) = \neg x_1 \vee x_2$.
        The procedure will guarantee that $C(\rho)$ is the clause of the current node in the resolution proof.
        
        The Prover has three possible actions given $\rho$ and its internal state register:
        \begin{description}
            \item[FORGET] If the Prover decides to forget $x_i$, then we generate a new node $C'$ with clause $C(\rho_{-i})$, where $\rho_{-i}$ is the partial assignment by \emph{forgetting} the value of $x_i$ from $\rho$. We mark that the current node is derived by a \emph{weakening} step from node $C'$. We then update the Prover's memory and continue our process at node $C'$ recursively.
            \item[QUERY] If the Prover queries $x_i$, then we generate two new nodes $C^0, C^1$ with clauses $C(\rho) \vee x_i$ and $C(\rho) \vee \neg x_i$ respectively. We mark that the current node is derived by a \emph{resolution} step from node $C^0$ and $C^1$. We first recursively proceed to $C^0$ by updating Prover's memory with $x_i = 0$, and then proceed to $C^1$ with $x_i=1$.
            \item[OUTPUT] If the Prover outputs a falsified clause $D$, it must be the case that $D$ is a sub-clause of $C(\rho)$. 
            If $D$ is equal to  $C(\rho)$, then we simply stop; otherwise, we add a new node for the clause $D$ and add one or more intermediate \emph{weakening} steps towards the current node.
        \end{description}

        It is easy to verify that during the process, $C(\rho)$ is always the clause of the current node in the resolution proof. This process will stop since the Prover will stop, and its correctness is guaranteed by the Prover's correctness.
        Finally, note that the largest clause ever generated in the resolution proof is upper bounded by the memory size of the Prover. 
    \end{proof}

\subsection{Proof of \autoref{lem: res_ref to pls_ref}}\label{app:res_pls:fullproof}
    
\LemmaResReftoPLSRef*

\begin{proof}
     Let $\calF = \{F_n\}_{n\in \N}$. Suppose we are given a purported depth-$d$ reduction $(f,g)$ from $\SearchCNF(F_n)$ to $\Iter_M$, i.e., $(f,g)$ is a $\PLS$ formulation for $\SearchCNF(F_n)$. We now use $(f, g)$ to construct a purported resolution refutation $(C_0, \ldots, C_{L-1})$ for $F_n$ with width $w_0 = \poly(d, \log M)$, while satisfying the following two conditions.

    \begin{enumerate}
        \item Given any index $i \in [L]$, the $i$-th node $C_i$ can be calculated in $\polylog(n)$ queries to $(f,g)$ uniformly.
        \item If the $i$-th node is invalid, one can recover a pair $(\pAss, o^*)$ that refutes $(f,g)$ uniformly given $i$.
    \end{enumerate}

    As a high-level plan, we first apply the procedures described in \autoref{lem: pls to prover} which transforms the $\PLS$ formulation $(f,g)$ to a Prover's strategy of $O(d + \log M)$ memory for the Prover-Delayer game. We then use the procedure described in \autoref{lem: prover to res proof} to convert such a strategy into a width-$O(d + \log M)$ resolution proof for $F_n$.

    We now specify the details in these two steps to make sure the two conditions are met. The first condition can be achieved by letting an index $i \in [L]$ to be an encoding of the Prover's memory in a single step, as described in \autoref{lem: memory encoding}.\footnote{Note that the encoding described in \autoref{lem: memory encoding} is lexicographically increasing, but we can easily make it lexicographically decreasing by flipping all the bits.}
    
    By \autoref{obs: uniform prover}, given any valid index (encoding), one can calculate the next action of the Prover using $O(d)$ number of queries to $(f,g)$. Then, for all three actions $\{\textbf{FORGET}, \textbf{QUERY}, \textbf{OUTPUT}\}$, we can also calculate the indices of the one or two previous nodes in the resolution proof uniformly. 
    For any invalid index (encoding), we pad a trivially correct resolution node using axioms from the beginning.

    To see the second condition above, note that the procedure described in \autoref{lem: prover to res proof} generates an invalid node of the resolution proof only when the Prover is taking an \textbf{OUTPUT} step. That is, when the Prover outputs a falsified clause $D$, $D$ might not be a sub-clause of $C(\rho)$, so the \emph{weakening} steps from $D$ to $C(\rho)$ will be wrong. This happens because a solution $o$ of the $\Iter_M$ instance is found by the Prover, but querying $g_o$ does not lead to a clause $D$ in $F_n$ that is falsified by the current partial assignment $\pAss$ in her memory. 
    
    Note that the partial assignment $\pAss$ and $o^*$ used to refute $(f,g)$ can be recovered from the index (encoding) of the invalid node. We then complement the partial assignment $\pAss$ to $\pAss'$ by setting all unassigned variables that appear in $D$ to satisfy clause $D$. By definition, the pair $(\pAss', o)$ is a valid solution to refute the reduction $(f,g)$. 
    
    Finally, if $\calF$ is a uniform family of formulas, then the whole reduction is also uniform.
\end{proof}

\end{document}